\documentclass[aps,prx,twocolumn,footinbib,floatfix,superscriptaddress]{revtex4-2}
\usepackage{graphicx}
\usepackage{indentfirst}
\usepackage{braket}
\usepackage{float}
\usepackage{amsmath}
\usepackage{physics}
\usepackage{amssymb}
\usepackage{CJK}
\usepackage{esint}
\usepackage{color}
\usepackage[T1]{fontenc}
\usepackage{subfigure}
\usepackage{amsfonts}
\usepackage{footmisc}
\usepackage{scrextend}
\usepackage{multirow}
\usepackage[outdir=./]{epstopdf}

\usepackage[english]{babel}
\usepackage{url}
\usepackage{comment}
\usepackage{array}
\usepackage{subfiles}
\usepackage{tikz}
\usetikzlibrary{shapes,arrows}
\usepackage{mathrsfs}
\usepackage[mathscr]{eucal}

\usepackage[hyperfootnotes=false]{hyperref}
\usepackage{cleveref}
\definecolor{darkblue}{rgb}{0,0,0.5}
\hypersetup{
colorlinks=true,
linkcolor=black,
filecolor=blue,
citecolor=darkblue,  
urlcolor=black,
}

\newtheorem{theorem}{Theorem}

\newtheorem{conjecture}{Conjecture}
\newtheorem{corollary}{Corollary}

\newtheorem{lemma}{Lemma}

\newtheorem{proposition}{Proposition}

\newenvironment{proof}[1][Proof]{\noindent\textbf{#1.} }{\ \rule{0.5em}{0.5em}}

\def\be{\begin{equation}}
\def\ee{\end{equation}}
\def\ba{\begin{eqnarray}}
\def\ea{\end{eqnarray}}
\def\bal{\begin{equation}\begin{aligned}}
\def\eal{\end{aligned}\end{equation}}

\def\bp{\begin{pmatrix}}
\def\ep{\end{pmatrix}}

\usepackage{bm}

\newcommand{\calA}{{\cal A}}
\newcommand{\calB}{{\cal B}}
\newcommand{\calC}{{\cal C}}

\newcommand{\calE}{{\cal E}}

\newcommand{\calI}{{\cal I}}
\newcommand{\calL}{{\cal L}}

\newcommand{\calN}{{\cal N}}

\newcommand{\calH}{{\cal H}}

\newcommand{\calR}{{\cal R}}

\newcommand{\calU}{{\cal U}}
\newcommand{\calV}{{\cal V}}

\newcommand{\1}{^{(1)}}

\newcommand{\state}[1]{\ketbra{#1}{#1}}

\usepackage[normalem]{ulem}

\begin{document}
\title{Computable limits of optical multiple-access communications}
\author{Haowei Shi}
\affiliation{
James C. Wyant College of Optical Sciences, University of Arizona, Tucson, Arizona 85721, USA
}

\author{Quntao Zhuang}
\email{zhuangquntao@email.arizona.edu}
\affiliation{
Department of Electrical and Computer Engineering, University of Arizona, Tucson, Arizona 85721, USA
}
\affiliation{
James C. Wyant College of Optical Sciences, University of Arizona, Tucson, Arizona 85721, USA
}

\begin{abstract}
Communication rates over quantum channels can be boosted by entanglement, via superadditivity phenomena or entanglement assistance. Superadditivity refers to the capacity improvement from entangling inputs across multiple channel uses. Nevertheless, when unlimited entanglement assistance is available, the entanglement between channel uses becomes unnecessary---the entanglement-assisted (EA) capacity of a single-sender and single-receiver channel is additive. We generalize the additivity of EA capacity to general multiple-access channels (MACs) for the total communication rate. Furthermore, for optical communication modelled as phase-insensitive bosonic Gaussian MACs, we prove that the optimal total rate is achieved by Gaussian entanglement and therefore can be efficiently evaluated. To benchmark entanglement's advantage, we propose computable outer bounds for the capacity region without entanglement assistance. Finally, we formulate an EA version of minimum entropy conjecture, which leads to the additivity of the capacity region of phase-insensitive bosonic Gaussian MACs if it is true. The computable limits confirm entanglement's boosts in optical multiple-access communications. 

\end{abstract}


\maketitle

\section{Introduction}

Modern communication systems often transmit information via optical channels, such as fibers or free-space links. With the development of quantum source generation and detection, quantum effects, such as the uncertainty principle, squeezing and entanglement, become relevant in characterizing the capacity of optical channels. For example, entanglement can sometimes lead to the superadditivity phenomena~\cite{hastings2009superadditivity,smith2008quantum,czekaj2009purely,czekaj2011,zhu2018superadditivity,zhu2017,zhu2018superadditivity}, where the communication capacity is increased due to entanglement between inputs among multiple channel uses. While an increased capacity is beneficial in practice, the evaluation of channel capacities becomes challenging as it requires an optimization of the joint input over an arbitrary number of channel uses. For this reason, it took great efforts~\cite{GiovannettiV2014} to solve the classical capacity of optical communication since first conjectured~\cite{giovannetti2004minimum}.

Entanglement can also be pre-shared as assistance to boost the communication rates, as theoretically proposed~\cite{bennett2002entanglement,shi2020practical} and experimentally demonstrated~\cite{hao2021entanglement}. Moreover, unlimited entanglement assistance (EA) in each channel use rules out the need to entangle inputs among different channel uses in the capacity evaluation of a single-sender and single-receiver channel---their entanglement-assisted (EA)~\footnote{We use `EA' for both `entanglement assistance' and `entanglement-assisted'.} classical capacity is additive, avoiding the superadditivity conundrum~\cite{adami1997neumann}. Utilizing Gaussian extremality~\cite{Holevo2001,wolf2006} upon the additivity, for optical communication links modeled as single-mode bosonic Gaussian channels (BGCs), one can analytically solve the EA classical capacity and prove rigorous advantages over the unassisted case.

The development of an optical network further complicates the story, as multiple senders and receivers in a network can communicate simultaneously, for example over broadcast channels and multiple-access channels (MACs). In this case, the communication limit is characterized by a trade-off capacity region among multiple users. The classical capacity region of a general broadcast channel is still an open problem with or without EA~\cite{cover1972broadcast,dupuis2010father}, except for special cases~\cite{cover1972broadcast,gallager1974capacity,yard2011quantum,savov2015classical}. 
For bosonic optical broadcast channels, although the capacity formula is known in the pure-loss case~\cite{yard2011quantum,savov2015classical}, the classical capacity region is still subject to the multi-mode `entropy photon number inequality' conjecture~\cite{de2016gaussian,de2017gaussian,de2018gaussian,de2019new}; for phase-insensitive bosonic Gaussian MACs (BGMACs) that model optical networks, the classical capacity formula with and without EA are known~\cite{cover1999elements,winter2001capacity,yard2008capacity,hsieh2008entanglement,shi2021entanglement}, whereas the capacity evaluation remains an open question~\cite{yen2005multiple,shi2021entanglement}. Indeed, unlike the single-sender and single-receiver case, the problem of additivity is complicated by the trade-off between the different senders.

In this paper, we provide computable limits for optical multiple-access communication over BGMACs with and without EA. For the unassisted case, we obtain general applicable outer bounds for the capacity region. For the EA case, despite the entire capacity region being still elusive, we show that the total communication rate is maximized by Gaussian entanglement and can therefore be efficiently evaluated. We prove so via extending the additivity of total rate known for the two-sender case~\cite{hsieh2008entanglement} to the general case. While the sum rate provides an outer bound of the EA capacity region,
the Gaussian optimality in total rate inspires us to examine the Gaussian-state capacity region of a class of interference BGMACs, which is practically relevant and includes all BGMACs considered in the literature~\cite{yen2005multiple,xu2013sequential,wilde2012explicitMAC,shi2021entanglement}. With two or three senders, we numerically find that two-mode squeezed vacuum (TMSV) states to be the optimal Gaussian entanglement.
Comparing with the rate limits of the unassisted case, we identify a large advantage in the noisy and lossy limit. We also extend the above results to interference BGMACs with memory effects~\cite{kretschmann2005quantum,lupo2010capacities}. Finally, we propose an EA version of minimum entropy conjecture, which leads to the additivity of the capacity region of phase-insensitive bosonic Gaussian MACs if it is true.



\section{Bosonic Gaussian multiple-access channels}
\label{sec:general_definition}

As shown in Fig.~\ref{fig:EAMACschematic_main}(a), to communicate via a quantum MAC, the $s\ge 1$ senders encode the messages on quantum systems $A_1,\ldots,A_s$ and send them to a common receiver; a quantum MAC is described by a completely positive and trace-preserving map $\calN_{A\to B}$, where we denote the quantum systems of the $s$ senders together as a composite system $A=\otimes_{k=1}^s A_k$ and the output as $B$. To facilitate our analyses, we introduce the Stinespring unitary dilation of the MAC acting jointly on the environment system $E^\prime$ and the input $A$. The common receiver measures the quantum system $B$ to decode all messages. The decoding suffers from not only the environment noise but also the interference between the multiple senders.  As illustrated in Fig.~\ref{fig:EAMACschematic_main}(b), an EA MAC communication protocol incorporates an EA system $A_k^\prime$ pre-shared to the receiver for each sender $k$ to improve the communication rate, so that the state in each pair $A_kA'_k$ is pure. Similar to the composite system $A$, we denote $A'=\otimes_{k=1}^s A'_k$ as the composite system of the EA ancilla. Note that the single sender ($s=1$) case of a MAC reduces to a point-to-point quantum channel.

\begin{figure}[t]
    \centering
    \includegraphics[width=0.5\textwidth]{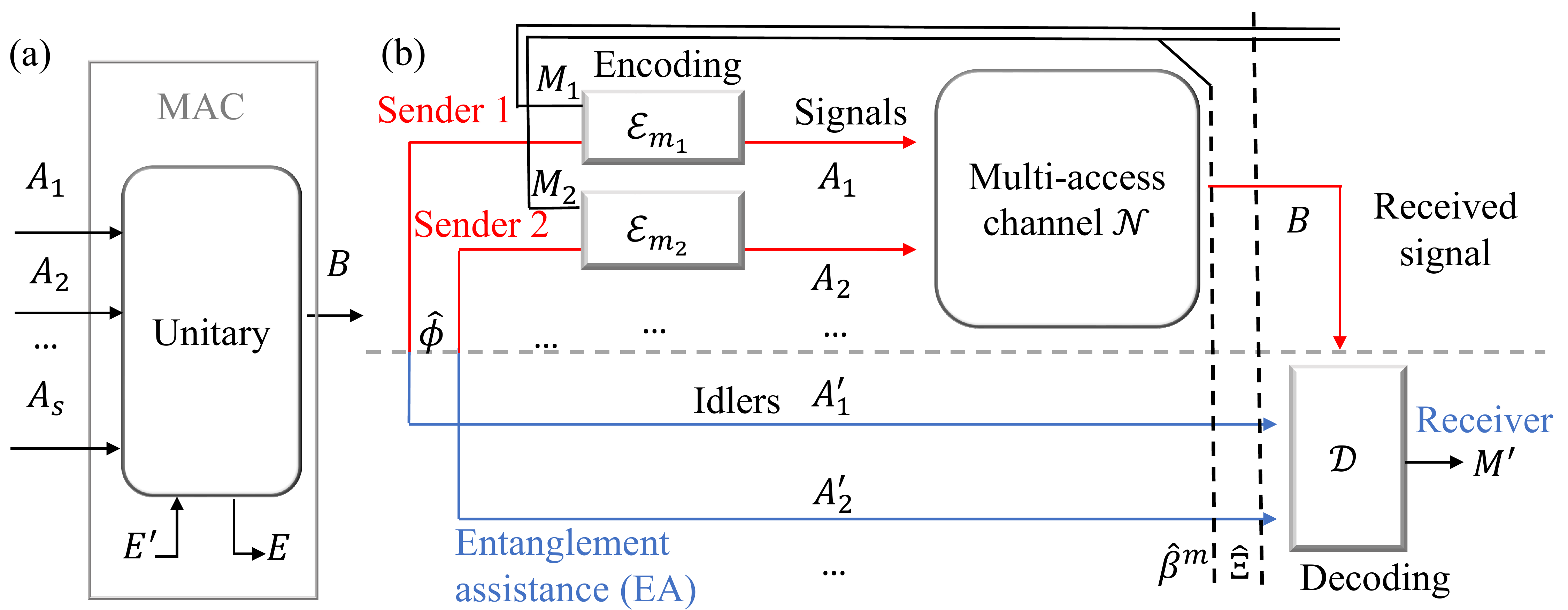}
    \caption{Schematic of (a) a general multiple-access channel and (b) entanglement-assisted communication through quantum multiple-access channel. See also Fig. 2 in Ref.~\cite{shi2021entanglement}.}
    \label{fig:EAMACschematic_main}
\end{figure}

In optical communication, the relevant channels of interest are phase-insensitive BGCs, where the channel unitary dilation and the environment state are Gaussian~\cite{Weedbrook_2012}. Formally, phase-insensitivity invokes a symmetry about the phase rotation unitary $\hat{R}(\theta)\equiv \exp\left(-i\theta \hat{a}^\dagger \hat{a}\right)$ with $\theta\in[0,2\pi)$, acting on a mode with the annilation operator $\hat{a}$. For a phase-insensitive BGMAC, a phase rotation $\hat R(\theta)$ on the output mode can be equivalently implemented on the input by properly choosing a phase rotation $\hat{R}_k(\theta_k)$ on each input $A_k$, where $\theta_k=(-1)^{\delta_k} \theta$ with $\delta_k\in\{0,1\}$; namely
\be 
\calN_{A\to B}\left[\hat{R}\left(\bm \theta\right) \hat{\rho}_{A} \hat{R}^\dagger\left(\bm \theta\right)\right]=\hat{R}(\theta)\calN_{A\to B}\left(\hat{\rho}_{A}\right)\hat{R}^\dagger(\theta),
\label{eq:phase_invariance}
\ee 
where the overall phase rotation $\hat{R}(\bm \theta)=\otimes_{k=1}^s \hat{R}_k\left(\theta_k\right)$. When the symmetries across different users are homogeneous---$\delta_k$'s are equal, we call the BGMAC global covariant ($\delta_k=0$) or global contravariant ($\delta_k=1$).

For the single-sender case of $s=1$, the BGMAC reduces to a point-to-point BGC represented by the Bogoliubov transforms on input mode $\hat{a}_A$ as
\be 
\hat{a}_B= w \left(\left(1-\delta\right)\hat{a}_{A}+\delta\hat{a}_{A}^\dagger \right)+
\hat \xi,
\label{aB_BGC}
\ee 
where $\delta=0,1$ indicates that the channel is covariant or contravariant and the noise term $\hat \xi$ combines two vacuum modes~\cite{caruso2008multi,caruso2011optimal} 
\be 
\hat \xi=u_1\hat{a}_{E^\prime,1}+u_2\hat{a}_{E^\prime,2}^\dagger\,.
\label{eq:xi}
\ee 
The mean photon number of the noise term is $\expval{\hat \xi^\dagger \hat \xi}=u_2^2$, where we have chosen $u_1,u_2$ to be real due to the phase degree-of-freedom of vacuum modes. The canonical commutation relation $[\hat{a}_B,\hat{a}_B^\dagger]=1$ requires that
$ 
(1-2\delta) |w|^2+(u_1^2-u_2^2)=1\,.
$
On vacuum inputs, the channel $\calN_{A\to B}$ produces a thermal state with the mean photon number $N_{\rm B}\equiv u_2^2+|w|^2 \delta$, which corresponds to `dark photon counts'. For a bona fide BGC, $N_{\rm B}\ge \max\{-1+|w|^2(1-\delta),|w|^2\delta\}$, where the channel is quantum-limited when the equality holds.

The phase-insensitive point-to-point BGCs contain four classes~\cite{GiovannettiV2014}.
The covariant family consists of three classes: the thermal-loss channels $\calL^{|w|^2,N_{\rm B}}$ with transmissivity $|w|^2<1$; the AWGN channels $\calE^{N_{\rm B}}$ when $|w|=1$; and the amplifier channels $\calA^{|w|^2,N_{\rm B}}$ with gain $|w|^2>1$. The contravariant family includes the conjugate amplifiers $\tilde{\calA}^{|w|^2+1,N_{\rm B}}$ with gain $|w|^2+1$. The $N_{\rm B}$ factor in the superscript denotes the mean dark photon counts of the channel.

Similar to the single-sender case, we can descirbe the multi-sender ($s\ge 2$) phase-insensitive BGMAC via the Bogoliobov transform
\bal
\hat{a}_B=&\left[\sum_{k=1}^s w_k \left(\left(1-\delta_k\right)\hat{a}_{A_k}+\delta_k\hat{a}_{A_k}^\dagger \right)\right]+
\hat \xi,
\label{eq:BGMAC_phaseinsensitive}
\eal
where the weights $\{w_k\}$'s are in general complex. 
It is natural to define the weight vector $\bm w=[w_1,w_2,\ldots,w_s]^T$, with its norm $|\bm w|=\sqrt{\sum_{k=1}^s|w_k|^2}$. The noise term $\hat\xi$ can still be reduced to the linear combination of two vacuum environment modes as in Eq.~\eqref{eq:xi} (see Appendix~\ref{app:noise_reduction}).

The canonical commutation relation $[\hat{a}_B,\hat{a}_B^\dagger]=1$ requires that
$
\left[\sum_{k=1}^s (1-2\delta_k) |w_k|^2\right]+(u_1^2-u_2^2)=1\,.
$
Similar to the single-sender case, on vacuum inputs, the MAC $\calN_{A\to B}$ produces a thermal state with the mean photon number $N_{\rm B}\equiv u_2^2+\sum_{k=1}^s|w_k|^2\delta_k$. Also, a bona fide BGMAC requires $N_{\rm B}\ge  \max\{-1+\sum_{k=1}^s |w_k|^2(1-\delta_k),\sum_{k=1}^s|w_k|^2\delta_k\}$. One can prove (see Appendix~\ref{app:outerbd}) that global covariant or contravariant BGMACs can always be constructed via interfering the users' input on a beamsplitter and then passing the mixed mode through a phase-insensitive point-to-point BGC. We will address this interference BGMAC in Sec.~\ref{sec:interference_def}.

In a BGMAC, as the Hilbert space of the quantum system is infinite-dimensional---an arbitrary number of photons can occupy a single mode due to the bosonic nature of light. To model a realistic communication scenario, we will consider an energy constraint on the mean photon number (brightness) of each sender's signal mode,
\be 
\expval{\hat{a}_{A_k}^\dagger \hat{a}_{A_k}}\le N_{S,k}, 1\le k \le s.
\label{eq:energy_constraint}
\ee 

In the above, we have focused on the single-mode BGMACs where a single party in the communication protocol has access to a single bosonic mode per channel use. 
We address the multi-mode generalization in Appendix~\ref{sec:Methods_multimode}.

\section{Communication over a MAC}
In this section, we introduce the communication rate region over a general MAC, with and without entanglement as the assistance.

\subsection{Unassisted communications}
To establish a multiple-access communication link to transmit classical messages, the $k$th sender independently sends quantum states according to a random variable $x_k$ generated from the type of codebook $p_{X_k}(x_k)$. Denote the classical register as $X_k$, the overall classical-quantum state describing the channel is
\be 
\hat \Xi=\sum_{x_1,\ldots, x_s}\left(\otimes_{k=1}^s p_{\small X_k}(x_k)\ketbra{x_k}_{X_k}\right) \otimes \hat\beta^{x_1\ldots x_s}_{B},
\ee
where $\hat \beta^x$ is the received state conditioned on the code word $x$.
Ref.~\cite{winter2001capacity} has derived the unassisted capacity region of an arbitrary quantum MAC to be the convex closure of all non-negative rate tuples $(R_1,R_2,\ldots,R_s)$ satisfying 
\be 
\sum_{k\in J}R_k\le I(X[J];B|X[J^c])_{\hat\Xi}\,,\forall J\subseteq \{1,2,\ldots,s\}\,,
\label{eq:C_def}
\ee
for some input distribution $\prod_{k=1}^s p_{X_k}(x_k)$.

Rate regions under specific input setups of an unassisted BGMAC have been investigated in Ref.~\cite{yen2005multiple} for an interference MAC with additive white Gaussian noises (AWGNs). Despite that an exact solution of the capacity region is still elusive, Ref.~\cite{yen2005multiple} has made a substantial progress by the combination of a coherent-state rate region and an outer-bound region. In Section \ref{sec:main_results}, we will extend the results for an arbitrary phase-insensitive BGMAC. 

\subsection{Entanglement-assisted communications}
The EA communication scenario is depicted in Fig.~\ref{fig:EAMACschematic_main}(b).
We consider the entanglement to be pairwise between each sender and the receiver such that the overall quantum state
$ 
\hat{\phi}_{AA^\prime}=\otimes_{k=1}^s \hat{\phi}_{A_k A_k^\prime}
$
is in a product form.
For convenience, we also define a quantum state after the channel but without the encoding,
\be 
\hat{\rho}_{BA^\prime}= \left[\calN_{A\to B}  \otimes \calI\right] \left(\hat{\phi}_{AA^\prime}\right).
\label{eq:rho_main}
\ee 
We can also consider the unitary dilation $\calU_\calN$ of the channel and write out the overall pure state output
\be 
\hat{\rho}_{EBA^\prime}= \calU_\calN \otimes \calI_{A^\prime}
\left[\hat{\phi}_{AA^\prime}\otimes \state{0}_{E^\prime}\right].
\ee 
Here $\ket{0}$ can be any pure state. In the case of phase-insensitive BGMACs, it is a multi-mode vacuum state. We will frequently utilize the purity in our analyses.

The EA capacity region of an $s$-sender MAC is conjectured in Ref.~\cite{hsieh2008entanglement} and proven in Ref.~\cite{shi2021entanglement}. Formally, the capacity is given by the regularized union 
\be
\calC_{\rm E}(\calN)=\overline{\bigcup_{n=1}^\infty \frac{1}{n}  \calC_{\rm E}^{(1)}(\calN^{\otimes n})},
\label{eq:CE_regularize_general}
\ee 
where the ``one-shot'' capacity region $\calC_{\rm E}^{(1)}(\calN)$ is the convex hull of the union of ``one-shot, one-encoding'' regions
\be 
\calC_{\rm E}^{(1)}(\calN)={\rm Conv}\left[\bigcup_{\hat{\phi}} \tilde \calC_{\rm E}(\calN,\hat{\phi} )\right].
\label{eq:CE_union_state}
\ee 
The ``one-shot, one-encoding''  rate region $\tilde \calC_{\rm E}(\calN,\hat{\phi} )$ for the 2s-partite pure product state $\hat{\phi}_{AA^\prime}=\otimes_{k=1}^s\hat{\phi}_{A_kA_k^\prime}$ over $AA^\prime$, is the set of rates $(R_1,\cdots,R_s)$ satisfying the following $2^s$ inequalities
\bal
\sum_{k\in J}R_k&\leq  I(A'[J];B|A'[J^c])_{\hat{\rho}}\equiv  \tilde C_{{\rm E},J}(\calN,\hat\phi), \forall J,
\label{eq:CeMAC_qinfo}
\,\eal 
where $A'[J]=\otimes_{k\in J}A'_k$ denotes a composite system of a subset $J$ of the $s$ senders and the conditional quantum mutual information is evaluated over the output state 
\be 
\hat{\rho}_{BA^\prime}=\calN_{A\to B}\otimes\calI_{A^\prime}(\hat{\phi}_{AA^\prime})\,.
\label{eq:rho}
\ee 
Here the quantum conditional mutual information between quantum systems $X$ and $Z$ conditioned on $Y$ is defined as
$ 
I(X;Z|Y)_{\hat{\alpha}}=S(XY)_{\hat{\alpha}}+S(YZ)_{\hat{\alpha}}-S(XYZ)_{\hat{\alpha}}-S(Y)_{\hat{\alpha}},
$ 
where $S(X)_{\hat{\alpha}}=S(\hat{\alpha}_X)=-\tr \left(\hat{\alpha}_X\log_2 \hat{\alpha}_X\right)$ is the von Neumann entropy. The independence of the coding among all $s$ senders ensures the independence among $A'_1,\ldots,A'_s$, and thereby a convenient relation
\be 
I(A'[J];B|A'[J^c])_{\hat{\rho}}=I(A'[J];BA'[J^c])_{\hat{\rho}},
\label{eq:MACindependence}
\ee
which will simplify the analyses in Sec.~\ref{sec:main_results}.

In Ref.~\cite{shi2021entanglement}, the rate-region of the common entanglement source TMSV has been evaluated for a thermal-loss MAC. However, the exact capacity is unknown due to the difficulty brought by: (a) the regularization in Eq.~\eqref{eq:CE_regularize_general}, which requires the entangled input over an arbitrary number of channel uses; (b) the union over all states in Eq.~\eqref{eq:CE_union_state}, which in general live in the infinite dimensional Hilbert space and are energy constrained; (c) the $2^s$ exponentially large number of inequalities in \eqref{eq:CeMAC_qinfo} that describe the boundary. 

Instead of obtaining the exact solution, we can also place bounds on the capacity region. To do that, we define the capacity of the partial communication rate in sender set $J$,
\be 
C_{{\rm E},J}(\calN)=\lim_{n\to\infty} \frac{1}{n} \max_{\hat\phi}\tilde C_{{\rm E},J}(\calN^{\otimes n},\hat\phi)
\label{eq:CE_J}
\ee 
where the one-shot one-state capacity of the partial rate for $n$ channel uses $\tilde C_{{\rm E},J}(\calN^{\otimes n},\hat\phi)$ is defined in Eq.~\eqref{eq:CeMAC_qinfo}.
The individual capacities form an outer bound for the capacity region, however, the outer bound is sometimes unattainable when the optimum states are different among different sets of senders.
Our main theorems focus on the total rate within the universal set $U\equiv\{1,2,\ldots,s\}$, where Eq.~\eqref{eq:CE_J} reduces to
\be 
\tilde C_{{\rm E},U}(\calN,\hat\phi)=I(A';B)_{\hat{\rho}}\,,
 \label{eq:CeMAC_qinfo_overall}
\ee
when combining with Eq.~\eqref{eq:CeMAC_qinfo}.

\section{Main results}

\label{sec:main_results}

In this section, we propose computable limits for both the unassisted and the EA communication protocols. For the unassisted protocols, we solve the coherent-state capacity region for general phase-insensitive BGMACs and an outer bound for the noisy phase-insensitive BGMACs. For the EA protocols, we prove for the total rate that the EA classical capacity of a general quantum MAC is additive, and that the ultimate capacity of a single-mode phase-insensitive BGMAC is achieved by an $s$-partite product of TMSV states, which immediately leads to a capacity theorem for BGCs when $s=1$. Similar to the unassisted case, an outer bound for the EA case is derived. For a multi-mode phase-insensitive BGMAC, we find that a general Gaussian state to be sufficient to achieve the optimum total rate.

\subsection{Limits of unassisted communications}
\label{sec:unassisted}

The coherent-state region constrains the input state to be an ensemble of coherent states, which is widely used to model laser-communication scenarios. The input state conditioned on code words $x_1,\ldots,x_s$ is a product coherent state
\be 
\hat\sigma^{x_1\ldots x_s}=\otimes_{k=1}^s\ketbra{x_k}{x_k}_{A_k}\,.
\ee
For communication over a general phase-insensitive BGMAC defined by Eq.~\eqref{eq:BGMAC_phaseinsensitive},  given the coefficients $\bm w$, the thermal dark count noise $N_{\rm B}$, and the signal brightness per sender $\{N_{{\rm S},k}\}_{k=1}^s$, the optimal input distribution that maximizes Eq.~\eqref{eq:C_def} is the circularly-symmetric Gaussian distribution $p_{X_k}(x_k)\propto\exp{-{|x_k|^2}/{N_{{\rm S},k}}},\,1\le k\le s$ over complex displacement $x_k$.
Therefore, the coherent-state capacity region can be obtained as the region inside the boundaries
\begin{align} 
\sum_{k\in J}R_k&\le  C_{{\rm coh}, J} \equiv g\left(\sum_{k\in J}|w_k|^2 N_{{\rm S},k}\!+\!\! N_{\rm B}\!\right)
-g\left(N_{\rm B}\!\right),
\label{eq:C_coh}
\end{align}
where $J\subseteq \{1,2,\ldots,s\}$ is an arbitrary sender set and the function $g(x)=(x+1)\log(x+1)-x\log (x)$. 

Meanwhile, we can obtain an outer-bound region via individual upper bounds per sender and an upper bound for the total rate (see Appendix~\ref{app:outerbd}). We solve the outer bounds for channels satisfying either one of the following conditions,
\begin{align}
&\mbox{(a): } N_{\rm B}\ge \max\{|\bm w|^2-1,0\}+\sum_{k=1}^s|w_k|^2\delta_k; 
\label{condition_a}
\\
&\mbox{(b): } N_{\rm B}\ge |\bm w|^2+\sum_{k=1}^s|w_k|^2(1-\delta_k)
\label{condition_b}
\end{align}
The former fits better a BGMAC with more covariant components such that $\sum_{k=1}^s |w_k|^2(1-2\delta_k)\ge 0$, otherwise the latter would provide a tighter bound. 

In the noisy scenario (a), the individual bound of the $k$th sender is evaluated as if $N_{{\rm S},k}$ photons from the $k$th sender alone travel through a covariant ($\delta_k$=0) or contravariant ($\delta_k$=1) BGC which amplifies/attenuates the signal mode by $|\bm w|^2$ with dark photon count 
\be 
N_{\rm B}^k= N_{\rm B}+|\bm w|^2\delta_k-\sum_{\ell=1}^s|w_\ell|^2\delta_\ell\,.
\label{eq:def_Nbk}
\ee 
In this case, each individual rate $R_k$ satisfies an upper bound
\begin{align}
R_k\le g\left( |\bm w|^2  N_{{\rm S},k}+N_{\rm B}^k\right) -g\left(N_{\rm B}^k\right)\, .
\label{eq:C_outer}
\end{align} 
The bound on the total rate is evaluated as if $\sum_{k=1}^s |w_k|^2/|\bm w|^2(N_{{\rm S},k}+\delta_k)$ photons from all the $s$ senders travel through a covariant BGC of gain/loss $|\bm w|^2$ with the dark photon count $N_{\rm B}-\sum_{k=1}^s|w_k|^2\delta_k$.
For the total rate of all $s$ senders, by the bottleneck inequality (data processing inequality) we have the upper bound 
\be 
\sum_{k=1}^s R_k\le g\left(\sum_{k=1}^s |w_k|^2N_{{\rm S},k}+N_{\rm B}\right)-g\left(N_{\rm B}-\sum_{k=1}^s|w_k|^2\delta_k\right)\,.
\label{eq:C_outer_sum}
\ee
When the BGMAC is global covariant, coherent-state encoding achieves the optimal unassisted total rate, as Eq.~\eqref{eq:C_coh} coincides with the upper bound in Ineq.~\eqref{eq:C_outer_sum}. 

In scenario (b), following a similar derivation but switching the choices of the BGCs, we have the outer bound in this case 
\begin{align}
R_k&\le g\left( |\bm w|^2  N_{{\rm S},k}+N_{\rm B}^{k,c}\right) -g\left(N_{\rm B}^{k,c}\right)\,,
\label{eq:C_outer_contra}
\\
\sum_{k=1}^s R_k&\le g\left(\sum_{k=1}^s |w_k|^2 N_{{\rm S},k}+N_{\rm B}\right)\!\!
\nonumber
\\
&\ \ \ \ -\! g\!\left(\! N_{\rm B}\!-\!\!\sum_{k=1}^s|w_k|^2(1-\delta_k)\right)\,.
\label{eq:C_outer_sum_contra}
\end{align} 
where the dark photon counts per sender in this case are 
\be 
N_{\rm B}^{k,c}\equiv N_{\rm B}+|\bm w|^2(1-\delta_k)-\sum_{\ell=1}^s|w_\ell|^2(1-\delta_\ell)\,.
\label{eq:def_Nbkc}
\ee 
It is noteworthy that, an arbitrary global covariant or contravariant BGMAC always satisfies one of the noise conditions (a) or (b), therefore the corresponding outer bounds Eqs.~\eqref{eq:C_outer}~\eqref{eq:C_outer_sum} or Eq.~\eqref{eq:C_outer_contra}~\eqref{eq:C_outer_sum_contra} can always be applied, as will be discussed in Sec.~\ref{sec:interference_def}.

\subsection{Limits of EA communications}
\label{sec:ea} 

We prove the additivity of the total rate of EA communication over a general MAC as a generalization of the two-sender version in Ref.~\cite{hsieh2008entanglement}.
 
\begin{theorem}[additivity of the total rate]
\label{theorem: EA_MAC_main_additivity}
The EA classical capacity of the total communication rate over an $s$-sender MAC $\calN$ is additive,
\be 
C_{{\rm E},U}(\calN)=C_{{\rm E},U}^{(1)}(\calN).
\label{eq:CE_additivity}
\ee 
\end{theorem}

The proof of Theorem~\ref{theorem: EA_MAC_main_additivity} relies on the additivity of quantum mutual information in Eq.~\eqref{eq:CeMAC_qinfo_overall}. However, for the partial information rates $C_{{\rm E},J\neq U}(\calN)$, we note that a proof of the additivity is still elusive, due to the emergence of the EA systems $A'[J^c]$ in Eq.~\eqref{eq:CeMAC_qinfo} when $J^c$ is non-empty. In this case, $A'[J^c]$ is unfortunately indivisible into one subsystem per channel use when senders in $J^c$ apply joint coding among channel uses. Nevertheless, it is noteworthy that a conditional additivity for partial rate of $J$ holds under the constraint that senders in $J^c$ apply independent coding among channel uses. This reveals an intriguing phenomenon, that entangling $A[J]$ among channel uses does not improve the communication rate when $A[J^c]$ is uncorrelated, which will help our proof of Proposition \ref{proposition:gaugeinv_local}.

Using the additivity, we solve the ultimate capacity for the total rate of EA communication over a phase-insensitive BGMAC as summarized below.

\begin{theorem}[Total rate of BGMACs]
\label{theorem: EA_MAC_main_TMSV}
The EA classical capacity of total rate $C_{{\rm E},U}(\calN)$ over a single-mode phase-insensitive $s$-sender BGMAC $\calN$, under the energy constraint $\{ N_{{\rm S},k}\}_{k=1}^s$, is additive and achieved by an $s$-partite TMSV state, i.e.
\be  
C_{{\rm E},U}(\calN)=\tilde C_{{\rm E},U}(\calN,\otimes_{k=1}^s \ket{\zeta}_{A_kA_k'}(N_{{\rm S},k}) )\,.
\label{eq:CE_totalrate}
\ee 
\end{theorem} 
Here a TMSV between two modes $D,D^\prime$, with a mean photon number $N_{\rm S}$ per mode, is given by the wavefunction
\be
\hat{\zeta}_{D D^\prime}(N_{\rm S})\propto \sum_{n_k=0}^\infty (1+1/N_{\rm S})^{-n/2} \ket{n}_{D}\ket{n}_{D^\prime},
\ee 
where $\ket{n}$ is the Fock state defined by $\hat{a}^\dagger \hat{a}\ket{n}=n\ket{n}$.

Note that Theorem~\ref{theorem: EA_MAC_main_TMSV} also applies to general point-to-point phase-insensitive BGCs, leading to the following corollary.
\begin{corollary}[Capacity of BGCs]
The EA capacity of a single-mode point-to-point phase-insensitive bosonic Gaussian channel is achieved by a TMSV state. For a BGC described by Eq.~\eqref{aB_BGC}, the EA capacity $C_{\rm E}(\delta,N_{\rm S}, |w|^2,N_{\rm B})$ is specified as the following: For the covariant BGCs $(\delta=0)$, 
\be
C_{\rm E}(0,N_{\rm S}, |w|^2,N_{\rm B})=g(N_{\rm S})+g(N_{\rm S}^\prime)-g(A_+)-g(A_-),
\label{eq:CE_BGC}
\ee 
where
$A_\pm=(D-1\pm(N_{\rm S}^\prime-N_{\rm S}))/2$, $N_{\rm S}^\prime=|w|^2 N_{\rm S}+N_{\rm B}$ and $D=\sqrt{(N_{\rm S}+N_{\rm S}^\prime+1)^2-4|w|^2 N_{\rm S}(N_{\rm S}+1)}$; for the contravariant BGCs ($\delta=1$), 
\be
C_{\rm E}(1,N_{\rm S}, |w|^2,N_{\rm B})=g(N_{\rm S})+g(N_{\rm S}^{\prime } )-g(A_+^{\rm c})-g(A_-^{\rm c}),
\label{eq:CE_BGC_contra}
\ee 
where 
$A_\pm^{\rm c}=(\pm D^{\rm c}+N_{\rm S}^{\prime }+N_{\rm S})/2$, $N_{\rm S}^{\prime }=|w|^2 N_{\rm S}+N_{\rm B}$ and $D^{\rm c}=\sqrt{(N_{\rm S}^{\prime }-N_{\rm S})^2+4|w|^2 N_{\rm S}(N_{\rm S}+1)}$. 


\label{corollary:opt_singlesender}
\end{corollary}
This result not only generalizes the capacity of the lossy and noisy case in Ref.~\cite{Holevo2001}, but also provide a start point for obtaining the outer bound below.

Indeed, the derivation of the outer bound in the unassisted case also applies to the EA case. For channels satisfying Condition (a) of Eq.~\eqref{condition_a}
\begin{align}
&R_k\le 
C_{\rm E}(\delta_k,N_{{\rm S},k},|\bm w|^2,N_{\rm B}^k)\,
\label{eq:CE_outer}
\\
&\sum_{k=1}^s R_k\le 
\nonumber
\\
&C_{\rm E}\left(0,\sum_{k=1}^s \frac{|w_k|^2}{|\bm w|^2}(N_{{\rm S},k}+\delta_k),|\bm w|^2,N_{\rm B}-\sum_{k=1}^s|w_k|^2\delta_k\right)\,;
\label{eq:CE_outer_sum}
\end{align}
and for channels satisfying Condition (b) of Eq.~\eqref{condition_b}
\begin{align}
&R_k\le 
C_{\rm E}(\delta_k, N_{{\rm S},k},|\bm w|^2,N_{\rm B}^{k,c})\,,
\label{eq:CE_outer_contra}\\
&\sum_{k=1}^s R_k\le \nonumber\\
& C_{\rm E}\left(1,\sum_{k=1}^s \frac{|w_k|^2}{|\bm w|^2}(N_{{\rm S},k}+\delta_k),|\bm w|^2,N_{\rm B}-\sum_{k=1}^s|w_k|^2(1-\delta_k)\right)\,.
\label{eq:CE_outer_sum_contra}
\end{align}
Here the dark photon counts $N_{\rm B}^k, N_{\rm B}^{k,c}$ are defined in Eqs.~\eqref{eq:def_Nbk}~\eqref{eq:def_Nbkc}.
Similarly, our results are complete for all global covariant or global contravariant cases---interference BGMACs, as we will detail in Sec.~\ref{sec:interference_def}.


Finally, we generalize our results to the multi-mode case, where each sender and receiver can send multiple modes to the BGMAC. (See Appendix~\ref{sec:Methods_multimode} for a detailed definition)

\begin{theorem}[multi-mode BGMACs]
\label{theorem: EA_MAC_main_Gaussian}
The energy-constrained EA classical capacity of total rate $C_{{\rm E},U}(\calN)$ over a phase-insensitive $s$-sender BGMAC $\calN$ is additive and achieved by a zero-mean Gaussian state, i.e.
\be  
C_{{\rm E},U}(\calN)=  \max_{\hat\phi^G}\tilde C_{{\rm E},U}(\calN,\hat{\phi}^G ).
\label{CE_Gaussian_union}
\ee 
where the maximization is taken over all zero-mean Gaussian states $\hat{\phi}^G$ satisfying the energy constraints.
\end{theorem} 

A sketch of the proofs of the main results can be found in Appendix~\ref{sec:proof_sketch}, while more details are presented in later parts of the appendices.

\section{Example: interference BGMACs and memory BGMACs}

\begin{figure}[tbp]
    \centering
    \includegraphics[width=0.3\textwidth]{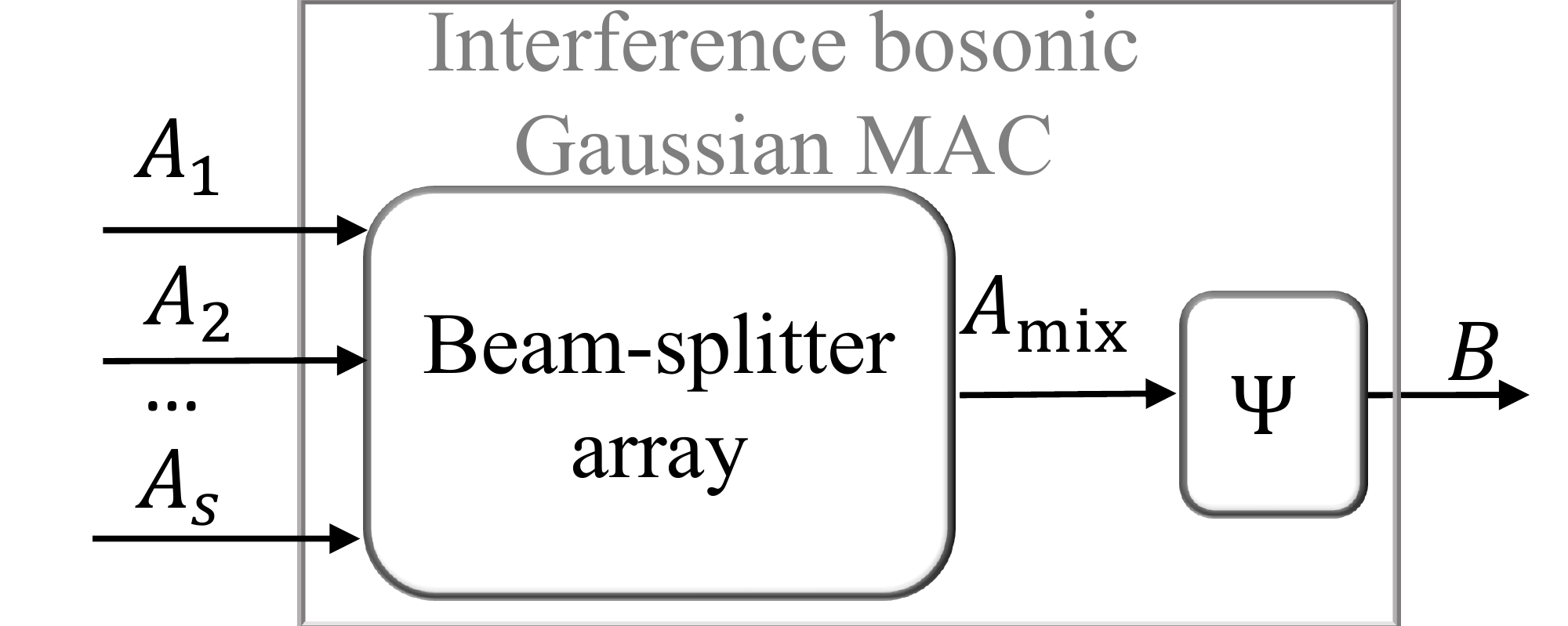}
    \caption{Schematic of interference bosonic Gaussian multiple-access channels. $\Psi$ is a single-mode BGC.}
    \label{fig:MACbosonic}
\end{figure}

So far we have derived the rate limits for both the unassisted and the EA communications over a phase-insensitive BGMAC. In this section, we evaluate the limits for interference BGMACs, and furthermore for its multi-mode version where memory effects are involved. We show that the product TMSV, with the proper additional passive linear operations accounting for any memory effects, achieves the total rate capacity. Compared with the coherent-state capacity region, the TMSV demonstrates a logarithmic EA advantage $\sim \log(1/N_{\rm S})$ in the overall signal brightness $N_{\rm S}=\sum_{k=1}^s N_{{\rm S},k}$, even when the memory effect is present.

\subsection{Interference-BGMAC}
\label{sec:interference_def}

Although our main theorems apply to general phase-insensitive BGMACs, in the numerical evaluation we will focus on the interference BGMACs as shown in Fig.~\ref{fig:MACbosonic}. For these channels, the input modes first interfere through a beamsplitter array then the mixed mode travels through a BGC. One can in general prove that as long as the BGMAC in Eq.~\eqref{eq:BGMAC_phaseinsensitive} is global covariant or contravariant ($\delta_k$'s are equal), it reduces to an interference BGMAC. Therefore, interference BGMACs form a fairly general class of BGMACs. Indeed, interference-BGMACs are relevant in applications, and all of the BGMACs considered in the literature are in this class~\cite{yen2005multiple,xu2013sequential,wilde2012explicitMAC,shi2021entanglement}.
Formally, the interference BGMAC $\calN$ is a concatenation 
\be 
\calN=\Psi\circ\calB
\label{eq:def_interf}
\ee
of an $s$-input-one-output beamsplitter $\calB$ and a single-mode BGC $\Psi$.
Upon the input modes $\hat{a}_{A_1}\cdots \hat{a}_{A_s}$ from the $s$ senders, the MAC $\calN$ first combines the modes through the beam-splitter $\calB$ to produce a mixture mode 
\be
\hat{a}_{A_{\rm mix}}=\sum_{k=1}^s \frac{w_k}{|\bm w|}\hat{a}_{A_k},
\label{eq:Amix}
\ee 
while all other ports of the beam-splitter array are discarded. Then the mixture mode goes through the single-mode BGC $\Psi$. For convenience, we introduce the power interference ratios $\{ \eta_k=|w_k|^2/|\bm w|^2\}_{k=1}^s$ that are normalized to unity. 
Interference BGMACs can be classified into four fundamental classes, depending on the four classes of the BGC $\Psi$ involved, as introduced in Sec.~\ref{sec:general_definition}.

For all the four classes of $\Psi$, we can respectively obtain the coherent-state capacity region and the outer bounds for the unassisted case.
Via the general result of Eq.~\eqref{eq:C_coh}, we have the coherent-state capacity region defined by
\bal  
C_{{\rm coh},J} =&g\left[\sum_{k\in J}|\bm w|^2 \eta_k N_{{\rm S},k}+N_{\rm B}\right]-g(N_{\rm B})\,,
\label{eq:C_coh_interference}
\eal 
for any $J\subseteq \{1,\ldots,U\}$.
Similarly, from Eqs.~\eqref{eq:C_outer}~\eqref{eq:C_outer_sum} for the covariant interference BGMACs and Eqs.~\eqref{eq:C_outer_contra}~\eqref{eq:C_outer_sum_contra} for the contravariant case, we have the outer bound
\begin{align} 
&R_k\le g[|\bm w|^2 N_{{\rm S},k}+N_{\rm B}]-g(N_{\rm B}), \, 1\le k\le s\,,
\label{eq:C_outer_interference}\\
&\sum_{k=1}^s R_k\le g\left(\sum_{k=1}^s \eta_k|\bm w|^2 N_{{\rm S},k}\!+\!N_{\rm B}\!\right)\!\!-\!g(N_{\rm B})\,.
\label{eq:C_outer_sum_interference}
\end{align}
The formulae for the two cases coincide in the form, owing to our definition using $\bm w$ and $N_{\rm B}$ that differs from Ref.~\cite{GiovannettiV2014}. The above agree with the results in Ref.~\cite{shi2021entanglement} for the $\delta=0$ case. The outer-bound region is a computation-friendly benchmark. Any EA protocol surpassing the outer-bound region possesses a provable advantage over all unassisted protocols.



\begin{figure}[t]
    \centering
    \includegraphics[width=0.45\textwidth]{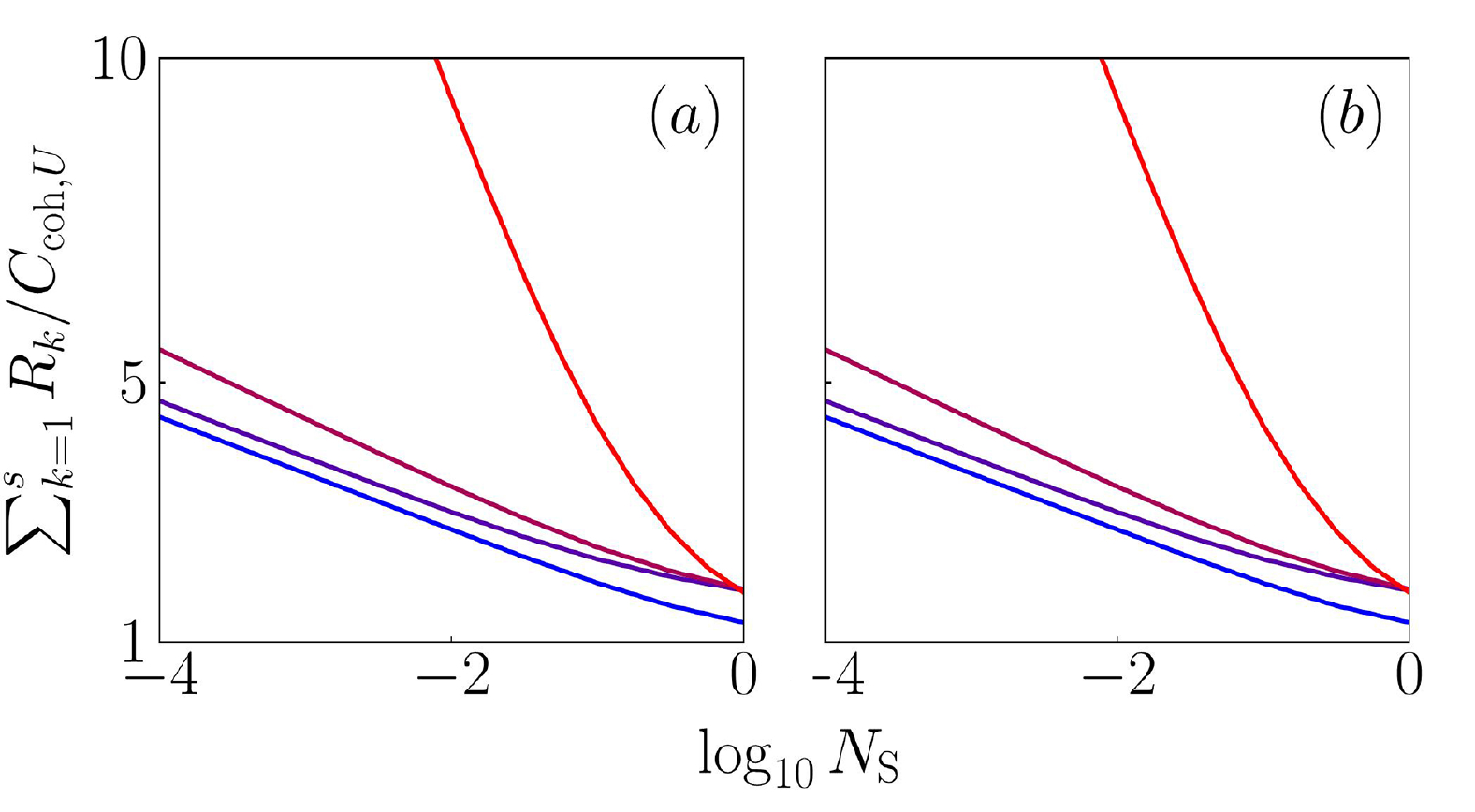}
    \caption{The logarithmic entanglement enhancement against signal brightness $N_{\rm S}$ in the total communication rate of an interference BGMAC $\calN$ defined in Eq.~\eqref{eq:def_interf}, over the unassisted coherent-state capacity $C_{{\rm coh}, U}$ Eq.~\eqref{eq:C_coh} which achieves the outer bound Eq.~\eqref{eq:C_outer} with $U=\{1,\ldots, s\}$.  Colored from blue to red for $\Psi$ being the thermal-loss channel, the AWGN channel, the amplifier channel, and the conjugate amplifier channel. (a) 2-sender case ($s=2$): signal brightness $N_{\rm S,1}=0.9N_{\rm S}, N_{\rm S,2}=0.1 N_{\rm S}$, interference ratio $\eta_1=0.9,\eta_2=0.1 $ (b) 3-sender case ($s=3$):  $N_{\rm S,1}=N_{\rm S}/2, N_{\rm S,2}= N_{\rm S}/3, N_{\rm S,3}= N_{\rm S}/6$, $\eta_1=1/2,\eta_2=1/3,\eta_3=1/6$. The gain/loss of the BGC component $\Psi$ is $|w|^2=0.1, 1, 1.1, 1.1$ with thermal noise $N_{\rm B}=0.1+\max\{(|w|^2-1)(1-\delta),0\}=0.1,0.1,0.2,0.1$ for the four cases respectively.
    \label{fig:EA_Rtot}
    }
\end{figure}

Similar to the unassisted case, we can also obtain the outer bound for the EA case from Eqs.~\eqref{eq:CE_outer}~\eqref{eq:CE_outer_sum} for the covariant cases, and Eqs.~\eqref{eq:CE_outer_contra}~\eqref{eq:CE_outer_sum_contra} for the contravariant case. Summarizing the results, we have
\begin{align}
R_k&\le C_{\rm E}(\delta,N_{{\rm S},k},|\bm w|^2,N_{\rm B}),\,1\le k\le s\,,
\label{eq:CE_outer_interference}\\
\sum_{k=1}^s R_k&\le C_{\rm E}(\delta,\sum_{k=1}^s \eta_kN_{{\rm S},k},|\bm w|^2,N_{\rm B})\,,
\label{eq:CE_outer_sum_interference}
\end{align}
where $\delta=0$ for the covariant case and $\delta=1$ for the contravariant case.
Although this bound is likely loose, it is easy to calculate since $C_{\rm E}$ is known explicitly in Eqs.~\eqref{eq:CE_BGC} and~\eqref{eq:CE_BGC_contra}. 

Besides the outer bounds, we can evaluate the ultimate capacity (of total rate) from Theorem~\ref{theorem: EA_MAC_main_TMSV}. We consider thermal-loss interference BGMACs as an example.
As shown in Fig.~\ref{fig:EA_Rtot}, we see a logarithmic EA advantage in the total communication rate, similar to the point-to-point communication in a thermal-loss channel~\cite{shi2020practical}. 
The EA advantage in (b) is slightly better than (a) because the brightness of the dominant sender (here sender 1) is lower, which enhances the EA advantage. 
To understand the influence of the beamsplitter ratio on the total rate, we consider the two-sender and three-sender cases in Fig.~\ref{fig:EA_Rtot_eta}. We find that the advantage's dependence on the beamsplitter ratio is weak, when the input energy is low. Therefore, the logarithmic EA advantage is robust for different BGMACs.

We also note that a gap emerges between the EA classical capacity and the bottleneck bound when the interference structure varies, as shown in Fig.~\ref{fig:EA_Rtot_eta}(a). Nevertheless, we find that the gap diminishes at least in the following two scenarios: (1) the BGMAC approaches a point-to-point BGC with $\eta\to 0$ or $1$; (2) the signal brightness distribution goes to homogeneous where the interference structure no longer influences the total rate. Since the former is described by a neat formula Eq.~\eqref{eq:C_outer_sum_interference}, it works as a fairly efficient benchmark in this regime. In conclusion, the infinite-fold EA advantage in point-to-point communication in noisy BGC extends to the multiple-access communication.

\begin{figure}[t]
    \centering
    \includegraphics[width=0.47\textwidth]{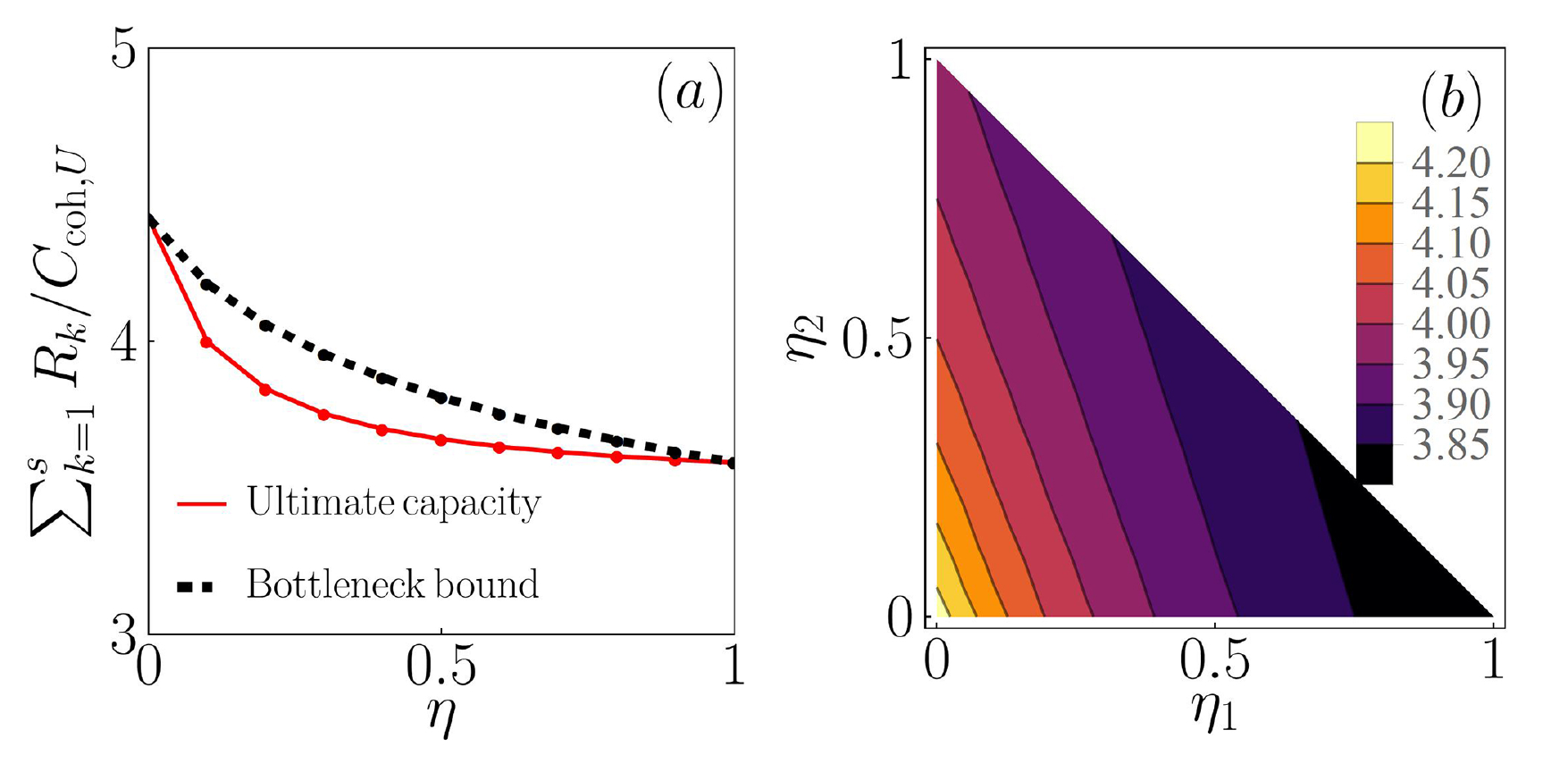}
    \caption{The total communication rate $R=\sum_{k=1}^sR_k$, normalized by the coherent-state capacity $C_{{\rm coh},U}$, $U=\{1,\ldots, s\}$, against the interference ratios of a thermal-loss interference BGMAC for (a) the 2-sender case ($s$=2) given the energy budget per sender $N_{S,1}=9\times 10^{-4}$, $N_{S,2}= 10^{-4}$; (b) the 3-sender case ($s$=3) given the energy budget per sender $N_{S,1}=1/2\times 10^{-3}$, $N_{S,2}=1/3\times 10^{-3}$, $N_{S,3}=1/6\times 10^{-3}$. The $\Psi$ component in each case is a thermal-loss BGC with $| w|^2=0.1,N_{\rm B}=0.1$. In (a) the EA bottleneck bound Eq.~\eqref{eq:CE_outer_sum_interference} is also presented for comparison. 
    \label{fig:EA_Rtot_eta} 
    }

\end{figure}


With the total rate understood, now we examine the one-shot capacity region with Gaussian encoding. Theorem~\ref{theorem: EA_MAC_main_TMSV} indicates that the capacity of total rate is achieved by a TMSV state, which is a special case of Gaussian states. In practice, Gaussian states are especially of interest since they are accessible with off-the-shelf devices. We evaluate the union of the rates over zero-mean Gaussian input states $\hat{\phi}_{AA^\prime}$. Due to the unitary degree of freedom in the EA $A^\prime$ that purifies the input system $A$, the only Gaussian state we need to consider is a product of squeezed TMSV states up to single-mode squeezing 
\be 
\ket{\phi\left(\bm r, \bm \theta\right)}_{AA^\prime}=
\otimes_{k=1}^s \hat{S}_{A_k}(r_k,\theta_k)\ket{\zeta(N_{S,k}^\prime)}_{A_kA^\prime_k},
\label{Gaussian_input}
\ee 
where the squeezing operator $\hat{S}(r,\theta)=\hat{R}(\theta)\hat{S}(r)$ is a concatenation of a single-mode squeezing operation with strength $r$ and a phase rotation of angle $\theta$. Mathematically,  $\hat{S}(r)=\exp\left(r\left(\hat{a}^2-\hat{a}^{\dagger 2}\right)/2\right)$, with $\hat{a}$ being the annihilation operation of the mode they act on.

Given the input in Eq.~\eqref{Gaussian_input}, the energy constraint of Eq.~\eqref{eq:energy_constraint} requires
\be 
N_{S,k}^\prime=\frac{4 N_{S,k} e^{2 r_k}+2 e^{2 r_k}-e^{4 r_k}-1}{2 \left(e^{4 r_k}+1\right)}\,.
\label{eq:Nsprime}
\ee 
The ovreall state $\hat{\phi}\left(\bm r, \bm \theta\right)$ is parameterized by two length-$s$ vectors $\bm \theta=[\theta_1,\cdots, \theta_s]^T$ and $\bm r=[r_1,\cdots, r_s]^T$. Therefore the one-shot Gaussian-state capacity region is given by the union
\be 
\calC^{(1)}_{\rm E,G}(\calN)= \bigcup_{\bm r, \bm \theta} \tilde \calC_{\rm E}\left(\calN,\hat{\phi}\left(\bm r, \bm \theta\right) \right),
\label{CG_formal}
\ee 
Due to the phase symmetry Eq.~\eqref{eq:phase_invariance}, the total number of parameters in the union is reduced to $2s-1$. This union can be numerically solved easily when $s$ is not too large.
Now we proceed to numerical solving the Gaussian capacity region of the two-sender case for the convenience of visualization.

\begin{figure}[t]
    \centering
    \includegraphics[width=0.5\textwidth]{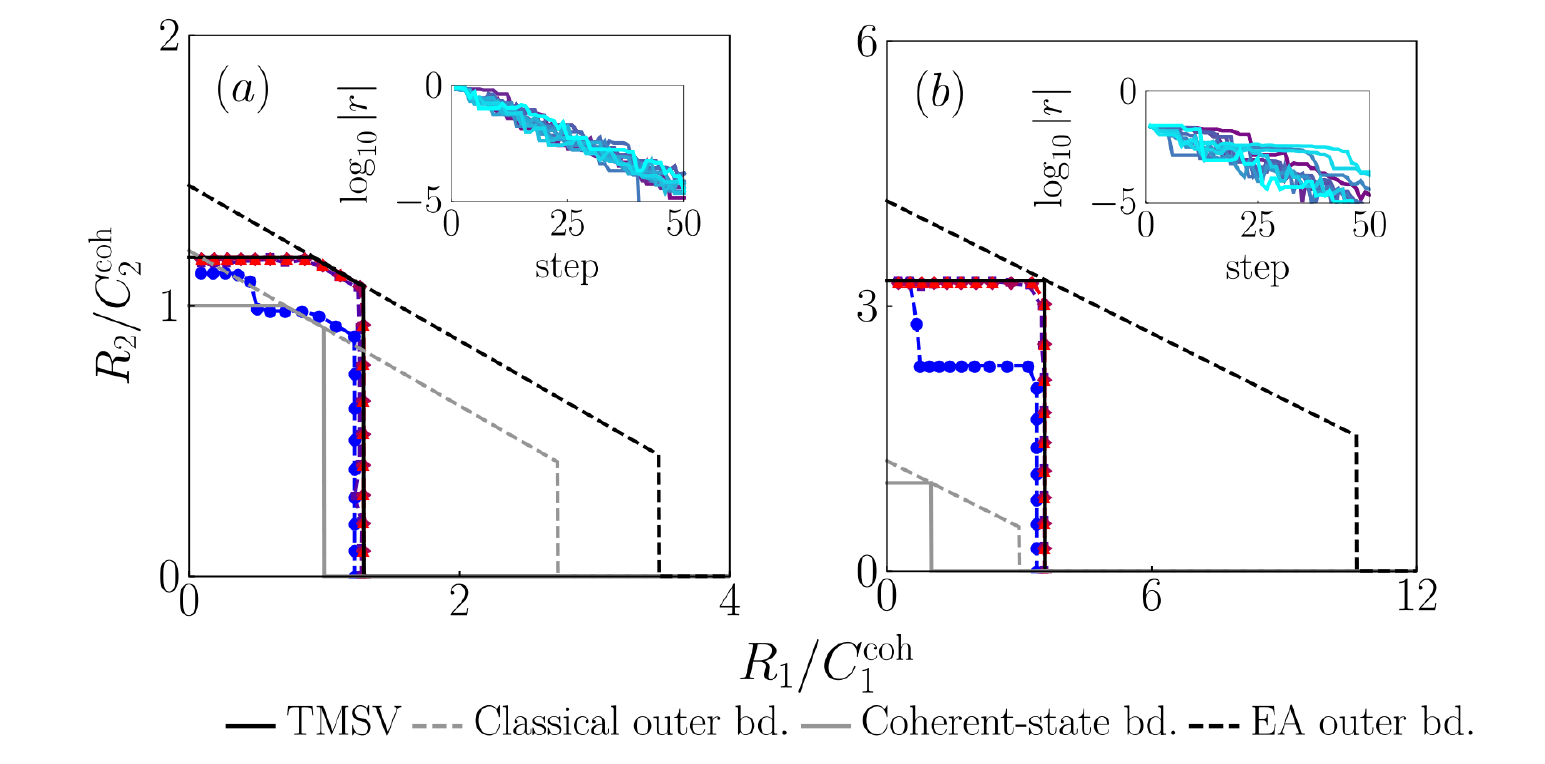}
    \caption{Optimization of the Gaussian-state rate regions of two-sender interference BGMACs with $\Psi$ being a thermal-loss channel. (a) Bright illumination $N_{{\rm S}, 1}=1,N_{{\rm S}, 2}=2$, interference ratio $\eta_1=1/3,\eta_2=2/3, |w|^2=0.1, N_{\rm B}=0.1$; (b) weak illumination $N_{{\rm S}, 1}=10^{-3},N_{{\rm S}, 2}=2\times 10^{-3}$, interference ratio $\eta_1=1/3,\eta_2=2/3, |w|^2=0.1, N_{\rm B}=0.1$.  Colored from blue to red at the $1,5,20,50$th progress steps of the numerical optimization. Insets: The evolution trend of $|r|=\sqrt{r_1^2+r_2^2}$ versus the progress step for different points on the boundary.  See Appendix~\ref{app:evalmethod} 
    for the evaluation method.}
    \label{fig:unionregion_main}
    
\end{figure}

For a thermal-loss MAC, Fig.~\ref{fig:unionregion_main} shows the evolution of the union region at the $1,5,20,50$th steps in the optimization over $(\bm r,\bm \theta)$ (see Appendix~\ref{app:evalmethod} 
for details).  We see that the union region approaches the TMSV region (black solid) as the numerical optimization proceeds. Indeed, the insets show that $(\bm r,\bm \theta)$ converges to $|\bm r|=0$ during the optimization, where $\bm \theta$ becomes irrelevant. Despite that the total rate of the TMSV region achieves the ultimate capacity, considering the individual rates, there is a appreciable gap between the TMSV region and the EA outer bound (black dashed), which implies further improvement surpassing the TMSV region is possible.
More details for the four classes including also the amplifier MAC, the conjugate amplifier MAC, and the AWGN MAC are shown in Appendix~\ref{app:supp_numerics}.

Above all, our numerical optimization shows that the TMSV state is the optimal Gaussian input in all the examined cases. Therefore, we conjecture that the optimal Gaussian input $\hat\phi$ of single-mode phase-insensitive BGMAC is an $s$-partite TMSV state, i.e.,
$ 
\calC^{(1)}_{\rm E,G}(\calN)= \tilde \calC_{\rm E}(\calN,\otimes_{k=1}^s\ket{\zeta(N_{S,k})}_{A_kA^\prime_k}).
$
Below we provide a proposition supporting the conjecture (See Appendix \ref{proof:prop_gaugeinv_local} 
for a detailed proof). 
\begin{proposition}
For a phase-insensitive BGMAC $\calN$, consider the $s$-partite TMSV state
$
\hat{\phi}_{AA^\prime}=\otimes_{k=1}^s\ket{\zeta(N_{S,k})}_{A_kA^\prime_k},
$ 
then $\tilde C_{{\rm E},J}(\calN,\hat{\phi}_{AA^\prime})$ in Eq.~\eqref{eq:CeMAC_qinfo} cannot be improved by squeezing any group of modes within $A[J]$. 
\label{proposition:gaugeinv_local}
\end{proposition}
To prove that $\bm r=\bm 0$ is the optimum in Eq.~\eqref{CG_formal}, this proposition is one step away from the exact proof: it is still unclear whether squeezing modes in both $A[J]$ and $A[J^c]$ at the same time improves the rate or not.


\subsection{Memory interference BGMAC}
\label{sec:memory_example}

\begin{figure}[t]
    \centering
    \includegraphics[width=0.45\textwidth]{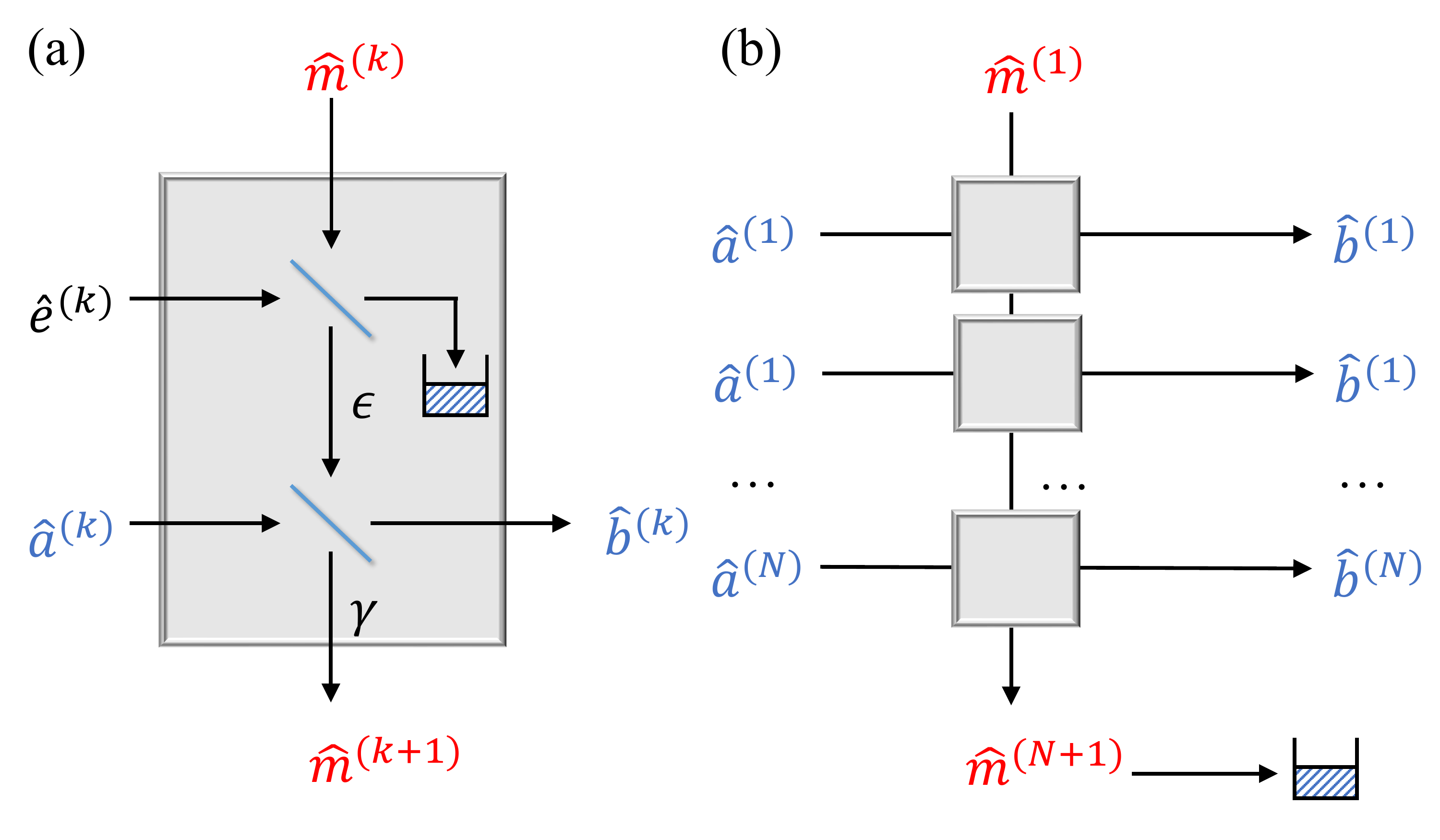}
    \caption{(a) The $k$th use of the causal memory thermal-loss channel; (b) overall representation of the $N$-fold causal memory thermal-loss channel, where the memory is inaccessible for the two communicating parties. The photon statistics at output mode $\hat b^{(k)}$ only depends on the input modes $\hat a^{(k')}$ with $k'\le k$ due to the causality.}
    \label{fig:causalmemory}
\end{figure}

In this section, we further address memory effects in interference BGMACs. To begin with, we review memory effects in BGCs. The model of memoryless BGCs assumes that the input modes from different channel uses suffer from independent noises, for example when the pulses are well-separated in a temporal sequence. In this case, for $N$ channel uses the noise environment $E$ can be decomposed into independent and identically distributed local environment modes $E^{(1)}, E^{(2)},\ldots, E^{(N)}$, which are associated with annihilation operators $\hat e^{(1)},\hat e^{(2)},\ldots,\hat e^{(N)}$. However, as the clock rate of communication increases, the signal pulses become dense in time and memory effects can result from unexpected overlaps between the input modes or interference due to the finite relaxation time of the local environment modes. For successive uses of the channel, the environment mode of a specific channel use can be correlated with input modes from previous channel uses. Such causal memory effect has been formulated in Refs.~\cite{kretschmann2005quantum,lupo2010capacities}. 

For a bosonic point-to-point thermal-loss channel, the causal memory effect can be modeled via successively connecting all of the environment modes $\{\hat e^{(k)}\}_{k=1}^N$ with a memory mode $\hat m$~\cite{lupo2010capacities}.
As shown in Fig.~\ref{fig:causalmemory}, an $N$-fold causal thermal-loss memory channel $\bm\Psi_{[N]}$ fulfills a product of identical unitary transforms $\hat U_1,\hat U_2,\ldots, \hat U_N$ with $\hat U_k$ being the interaction between the $k$th input mode, the environment mode $\hat{e}^{(k)}$ and the memory mode $\hat m^{(k)}$. Explicitly, we have
\be 
\bm\Psi_{[N]}(\hat{\bm\rho})= \hat{\bm U}_{[N]}(\hat{\bm\rho}\otimes \hat{\bm\sigma}_E)\hat {\bm U}_{[N]}^\dagger,
\ee
where $\hat{\bm \rho}$ is the channel input state and $\hat {\bm U}_{[N]}=\hat U_N\hat U_{N-1}\ldots \hat U_1$. Here $\hat{\bm \sigma}_E$ is the environment state of the memory mode $\hat m$ and all the local environment modes $\{\hat e^{(k)}\}_{k=1}^N$. Concretely, the evolution $\hat U_k$ of the $k$th memory mode and input mode $\hat a^{(k)}$ is modeled by two beamsplitters as shown in Fig.~\ref{fig:causalmemory}(a). The first beamsplitter couples the noise $\hat{e}^{(k)}$ and the memory $\hat m^{(k)}$ with transmissivity $\epsilon$; while the second beamsplitter couples the signal $\hat a^{(k)}$ and the memory $\hat m$ with transmissivity $\gamma$. The overall Bogoliubov transform is
\bal 
&\hat m^{(k+1)}=\sqrt{\epsilon\gamma}\,\hat m^{(k)}+\sqrt{1-\gamma}\,\hat a^{(k)}+\sqrt{\gamma(1-\epsilon)}\,\hat e^{(k)},
\\
&\hat b^{(k)}=-\sqrt{\epsilon(1-\gamma)}\,\hat m^{(k)}+\sqrt{\gamma} \,\hat a^{(k)}-\!\sqrt{(1-\epsilon)(1-\gamma)}\,\hat e^{(k)},
\label{eq:causalmemory_loss}
\eal 
where $\{\hat b^{(k)}\}_{k=1}^N$ are the output modes of the thermal-loss memory channel $\bm\Psi_{[N]}$.
The overall map of the $N$-fold causal memory BGC is fulfilled by iterating the Bogoliubov transform such as Eq.~\eqref{eq:causalmemory_loss} for $N$ times. Assuming the sender has no access to the memory mode, we initialize it the same as the environment noise modes. 

\begin{figure}[t]
     \centering
    \includegraphics[width=0.45\textwidth]{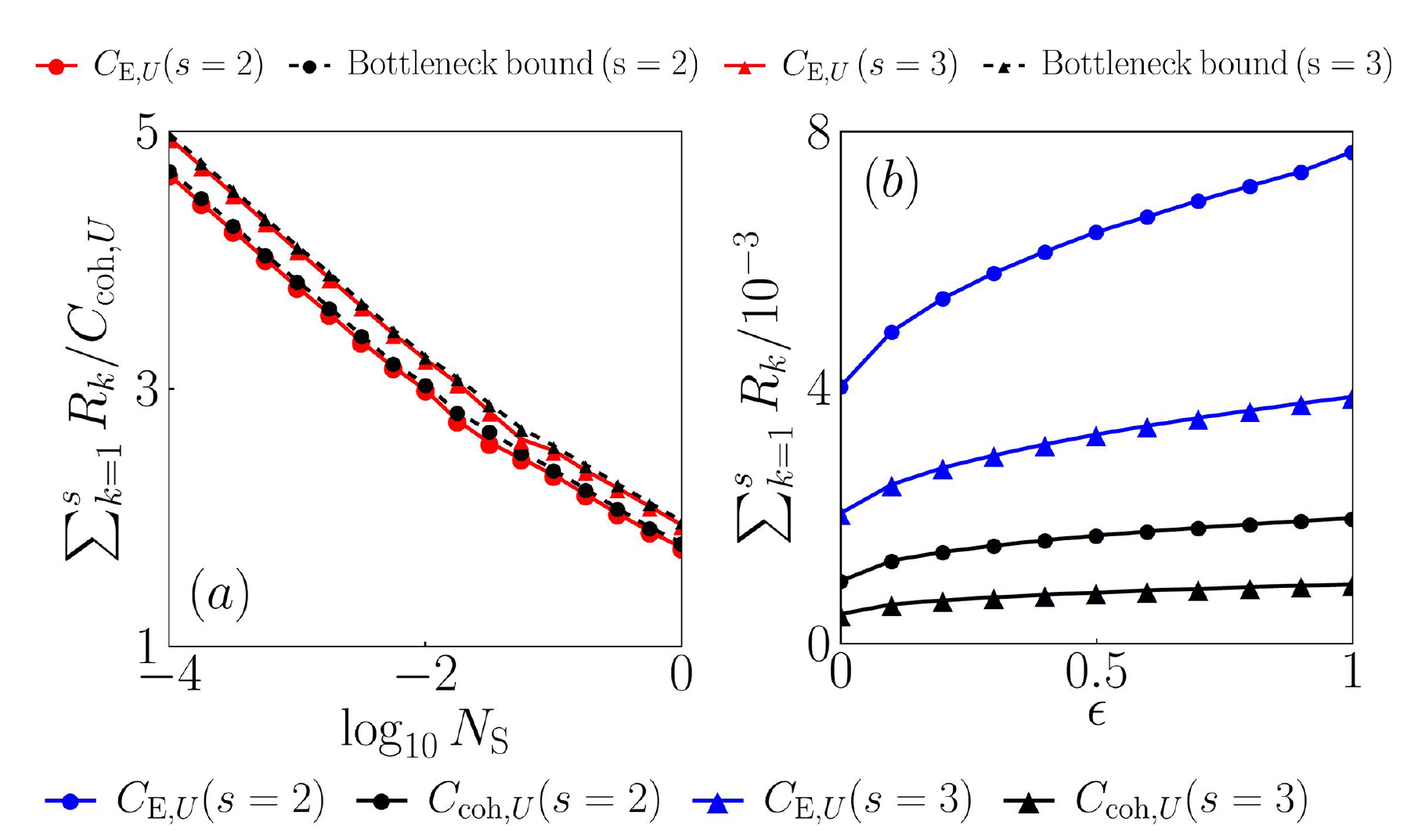}

    \caption{(a) The ultimate EA capacity (red solid) and the EA bottleneck bound (black dashed) of total rate $\sum_{k=1}^s R_k$ versus signal brightness $N_{\rm S}$ in a three-fold causal memory BGMAC $\bm\calN_{[3]}$ in the form of Eq.~\eqref{eq:def_memoryinterf}. Normalized by the coherent-state capacity $\sum_{d=1}^3 C_{{\rm coh}, U}(\calN_d)$, $U=\{1,\ldots,s\}$, where $\{\calN_d\}_{d=1}^3$ are the unravelled single-mode BGMACs. For the memory BGC component $\bm \Psi_{[N]}$, the attenuations are $\epsilon=1/2$ for the noise-memory coupling with thermal noise $N_{\rm B}=0.1$, and $\gamma=1/2$ for the signal-memory coupling. (b) The ultimate EA capacity (blue solid) and the coherent-state capacity (black solid) of total rate versus the attenuation parameter $\epsilon$ of noise-memory coupling. Dots: 2-sender case ($s=2$), total signal brightness budget per sender $N_{{\rm S},1}=0.9N_{{\rm S}}, N_{{\rm S},2}=0.1N_{{\rm S}}$, interference ratio $\eta_1=0.9,\eta_2=0.1$; triangles: 3-sender case, $N_{{\rm S},1}=N_{{\rm S}}/2$, $N_{{\rm S},2}=N_{{\rm S}}/3$, $N_{{\rm S},3}=N_{{\rm S}}/6$, $\eta_1=1/2, \eta_2=1/3, \eta_3=1/6$. 
    \label{fig:memory_evaluation} 
    }
\end{figure}

Finally, we define the $N$-fold causal memory interference BGMAC $\bm\calN_{[N]}$ as a concatenation of $N$ uses of beamsplitters and an $N$-fold causal memory thermal-loss channel $\bm\Psi_{[N]}$, namely 
\be 
\bm\calN_{[N]}=\bm\Psi_{[N]}\circ\calB^{\otimes N}\,.
\label{eq:def_memoryinterf}
\ee 
From Eq.~\eqref{eq:def_memoryinterf}, we can identify the memory interference BGMAC as a multi-mode BGMAC defined in Appendix~\ref{sec:Methods_multimode}.

Now we solve the ultimate total rate capacity of the causal memory interference BGMAC. The memory channel $\bm \Psi_{[N]}$ admits a decomposition $\bm \Psi_{[N]}=\bm \calU  \left(\otimes_{d=1}^N \Psi_d\right) \bm \calV$, where $\bm \calU$ and $\bm \calV$ are passive linear optics unitaries over the input modes and the output modes respectively~\cite{lupo2010capacities}. Here $\Psi_d$ is a single-mode BGC. Therefore, we have
\begin{align}
\bm\calN_{[N]}&=\left[\bm \calU  \left(\otimes_{d=1}^N \Psi_d\right) \bm \calV\right] \circ\calB^{\otimes N}
\\
&=\bm \calU  \left[\otimes_{d=1}^N \calN_d\right] \left(\otimes_{k=1}^s \bm \calV_k \right),
\end{align}
where we defined single-mode BGMACs $\calN_d=\Psi_d \circ\calB$ and $\calV_k$ is a unitary acting on each sender $k$. The second step following from commuting the passive linear optics unitary $\calV$ with the beamsplitters $\calB$, as detailed in Appendix~\ref{sec:Methods_unravel}. 
Now, because the passive linear optics unitary $\bm \calV_k$ conserves energy and unitaries on each individual senders or only on the receiver do not change the capacity, we have the EA total rate
\be 
C_{{\rm E},U}(\bm\calN_{[N]})=C_{{\rm E},U}(\otimes_{d=1}^N\calN_d)=  \sum_{d=1}^N C_{{\rm E},U}(\calN_d),
\ee 
where the last step is due to subadditivity (Lemma~\ref{lemma:subadditivity} of Appendix~\ref{app:Theorem1}). 
Note here each capacity $C_{{\rm E},U}(\calN_d)$ has the input energy constraint $\{N_{{\rm S},k,d}\}_{k=1}^s$ for all $s$ senders individually, which complies with the initial energy constraint $\{N_{{\rm S},k}\}_{k=1}^s$ for the $s$ senders on BGMAC $\bm\calN_{[N]}$. Namely, $\sum_{d=1}^{N}N_{{\rm S},k,d}\le N_{{\rm S},k}$ for any $1\le k\le s$. As the senders can choose an arbitrary distribution of the energy to each channel $\calN_d$, the overall EA capacity of the total rate $C_{{\rm E},U}(\bm\calN_{[N]})$ is given by an optimization over the energy constraints. The optimization can be evaluated efficiently, as each capacity $C_{{\rm E},U}(\calN_d)$ is solved by Theorem~\ref{theorem: EA_MAC_main_TMSV}, given the energy constraint $\{N_{{\rm S},k,d}\}_{k=1}^s$.

Now we evaluate the ultimate capacity of total rate over a memory thermal-loss BGMAC $\bm\calN_{[3]}$ for both 2-sender and 3-sender cases as examples.
In Fig.~\ref{fig:memory_evaluation}(a), we plot the the ultimate capacity of total rate, numerically optimized over $\{\bm N_{{\rm S},k}\}_{k=1}^s$, against $N_{\rm S}$. Meanwhile, we also optimize the coherent-state capacity $\sum_{d=1}^3 C_{{\rm coh}, U}(\calN_d)$ defined by Eq.~\eqref{eq:C_coh} as the classical benchmark, and the EA bottleneck bound $\sum_{d=1}^3 C_{\rm E}(\sum_{k=1}^s \eta_kN_{{\rm S},k},\Psi_d)$ using Eq.~\eqref{eq:CE_BGC} for comparison. As shown in Fig.~\ref{fig:memory_evaluation}(a), we see that the logarithmic EA advantage still holds. In these cases, the outer bounds simulates the ultimate EA capacity well. Similar to the memoryless case, there emerges a gap when the interference structure varies.
We note that numerically solving the whole capacity region is hard, as the additivity of the partial rates within sender set $J\neq U$ is unknown. To understand the influence of memory effects, we also vary the parameter $\epsilon$ that characterizes the strength of the memory effect in Fig.~\ref{fig:memory_evaluation}(b). We see that the total rate increases with a stronger memory effect for both the EA case and without EA. Indeed, encodings that adapt to the memory effect will benefit the communication rate.

\section{Conclusion and discussion}
We have presented three main results regarding the total EA communication rate over MACs. First, we show that it is additive for all MACs; second, we provide the formula of the ultimate EA capacity in a BGMAC; third, we show that the optimal input state of phase-insensitive BGMACs under energy constraint being a zero-mean Gaussian state. Specifically, for phase-insensitive single-mode BGMACs, the optimum EA input state is an $s$-partite TMSV state. 
Meanwhile, we have generalized the outer bounds of unassisted interference channel~\cite{yen2005multiple} to general unassisted phase-insensitive BGMACs. Equipped with our formula, numerical results show that the one-shot Gaussian-state capacity region has been sufficient to outperform the outer bounds of the unassisted protocols. 

The additivity of the whole capacity region of phase-insensitive BGMACs is still an open question. In order to prove additivity, one possible approach is to first show that the capacity of the partial rate in each possible $J$ is additive; then prove that all the partial rate capacities are simultaneously achieved by the same state on the boundaries of the capacity region. As we have shown the total rate is additive and achieved by a product of TMSV for the single-mode case, the optimal state for partial rate capacities is likely to be TMSV if the region additivity is true; on the other hand, it is also possible that superadditivity can be constructed to disprove additivity of the capacity region. By analogy with the minimum output entropy conjecture~\cite{guha2008entropy}, we propose the following conjecture to facilitate the future study of the additivity problem, which immediately leads to the optimality of TMSV and thereby the additivity of the ultimate capacity region in a single-mode thermal-loss interference BGMACs (see Appendix~\ref{app:connect} for details).

\begin{conjecture}[EA minimum entropy conjecture]
Let
$\hat{\bm a}$ be an $ns$-dimensional vector of annihilation operators acting on the composite input system $\bm A=\otimes_{\ell=1}^n(\otimes_{k=1}^s A_k^{(\ell)})$ and $\hat{\bm\xi}$ be the vector of noise operators (not necessarily canonical) from the environment system $\bm E=\otimes_{\ell=1}^n E^{(\ell)}$. The input satisfies the energy constraint in Eq.~\eqref{eq:energy_constraint} on average per channel use. Define a reference system $\bm A'=\otimes_{\ell=1}^n A^{\prime (\ell)}$ that purifies each $\otimes_{\ell=1}^nA_k^{(\ell)}$. The joint density operator is in a product state $\hat\phi_{\bm A\bm A'}\otimes \hat \sigma_{\bm E}$,
where $\hat\sigma_{\bm E} =\otimes_{\ell=1}^n \hat \sigma_{E}^{(\ell)}$
with each $\hat\sigma_{E}^{(\ell)}$
being a thermal state given a known average photon number.
At the receiver side, define an $n$-dimensional vector of annihilation operators
$\hat{\bm c}$ for the output system $\bm B=\otimes_{\ell=1}^nB^{(\ell)}$ from $n$ channel uses of Eq.~\eqref{eq:def_interf} as
$
\hat c_B^{(\ell)}=\left(\sum_{k=1}^s w_k\hat a_{A_k}^{(\ell)}\right)+\hat\xi_E^{(\ell)}, \, 1\le \ell\le n\,,
$
while the reference system is unchanged.
Denote the joint
density operator of output as $\hat\rho_{\bm B\bm A'}$. Given the energy constraint, the conjecture states that choosing $\hat\phi_{\bm A\bm A'}$ to be the $n$-mode TMSV
state minimizes the conditional von Neumann entropy $S(\bm A'[J]|\bm B\bm A'[J^c])_{\hat\rho}$ for any $J\subseteq \{1,2,\ldots,s\}$, where $\bm A'[J]=\otimes_{\ell=1}^n(\otimes_{k\in J}A_k^{(\ell)})$.
\label{conjecture: min_ent_conj}
\end{conjecture}

The proof or disproof of this conjecture will be an important future direction to understand the role of entanglement in network communication.

\begin{acknowledgements}
This project is supported by the National Science Foundation (NSF) Engineering Research Center for Quantum Networks Grant No. 1941583 and Defense Advanced Research Projects Agency (DARPA) under Young Faculty Award (YFA) Grant No. N660012014029. Q.Z. also acknowledges Craig M. Berge Dean’s Faculty Fellowship of University of Arizona. The authors thank Min-Hsiu Hsieh, Graeme Smith
and Saikat Guha for helpful discussions, and Patrick Hayden for his insightful question that leads to our formulation of Conjecture 1. 
\end{acknowledgements}


%

\appendix

\section{Multi-mode generalization }
\label{sec:Methods_multimode} 

Compared to the single mode case of Eq.~\eqref{eq:BGMAC_phaseinsensitive}, a multi-mode BGMAC allows the output $B$ to contain $N$ modes $\{\hat{a}_{B,d}\}_{d=1}^N$ and each input system $A_k$ to contain $N_k$ modes $\{\hat a_{A_k, n_k}\}_{n_k=1}^{N_k}$. Each output mode is given by the input-output relation
\bal
\hat{a}_{B,d}\!=\!&\left[\sum_{k=1}^s \sum_{n_k=1}^{N_k} \!\! w_{d,h}\!\! \left(\left(1-\delta_{d,h}\right)\hat a_{A_k, n_k}\!\!+\!\delta_{d,h}\hat a_{A_k, n_k}^\dagger \right)\right]\!\!+\!
\hat \xi_d\,,\\
\label{eq:BGMAC_phaseinsensitive_multimode}
\eal
for $1\le d\le N$, where $h(1,n_k)=n_k$ and $h(k,n_k)\equiv \sum_{j=1}^{k-1}N_j+ n_k$ for $k\ge2$ is a way to index all modes from $1$ to $\sum_{k=1}^s N_k$. 
Here the $k$th sender has access to an $N_k$-mode system $A_k$; the receiver has access to an $N$-mode system $B$; the $N$ noise modes $\{\hat\xi_d\}_{d=1}^N$ are mutually independent. It is worth pointing out that when $N=1$ and each sender has the same type of symmetry among the modes ( $\delta_{h(k,n_k)}=\delta_k$ for any $1\le n_k\le N_k$, $\delta_k=1$ for all contravariant and $\delta_k=0$ for all covariant), the channel reduces to a single-mode case subject to Eq.~\eqref{eq:BGMAC_phaseinsensitive} by combining the $N_k$ modes into one effective mode.

Similar to the single-mode case, we apply a constraint on the total mean photon number of the signal modes per sender
\be 
\sum_{n_k=1}^{N_k}\expval{\hat a_{A_k, n_k}^\dagger \hat a_{A_k, n_k}}=N_{S,k}, 1\le k \le s.
\label{eq:energy_constraint_multi_mode}
\ee

Below we address the capacity region of the multi-mode BGMAC. As Theorem~\ref{theorem: EA_MAC_main_additivity} still applies, we are able to derive Theorem~\ref{theorem: EA_MAC_main_Gaussian} and reduce the total rate evaluation to Gaussian input states. Due to the multi-mode nature, a complete characterization of the Gaussian states involves multi-mode Gaussian unitary operations, which allow the correlations over the modes within each sender. Similar to Proposition~\ref{proposition:gaugeinv_local}, we can prove the following (see Appendix~\ref{proof_multi_mode}).
\begin{proposition}
For a phase-insensitive multi-mode memory BGMAC $\calN$ with $N_k$ modes input from the $k$th sender, consider the $s$-partite TMSV
\be 
\hat{\phi}_{AA^\prime}=\otimes_{k=1}^s\left[\otimes_{n_k=1}^{N_k}\ket{\zeta(N_{S,k}^{n_k})}\right]_{A_kA^{\prime}_k}\,,
\ee 
that satisfies the energy constraint $\sum_{n_k=1}^{N_k}N_{S,k}^{n_k}=N_{S,k}$. The information quantities in Eq.~\eqref{eq:CeMAC_qinfo} cannot be improved by any of the following: 1. product of single-mode squeezing on a group of modes within $A[J]$; or 2. phase-sensitive multi-mode Gaussian unitary on modes within any specific sender.
\label{proposition:multi-mode}
\end{proposition}
The proposition above constrains the optimum state up to an extra freedom on the correlations between modes within each sender. In the case of a memory interference channel, as we shown in Sec.~\ref{sec:memory_example}, the TMSV is indeed optimal for the total rate, up to passive transforms within each sender.

For a single-sender BGC, similar to Corollary~\ref{corollary:opt_singlesender}, Proposition \ref{proposition:multi-mode} yields a necessary condition for the optimum state.
\begin{corollary}
The optimum input state of point-to-point multi-mode phase-insensitive BGC is the purification of an $s$-partite correlated thermal state
\be 
\hat{\phi}_{AA^\prime}=\otimes_{k=1}^s\Phi_{A_k\to A_kA_k'}[\rho_{A_k}]\,,
\ee
where $\Phi$ is an arbitrary purification of the $N_k$-mode correlated thermal state $\rho_{A_k}$.
\label{corollary:multi-mode}
\end{corollary}

\section{Reduction of the noise mode}
\label{app:noise_reduction}

In general, for an $s$-mode input channel, the noise system consists of $2s$ vacuum environment modes $\{\hat a_{E',k}\}_{k=1}^{2s}$~\cite{caruso2008multi,caruso2011optimal}. In a phase-insensitive BGMAC described by Eq.~\eqref{eq:BGMAC_phaseinsensitive}, the noise mode $\hat{\xi}$ can always be written in the linear combination of the $2s$ modes
\be 
\hat \xi=\sum_{k=1}^{2s} u_k\left( \left(1-\delta_{E^\prime,k}\right)\hat{a}_{E^\prime,k}+\delta_{E^\prime,k}\hat{a}_{E^\prime,k}^\dagger\right),
\ee 
where similarly the constants $\delta_{E',k}\in \{0,1\}$.
Because the $2s$ modes are identically vacuum, it can be further simplified to 
\be 
\hat \xi=|\bm u_1|\hat a_{E',1}'+|\bm u_2| \hat a_{E',2}^{\prime\dagger}
\label{eq:xi_app}
\ee
with 
\bal 
\hat a_{E',1}'= & \sum_{k=1}^s\frac{(1-\delta_{E',k})u_k}{|\bm u_1|}\hat{a}_{E^\prime,k}\,,\\
\hat a_{E',2}^{\prime\dagger}=&\sum_{k=1}^s \frac{\delta_{E',k} u_k}{|\bm u_2|}\hat{a}_{E^\prime,k}^\dagger\,,
\eal 
where $|\bm u_1|=\sqrt{\sum_{k=1}^s(1-\delta_{E',k})u_k^2}$, $|\bm u_2|=\sqrt{\sum_{k=1}^s\delta_{E',k} u_k^2}$. We see that $|\bm u_1|^2+|\bm u_2|^2=\sum_{k=1}^s u_k^2$.

\section{Proof Sketch of main results}
\label{sec:proof_sketch}

Below we provide the backbone of the proofs of the main results and theorems. More detailed proofs are presented in Appendices~\ref{app:outerbd},~\ref{app:Theorem1},~\ref{app:Theorem3},~\ref{app:LemmasForTheorem1&2}.

We begin with the unassisted case in Sec.~\ref{sec:unassisted}. To solve an outer-bound region, we unravel the general phase-insensitive BGMAC defined in Eq.~\eqref{eq:BGMAC_phaseinsensitive} into $s$ individual point-to-point BGCs for the noisy case of $N_{\rm B}\ge |\bm w|^2-1+\sum_{k=1}^s|w_k|^2\delta_k$ (see Appendix~\ref{app:outerbd}). Then, by applying a super-receiver that has access to all the $s$ outputs, we also obtain an outer-bound region that allows arbitrary input states. 


Next, we address the major part of the results of the EA case in Sec.~\ref{sec:ea}.
To prepare our analyses, we recast the conditional quantum information in Eq.~\eqref{eq:CeMAC_qinfo} into a functional $I_J$ with respect to the reduced state of the signal system $A$
\be
I_J(\calN,  \{\hat{\phi}_{A_i}\}_{i=1}^s)\equiv I(A^\prime[J];B|A^\prime[J^c])_{\hat\rho}\,.
\label{eq:I_J}
\ee
Here the output state $ \hat\rho=\calN_{A\to B}\otimes\calI_{A^\prime}({\Phi}_{A\to AA^\prime}(\otimes_{i=1}^s{\hat{\phi}_{A_i}}))$ is expressed by the reduced state of the input via a purification process, where we denote $\Phi_{X\to XX^\prime}(\hat{\zeta}_{X})$ as the purification of system $X$ in state $\hat{\zeta}$, defined in the joint system $XX^\prime$. 
Similarly, for the total rate of Eq.~\eqref{eq:CeMAC_qinfo_overall} within the universal set $U$, we have
\be
I_U(\calN,  \{\hat{\phi}_{A_i}\}_{i=1}^s)\equiv I(A^\prime;B)_{\hat\rho}\,.
\label{eq:I_U}
\ee
The above functional forms will facilitate us to prove the results.

\begin{figure}[tbp]
    \centering
    \includegraphics[width=0.45\textwidth]{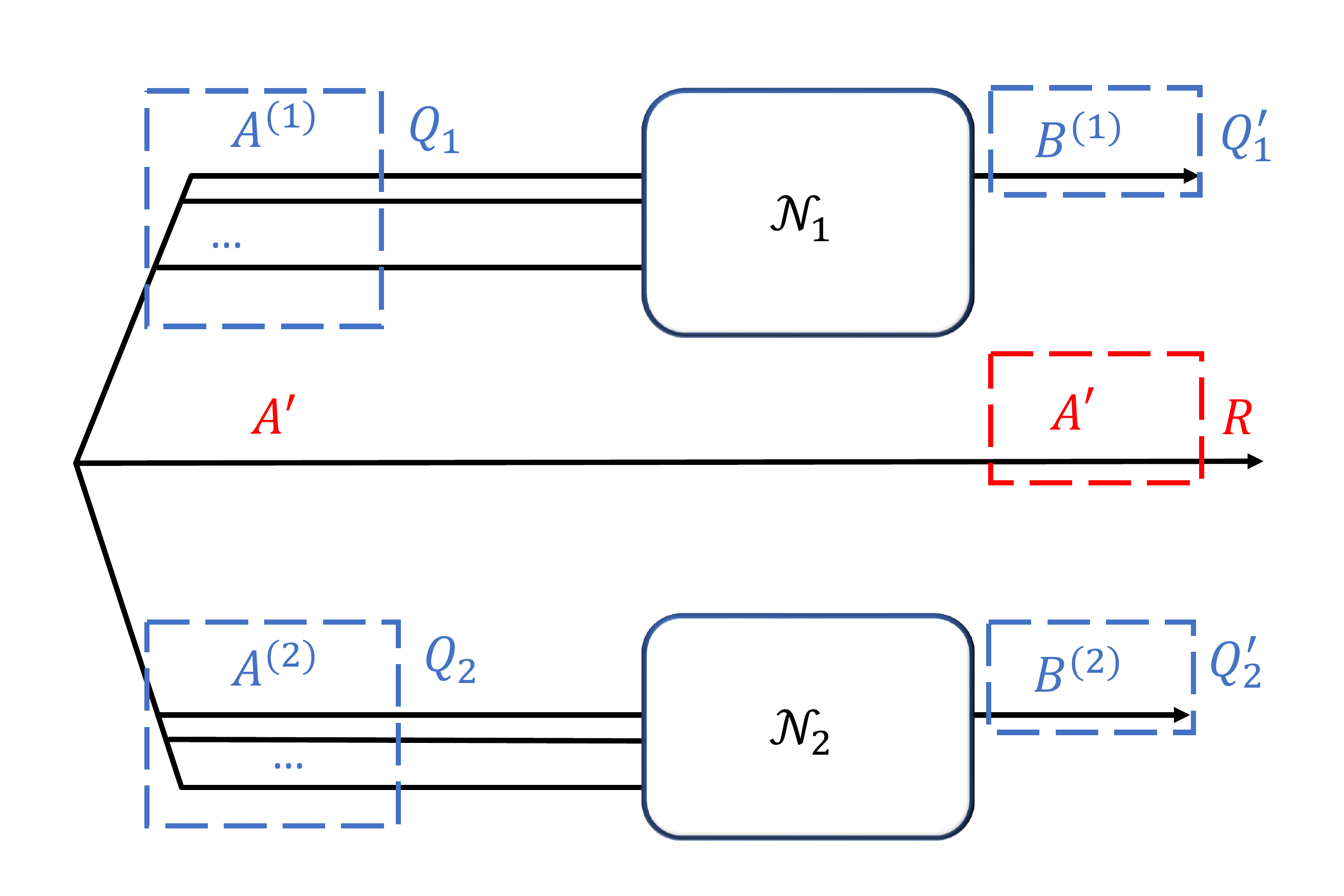}
    \caption{The input-output relation of the parallel use of two $s$-sender MACs $\calN_1,\calN_2$, for the derivation of subadditivity of the total rate with $J=\{1,2,\ldots,s\}$. $A^{(1)},A^{(2)}$ are the composite input systems of $\calN_1$, $\calN_2$ respectively, each containing $s$ independent subsystems. Here we use superscripts $(1),(2)$ to distinguish from the subscripts $k_1,k_2$ that label the sender indices as in $A_{k_1}^{(1)},A_{k_2}^{(2)}$. The subadditivity is derived by proving $I(R;Q^\prime_1Q^\prime_2)\le I(RQ_2;Q^\prime_1)+I(RQ_1;Q^\prime_2)$. Here the quantum system of interest $Q_1'Q_2'$ is $B^{(1)}B^{(2)}$. $Q_1,Q_2$ are the channel inputs for the output systems $Q_1',Q_2'$. The reference system $A'$, denoted as $R$, purifies $Q_1Q_2$. 
    }
    \label{fig:subadditivity_main}
\end{figure}

To prove the additivity in Theorem~\ref{theorem: EA_MAC_main_additivity}, we further modify the functional form of Eq.~\eqref{eq:I_J}. For each channel use, the conditional quantum information can be reduced to quantum mutual information via
\begin{align} 
I_J(\calN, \{\hat{\phi}_{A_i}\}_{i=1}^s)&=I(A^\prime[J];B|A^\prime[J^c])_{\hat\rho}
\nonumber
\\
&=I(A^\prime[J];BA^\prime[J^c])_{\hat\rho}
\nonumber
\\
&=I(R^\prime;Q^\prime)_{\hat\rho}.
\label{eq:IRQ}
\end{align} 
The second equality is due to the independence condition of Eq.~\eqref{eq:MACindependence}; in the last equality, we reorganize the systems: define the input quantum system of interest $Q=AA^\prime[J^c]$, which is purified by a reference system $R=A^\prime[J]$; after the channel $\Sigma_{Q\to Q^\prime}\equiv\calN_{A\to B}\otimes \calI_{A^\prime[J^c]}$, the quantum system becomes $Q^\prime=BA^\prime[J^c]$, while the reference $R^\prime=R$ is unchanged.

Now we examine the parallel use of two arbitrary MACs, as depicted in Fig.~\ref{fig:subadditivity_main}.
Consider the tensor product channel $\calN_1\otimes\calN_2$ as a whole, the total information rate is also in the standard form of quantum information
\be 
I_U(\calN_1\otimes\calN_2, \{\hat{\phi}_{A_i^{(1)}A_i^{(2)}}\}_{i=1}^s)=I(R^\prime;Q^{\prime}_1 Q^{\prime}_2)_{\hat\rho}\,,
\ee 
similar to Eq.~\eqref{eq:IRQ},
where $R'=A'$, $Q'_\ell=B^{(\ell)}$, $\ell\in\{1,2\}$ labels the subsystem associated with the $\ell$th channel. Assuming the same reduced state in each  channel use, the individual rate per channel use in a separable strategy
\bal 
&I_U(\calN_1, \{\hat{\phi}_{A_i^{(1)}}\}_{i=1}^s)=I(A^\prime B^{(2)} ;B^{(1)})_{\hat\rho}=I(R^\prime Q^{\prime}_2;Q^{\prime}_1)_{\hat\rho}\,,\\
&I_U(\calN_2, \{\hat{\phi}_{A_i^{(2)}}\}_{i=1}^s)=I(A^\prime B^{(1)} ;B^{(2)})_{\hat\rho}=I(R^\prime Q^{\prime}_1;Q^{\prime}_2)_{\hat\rho}\,.\\
\eal 
The subadditivity of quantum information~\cite{adami1997neumann,bennett2002entanglement}
\be 
I(R^\prime;Q^{\prime}_1 Q^{\prime}_2)_{\hat\rho}\le I(R^\prime Q^{\prime}_2; Q^{\prime}_1)_{\hat\rho}+ I(R^\prime Q^{\prime}_1; Q^{\prime}_2)_{\hat\rho}
\label{eq:subadditivity}
\ee
immediately yields the additivity of the capacity of total rate.
In the detailed proof, we shall frequently utilize the concavity and the subadditivity of quantum information with respect to input system $Q$, of which a general proof can be found in Ref.~\cite{adami1997neumann}. 

For partial communication rates within sender sets $J\neq \{1,2,\ldots,s\}$, the above proof method fails. This is because our method requires that the system of interest, e.g. $B^{(1)}B^{(2)}$ for the total rate $I(A';B^{(1)}B^{(2)})$, must be divided into one subsystem per channel use as shown in Fig.~\ref{fig:subadditivity_main}, which is not true for the partial rate $I(A'[J];B^{(1)}B^{(2)}A'[J^c])$ within sender set $J$: here the system of interest becomes $B^{(1)}B^{(2)}A'[J^c]$ with the reference system $A'[J^c]$ indivisible when $A^{(1)}[J^c]$ and $A^{(2)}[J^c]$ are correlated. For example, for $J=\{1\},J^c=\{2\}$, consider a correlated for $J^c$ in two channel uses
$ 
\hat{\rho}^\star\propto \ketbra{00}_{A_2^{(1)}A_2^{(2)}}+\ketbra{11}_{A_2^{(1)}A_2^{(2)}}.
$
By analogy with the proof above, one may expect a purification $\hat{\phi}^\star=(\ket{0000}_{A_2^{(1)}A_2^{(2)}A_2^{\prime(1)}A_2^{\prime(2)}}+\ket{1111}_{A_2^{(1)}A_2^{(2)}A_2^{\prime(1)}A_2^{\prime(2)}})/\sqrt{2}$. Unfortunately, the pairwise reduced states $\hat{\phi}^\star_{A_2^{(1)}A_2^{\prime(1)}}$ and $\hat{\phi}^\star_{A_2^{(2)}A_2^{\prime(2)}}$ are no longer purified. In this case, plugging in $R'=A_1^{\prime},Q_1'=B^{(1)}A_2^{\prime(1)},Q_2'=B^{(2)}A_2^{\prime(2)}$, the right hand side of the subadditivity Eq.~\eqref{eq:subadditivity} fails to match the individual $I_J$'s for two channel uses.

Finally, we summarize the proof of Theorem \ref{theorem: EA_MAC_main_TMSV}.
Ref.~\cite{wolf2006} gives a lemma of extremality of Gaussian states for any functional that is subject to the subadditivity and the symmetry under channel-wise Hadamard transform. 
Consider the $J$-set rate functional $I_J$. 
The Hadamard transform commutes with the MAC, thus it does not change $I_J$ as a unitary transform. Combining the subadditivity and the symmetry, we obtain Gaussian optimality in BGMACs. Therefore, we only need to optimize over the zero-mean Gaussian states to achieve the ultimate total rate. Note that any pure Gaussian state can be generated from a product of TMSV by a Gaussian unitary. Because the rate depends on merely the composite signal system $A$ and MAC forbids cooperation between senders, we can limit the Gaussian unitary to be a product of Gaussian unitaries, each on a single system $A_k$. For a single-mode BGMAC, as each sender has access to only a single signal mode, the Gaussian unitary on $A_k$ further reduces to a combination of single-mode squeezing operations and phase rotations. Later we will prove Proposition~\ref{proposition:gaugeinv_local}, which indicates that the total rate cannot be improved by any single-mode squeezing operation among the modes within $J=\{1,2,\ldots,s\}$ that includes all the $s$ senders. Hence, Theorem~\ref{theorem: EA_MAC_main_TMSV} is proven.



\section{Proof of the outer bounds}
\label{app:outerbd}

\begin{figure}[tbp]
    \centering
    \includegraphics[width=0.5\textwidth]{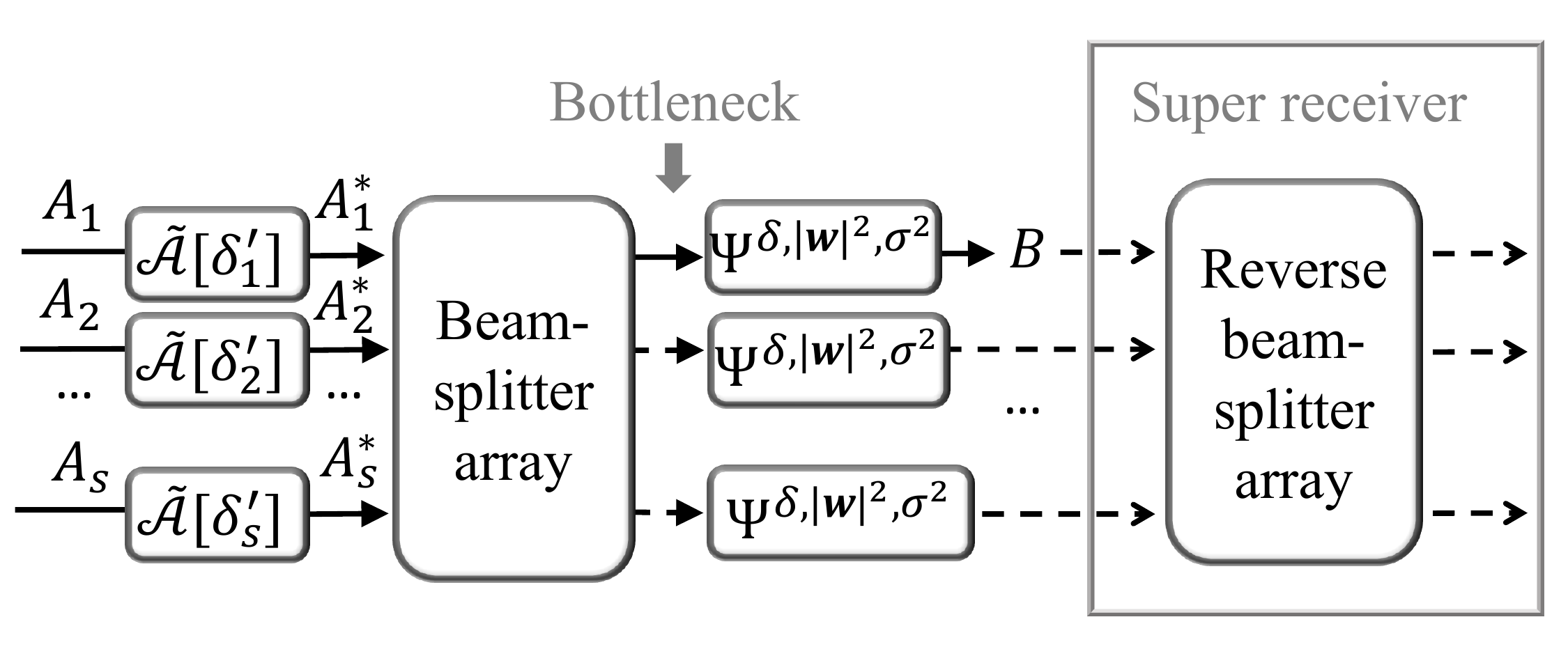}
    \caption{The unravelled version of an arbitrary phase-insensitive $s$-sender single-mode BGMAC. The original MAC maps $s$ inputs $A_1,A_2\ldots, A_s$ to a single output $B$, while here it is unravelled to an interference BGMAC after conjugate amplifiers (solid arrows). The individual bounds are evaluated by the super receiver that has access to all the $s$ outputs of the extended MAC, including the dashed arrows. The bound over the total rate is evaluated between the bottleneck right after the beamsplitter array and the final output $B$.}
    \label{fig:outerbd}
\end{figure}

Following Eq.~\eqref{aB_BGC}, we define the channel $\Psi^{\delta, w,\sigma^2}_{A\to B}$ via the Bogoliubov transform on input mode $\hat{a}_A$ as
\be 
\hat{a}_B= w \left(\left(1-\delta\right)\hat{a}_{A}+\delta\hat{a}_{A}^\dagger \right)+
\hat \xi,
\label{aB_BGC_app}
\ee 
where $\delta=0,1$.
On vacuum inputs, the channel $\Psi^{\delta, w,\sigma^2}_{A\to B}$ produces a thermal state with the mean photon number $\sigma^2$, which corresponds to `dark photon counts'. For a bona fide BGC, 
$
\sigma^2\ge \max\{(|w|^2-1)(1-\delta),|w|^2\delta\}.
$ 

We attempt to decompose the BGMAC defined in Eq.~\eqref{eq:BGMAC_phaseinsensitive} as
\be 
\calN_{A\to B}=\Psi^{\delta, |\bm w|^2,\sigma^2}\circ\calB^*_{A^*\to B}\circ(\otimes_{k=1}^s {\tilde\calA}_{A_k\to A_k^*}[\delta_k^\prime])\,,
\label{eq:simulatePIBGMAC}
\ee
where we denote the conditional phase-conjugator as 
 \be 
 {\tilde\calA}_{A_k\to A_k^*}[\delta_k^\prime]: \hat a_k^*\equiv (1-\delta_k^\prime)a_k+\delta_k^\prime(\hat a_k^\dagger+\sqrt{2} \hat v_k),\, 1\le k\le s\,.
 \ee
For the equality of Eq.~\eqref{eq:simulatePIBGMAC} to hold, we require
\begin{align}
&\delta_k=\delta_k^\prime +\delta \mod 2,
\label{constraint_delta}
\\ 
&N_{\rm B}=\sigma^2+\sum_{k=1}^s |w_k|^2\delta_k^\prime. 
\label{constraint_NB}
\end{align}
For $\Psi^{\delta, |\bm w|^2,\sigma^2}$ to be physical, we require
\be 
\sigma^2\ge \max\{(|\bm w|^2-1)(1-\delta),|\bm w|^2\delta\}.
\label{constraint_physical}
\ee

With Eq.~\eqref{eq:simulatePIBGMAC}, we construct a super-receiver that has access to all the $s$ outputs of the beamsplitter $\calB^\star$, and we denote the $s$-input-$s$-output beamsplitter as $\calB_{A^*}$ below. In this case, the $s$ outputs travel through the $s$-product channel $(\Psi^{\delta, |\bm w|^2,\sigma^2})^{\otimes s}$, which commutes with the beamsplitter since they act on orthogonal subspaces. Therefore, the receiver can fully reverse the effect of $\calB^\star$ by applying an inverse mapping $\calB^{*-1}$ to the overall channel output,
\bal  
&\calB_{A^*}^{*-1}\circ\left(\Psi^{\delta, |\bm w|^2,\sigma^2}\right)^{\otimes s}\circ\calB_{A^*}^*\circ(\otimes_{k=1}^s {\tilde\calA}_{A_k\to A_k^*}[\delta_k^\prime])\\
&=\otimes_{k=1}^s \left[\Psi^{\delta, |\bm w|^2,\sigma^2}\circ {\tilde\calA}_{A_k\to A_k^*}[\delta_k^\prime]\right]\,.
\eal
Since the unitary channel $\calB^{*-1}$ does not affect the entropic quantities, the super-receiver above provides an outer bound for the information rates. For the $k$th sender, we have the the rate $R_k$ upper bounded by the point-to-point unassisted classical capacity of channel $\Psi^{\delta, |\bm w|^2,\sigma^2}\circ {\tilde\calA}_{A_k\to A_k^*}[\delta_k^\prime]$ under input energy $N_{{\rm S},k}$, namely,
\be 
R_k\le C_{}\left(N_{{\rm S},k},\Psi^{\delta, |\bm w|^2,\sigma^2}\circ {\tilde\calA}[\delta_k^\prime]\right), \, 1\le k\le s\,.
\label{eq:C_outer_app}
\ee
Applying the minimal entropy result~\cite{holevo2013classical,giovannetti2015solution}, we obtain the bounds for individual rates as
\be 
R_k\le g\left[|\bm w|^2 \left(N_{{\rm S},k}+\delta_k^\prime\right)+\sigma^2\right]-g(\sigma^2+|\bm w|^2\delta_k^\prime),
\label{Rk_UB_supp}
\ee
where $g(x)=(x+1)\log(x+1)-x\log (x)$.

Next, we give an upper bound for the total rate of all $s$ senders.  According to the bottleneck inequality (data processing inequality), the total rate of $s$ senders through $\calN$ is upper bounded by the classical capacity of channel $\Psi^{\delta, |\bm w|^2,\sigma^2}$ under the input energy constraint $\sum_{k=1}^s \eta_k(N_{{\rm S},k}+ \delta_k^\prime)$, viz.,
\bal
\sum_{k=1}^s R_k&\le C\left(\sum_{k=1}^s \eta_k(N_{{\rm S},k}+\delta_k^\prime),\Psi^{\delta, |\bm w|^2,\sigma^2}\right)\\
&= g\left(\sum_{k=1}^s |w_k|^2(N_{{\rm S},k}+\delta_k^\prime)+\sigma^2\right)-g(\sigma^2)\,.
\label{eq:C_outer_sum_app}
\eal

To obtain the tightest bound, we minimize R.H.S. of Eq.~\eqref{eq:C_outer_sum_app} and Eq.~\eqref{Rk_UB_supp} over $\delta^\prime_k$ and $\delta$ subject to constraints of Eqs.~\eqref{constraint_delta},~\eqref{constraint_NB} and~\eqref{constraint_physical}. For the global covariant case $\delta_k=0$, we choose $\delta=\delta^\prime_k=0$, such that $\sigma^2=N_B\ge |\bm w|^2-1$ is always satisfied; for the global contravariant case $\delta_k=1$, we choose $\delta=1$ and $\delta^\prime_k=0$, such that $\sigma^2=N_{\rm B}\ge |\bm w|^2$ is always satisfied. 


The outer bounds of Ineqs.~\eqref{eq:C_outer} and ~\eqref{eq:C_outer_sum} are obtained by choosing $\delta_k^\prime=\delta_k$ and $\delta=0$, such that $\sigma^2=N_B-\sum_{k=1}^s |w_k|^2\delta_k$. In this case, a physical $\Psi^{\delta, |\bm w|^2,\sigma^2}$ requires
\be 
N_B\ge \max\{(|\bm w|^2-1),0\}+\sum_{k=1}^s |w_k|^2\delta_k\,.
\ee 
Meanwhile, the outer bounds of Ineqs.~\eqref{eq:C_outer_contra} and ~\eqref{eq:C_outer_sum_contra} are evaluated under $\delta_k^\prime=1-\delta_k$ and $\delta=1$, such that $\sigma^2=N_B-\sum_{k=1}^s |w_k|^2(1-\delta_k)$. In this case we need 
\be 
N_B\ge |\bm w|^2+\sum_{k=1}^s |w_k|^2(1-\delta_k)\,.
\ee


Indeed, by the unraveling in Eq.~\eqref{eq:simulatePIBGMAC}, we have demonstrated the following lemma.
\begin{lemma}
A global covariant BGMAC with $\delta_k=0$ (or a global contravariant BGMAC with $\delta_k=1$) can always be reduced into a covariant interference BGMAC (or a contravariant BGMAC).
\end{lemma}
For interference BGMACs defined by Eq.~\eqref{eq:def_interf}, by plugging in $w_k=\sqrt{\eta_k(\tau-\delta)}$, where $\delta=0,1$ for the covariant case and the contravariant case respectively, one immediately obtain Eqs.~\eqref{eq:C_outer_interference}~\eqref{eq:C_outer_sum_interference}.\\

Finally, we note that above unravelling also applies to the EA communication. In the EA communication protocol, the super-receiver for the individual rates and the bottleneck after the beamsplitter array for the total rate can be constructed following the same procedure in the derivation of Eqs.~\eqref{eq:C_outer_app} and~\eqref{eq:C_outer_sum_app}, giving
\be 
R_k\le C_{\rm E}\left(N_{{\rm S},k},\Psi^{\delta, |\bm w|^2,\sigma^2}\circ {\tilde\calA}[\delta_k^\prime]\right), \, 1\le k\le s\,
\label{eq:CE_outer_app}
\ee
and
\bal
\sum_{k=1}^s R_k&\le C_{\rm E}\left(\sum_{k=1}^s \eta_k(N_{{\rm S},k}+\delta_k^\prime),\Psi^{\delta, |\bm w|^2,\sigma^2}\right).
\label{eq:CE_outer_sum_app}
\eal
Note that $\Psi^{\delta, |\bm w|^2,\sigma^2}\circ {\tilde\calA}[\delta_k^\prime]$ and $\Psi^{\delta, |\bm w|^2,\sigma^2}$ are both covariant BGCs ($\delta$=0) or contravariant BGCs ($\delta$=1), the outer bound above can be easily evaluated using Eq.~\eqref{eq:CE_BGC} for $\delta$=0 or Eq.~\eqref{eq:CE_BGC_contra} for $\delta$=1. 
For the global covariant case $\delta_k=0$, we choose $\delta=\delta^\prime_k=0$, such that $\sigma^2=N_B\ge |\bm w|^2-1$ is always satisfied.
For the global contravariant case $\delta_k=1$, we choose $\delta=1$ and $\delta^\prime_k=0$, such that $\sigma^2=N_{\rm B}\ge |\bm w|^2$ is always satisfied. 
By plugging in $w_k=\sqrt{\eta_k(\tau-\delta)}$, where $\delta=0,1$ for the covariant case and the contravariant case respectively, one immediately obtain Eqs.~\eqref{eq:CE_outer_interference}~\eqref{eq:CE_outer_sum_interference} and their contravariant generalizations.

\section{Connecting the minimum entropy conjecture to the capacity region}
\label{app:connect}
Consider an $s$-sender quantum MAC $\calN$, each of the EA capacity region boundary involves a specific sender block $J$. For $n$ channel uses, define the $ns$-mode input system of $\calN$ as $\bm A=\otimes_{\ell=1}^n(\otimes_{k=1}^s A_k^{(\ell)})$, purified by $\bm A'=\otimes_{k=1}^s \bm A'_k$ for each of the $s$ senders, and an $n$-mode output system $\bm B=\otimes_{\ell=1}^n B^{(\ell)}$. For any $J\subseteq \{1,2,\ldots,s\}$, from Eq.~\eqref{eq:CE_J} we have
\bal 
&C_{{\rm E},J}(\calN)=\lim_{n\to\infty}  \frac{1}{n}\max_{\hat{\bm\phi}}\tilde C_{{\rm E},J}(\calN^{\otimes n},\hat{\bm\phi})\\
&=\lim_{n\to\infty}  \frac{1}{n} \max_{\hat{\bm\phi}} \left[S(\bm A'[J])_{\hat{\bm\rho}}+S(\bm B\bm A'[J^c])_{\hat{\bm\rho}}\right]-S(\bm B\bm A')_{\hat{\bm\rho}}\\
&=\lim_{n\to\infty}  S_{\max}(\calN, \{N_{S,k}\}_{k=1}^s)-\frac{1}{n}S_{\min}(\calN^{\otimes n}),
\eal 
where the von Neumann entropies are evaluated on the output state $\hat{\bm\rho}=(\calN_{A\to B}\otimes \calI_{A'})^{\otimes n}(\hat {\bm\phi}_{\bm A\bm A'})$. In the last equality, we define two quantities, detailed below. First,
\bal 
S_{\max}(\calN, \{N_{S,k}\}_{k=1}^s)&\equiv \max_{\hat{\bm\phi}\in {\bm \calH}(\{N_{S,k}\}_{k=1}^s)} S(\bm A'[J])_{\hat{\bm\rho}}/n
\\
&=\max_{\hat{\bm\phi}\in {\bm \calH}(\{N_{S,k}\}_{k=1}^s)} S(\bm A[J])_{\hat{\bm\rho}}/n,
\label{eq:Smax_app}
\eal  
where $\hat{\rho}=\calN_{A\to B}\otimes \calI_{A'}(\hat {\bm\phi}_{ A A'})$ is the output state of a single channel use, and $\hat{\bm\phi}$ is maximized over the Hilbert space ${\bm \calH}(\{N_{S,k}\}_{k=1}^s)$ satisfying the energy constraint of on average.
In the last step, we utilized the purity of the input $\hat{\bm \phi}_{\bm A \bm A^\prime}$.
The second quantity is a minimum entropy
\be 
S_{\min}(\calN^{\otimes n})\equiv \min_{\bm{\hat\phi}} S(\bm B\bm A')_{\bm{\hat\rho}}-S(\bm B\bm A'[J^c])_{\bm{\hat\rho}}\,.
\label{eq:Smin_app}
\ee
Given the energy constraint of Eq.~\eqref{eq:energy_constraint} on average per channel use, Eq.~\eqref{eq:Smax_app} is the total entropy, which is maximized by an identical product of $s$-partite thermal state under the input energy constraint~\cite{giovannetti2013electromagnetic}. 
Note that the purification of a $s$-partite thermal state is the $s$-partite TMSV state $\hat{\zeta}$, we have
\be 
S_{\max}(\calN, \{N_{S,k}\}_{k=1}^s)=S(A[J])_{\hat\zeta}\,.
\label{eq:Smax_tmsv_app}
\ee 
Meanwhile, Conjecture~\ref{conjecture: min_ent_conj} states that
\begin{align}
S_{\min}(\calN^{\otimes n})=S[(\calN\otimes\calI)^{\otimes n}(\hat{\bm \zeta})]=n S[\calN\otimes\calI(\hat\zeta)]\,,
\label{eq:Smin_tmsv_app}
\end{align}
where $\hat {\bm \zeta}=\hat{\zeta}^{\otimes n}$. Combining with Eq.~\eqref{eq:Smax_tmsv_app}, Conjecture~\ref{conjecture: min_ent_conj} leads to the optimality of the $s$-partite TMSV for each $J\subseteq \{1,2,\ldots,s\}$ simultaneously
\be 
C_{{\rm E},J}(\calN)=S( A'[J])_{\hat\zeta}+S( B A'[J^c])_{\hat\zeta}-S( B A')_{\hat\zeta}\,.
\ee
This immediately leads to the additivity of the capacity region.


\section{Proof of Theorem \ref{theorem: EA_MAC_main_additivity}}
\label{app:Theorem1}


We first prove the subadditivity using the definition with $I_U$ in Eq.~\eqref{eq:I_U}. The subadditivity of $I_U$ is stated as follows.
\begin{lemma}
For any two parallel channel uses acting on $A_1$, $A_2$, the subadditivity follows: 
\begin{align} 
&I_U(\calN_1\otimes\calN_2, \{\hat\phi_{A_i^{(1)}A_i^{(2)}}\}_{i=1}^s)
\nonumber
\\
&\le I_U(\calN_1, \{\hat\phi_{A_i^{(1)}}\}_{i=1}^s)+I_U(\calN_2, \{\hat\phi_{A_i^{(2)}}\}_{i=1}^s).
\end{align}
\label{lemma:subadditivity}
\end{lemma}
We defer the proof in Sec.~\ref{app:subadditivity}. The subadditivity of the EA classical capacity of total rate follows
\bal 
C_{{\rm E},U}(\calN)&=\frac{1}{n}\max_{\hat\phi,n} I_U(\calN^{\otimes n}, \{\hat\phi_{\otimes_{\ell=1}^nA_i^{(\ell)}}\}_{i=1}^s)\\
&\le \frac{1}{n}\sum_{\ell=1}^n \max_{\hat\phi,n}I_U(\calN, \{\hat\phi_{A_i^{(\ell)}}\}_{i=1}^s)\\
&=\frac{1}{n}\cdot nC_{{\rm E},U}^{(1)}(\calN)\\
&=C_{{\rm E},U}^{(1)}(\calN)\,.
\eal

On the other hand, the superadditivity of $C_{{\rm E},U}$ is trivial since $nC^{(1)}_{{\rm E},U}(\calN)$ is immediately achieved by $n$-fold repetitive optimum encoding of $C^{(1)}_{{\rm E},U}(\calN)$ for $\calN^{\otimes n}$.

Combining the subadditivity and the superadditivity, we obtain the additivity
\be
C_{{\rm E},U}(\calN)=C^{(1)}_{{\rm E},U}(\calN)\,.
\ee\\

For later convenience, we generalize the subadditivity to $I_J$ conditioned on an extra independence constraint as follows (See Sec.~\ref{proof_lemma3} for a proof).
\begin{lemma}
For any two parallel channel uses acting on $A_1$, $A_2$, the additivity follows: 
\begin{align} 
&I_J(\calN_1\otimes\calN_2, \{\hat\phi_{A_i^{(1)}A_i^{(2)}}\}_{i=1}^s)
\nonumber
\\
&\le I_J(\calN_1, \{\hat\phi_{A_i^{(1)}}\}_{i=1}^s)+I_J(\calN_2, \{\hat\phi_{A_i^{(2)}}\}_{i=1}^s)\,,
\end{align}
\label{lemma:subadditivity_region}
under the constraint that the inputs of senders in $J^c$ for two channels $A^{(1)}[J^c],A^{(2)}[J^c]$ are mutually independent.
\end{lemma}

\section{Proof of Theorem \ref{theorem: EA_MAC_main_Gaussian}}

\label{app:Theorem3}

The proof relies on the Gaussian extremality theorem~\cite{wolf2006} and the sub-additivity of quantum mutual information $I_U$. Note that $I_U$ only depends on $\hat\phi_A$, we only need to prove that the optimal reduced state in $A$ is Gaussian, then its purification $\hat{\phi}_{AA'}$ is immediately Gaussian. We first prove that Eq.~\eqref{eq:CeMAC_qinfo_overall}
is optimized by a Gaussian state.

Revealed in Lemma 1 of Ref.~\cite{wolf2006}, the quantum central limit theorem states that, given any state $\hat{\rho}$ and a Gaussian state $\hat{\rho}_G$ with the same first and second order moments, the lemma states
\be
f(\hat{\rho})\le f(\hat{\rho}_G)
\label{eq:Gaussopt}
\ee
for a sub-additive continuous $s$-mode functional $f:\calH^{\otimes s}\to R$ which is invariant under local unitary acting on the $n$-copy state $f[\calU^{\otimes s}(\hat{\rho}^{\otimes n})]=f(\hat{\rho}^{\otimes n})$. Specifically, $\calU$ fulfills an $n\times n$ Hadamard transform $H^{\otimes m}$ acting on the canonical operators in the Heisenberg picture per mode $j\in\{1,2,\ldots,s\}$ as
\be 
\hat{q}_j^{(k)} \to \sum_{l=1}^n \frac{H^{\otimes m}_{kl}}{\sqrt n}\hat q_j^{(l)}, \quad \hat{ p}_j^{(k)}\to \sum_{l=1}^n \frac{H^{\otimes m}_{kl}}{\sqrt n}\hat p_j^{(l)}
\label{eq:Hadamard_def}
\ee
where $H$ is the $2\times 2$ Hadamard matrix, $n=2^m$. This lemma indicates that the optimum of any functional $f$ can be achieved by Gaussian states when the argument $\hat{\rho}$ is constrained by a condition with respect to the covariance matrix.

The sub-additivity of $I_U$ has been proven in our Lemma \ref{lemma:subadditivity}. In Sec.~\ref{app:LemmasForTheorem1&2} we will prove the invariance of $I_U$, with respect to the argument $\otimes_{k=1}^s \hat\phi_{A_k}$, under the Hadamard transform as the following lemma
\begin{lemma}
Any entropic functional, which is a linear combination of Von Neumann entropies of the argument, is invariant under an $n\times n$ Hadamard transform $H^{\otimes \log_2 n}$ defined by Eq.~\eqref{eq:Hadamard_def}.
\label{lemma:invariance_Hadamard}
\end{lemma}
Thus we have Eq.~\eqref{eq:Gaussopt} for functional $I_U$ with respect to $\otimes_{k=1}^s \hat\phi_{A_k}$, and it leads to the Gaussian maximality of $I_U$
\begin{lemma} 
For any state $\hat{\phi}$, the Gaussian state $\hat{\phi}^G$ with the same mean and covariance matrix of $\hat{\phi}$ provides a larger value,
\be 
I_U(\calN, \{\hat\phi_{A_k}\}_{k=1}^s)\le I_U(\calN, \{\hat\phi^G_{A_k}\}_{k=1}^s),
\ee 
while satisfying the energy constraints.
\label{lemma:Gaussopt}
\end{lemma}
Therefore, for any state $\hat{\phi}$,
\be 
\tilde C_{{\rm E},U}(\calN,\hat{\phi} )\le \tilde C_{{\rm E},U}(\calN,\hat{\phi}^G ),
\ee 
we can always restrict 
\be 
C_{{\rm E},U}(\calN)=C_{{\rm E},U}^{(1)}(\calN)=\max_{\hat{\phi}^G} \tilde C_{{\rm E},U}(\calN,\hat{\phi}^G ),
\ee 
with some Gaussian state $\hat{\phi}^G$ satisfying the same energy constraints $\sum_{n_k}\expval{\hat a_{A_k, n_k}^\dagger \hat a_{A_k, n_k}}\le N_{{\rm S},k}$. 

Now we consider the mean of the Gaussian state. Because $I_U$ can be written as a linear combination of entropy, which is only a function of the covariance matrix, the mean does not change $\tilde C_{{\rm E},U}(\calN,\hat{\phi}^G )$. Since a non-zero mean always consumes energy, the optimum is achieved by zero-mean Gaussian states.

\section{Proofs of Lemmas for Theorem~\ref{theorem: EA_MAC_main_additivity} and Theorem~\ref{theorem: EA_MAC_main_Gaussian}}

\label{app:LemmasForTheorem1&2}

\subsection{Proof of Lemma~\ref{lemma:subadditivity}}
\begin{figure}[tbp]
    \centering
    \includegraphics[width=0.45\textwidth]{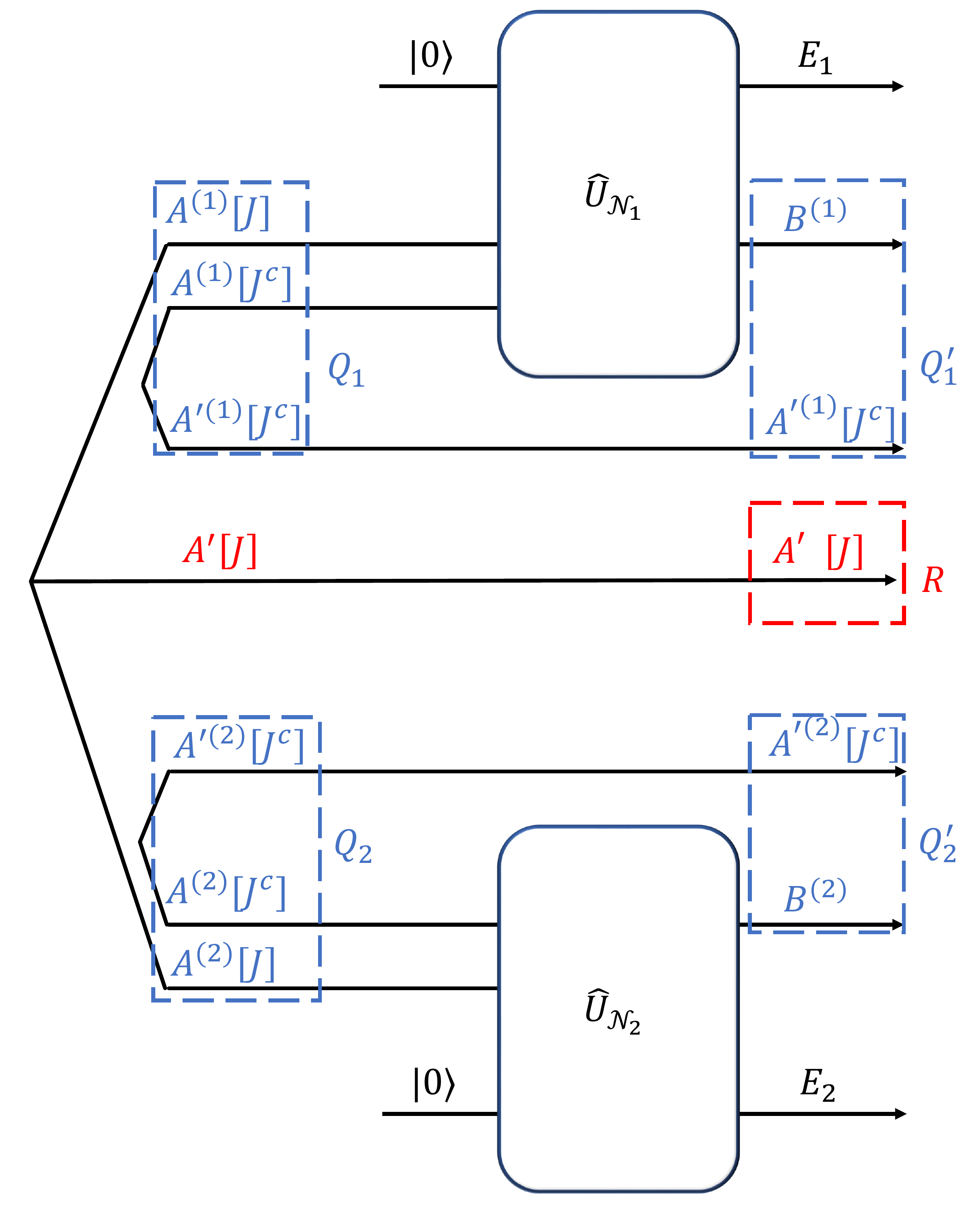}
    \caption{The input-output relation of the parallel uses of two $s$-sender MACs $\calN_1,\calN_2$ assuming $A^\prime[J^c]$ separable, for the derivation of the subadditivity within $A[J]$ as $I(R;Q^\prime_1Q^\prime_2)\le I(RQ_2;Q^\prime_1)+I(RQ_1;Q^\prime_2)_{}$. Here the quantum system of interest $Q'_1Q'_2$ is $B^{(1)}A^{\prime(1)}[J^c]B^{(2)}A^{\prime(2)}[J^c]$. $Q_1,Q_2$ are the input systems for the two uses of channel $\calN\otimes\calI_{A[J^c]}$ and the output systems $Q'_1,Q'_2$. The channels $\calN_1,\calN_2$ are extended to unitaries $\hat\calU_{\calN_1},\hat\calU_{\calN_2}$ including environment modes $E_1,E_2$. The reference system $R$ purifies $Q_1'E_1Q_2'E_2$.}
    \label{fig:subadditivity}
\end{figure}

\label{app:subadditivity}
We first prove the subadditivity of quantum information Eq.~\eqref{eq:subadditivity}
\be 
I(R;Q^\prime_1Q^\prime_2)\le I(RQ_2;Q^\prime_1)+I(RQ_1;Q^\prime_2)
\ee 
for the output state in Fig.~\ref{fig:subadditivity} with $J=\{1,2,\ldots,s\}, J^c=\varnothing$. We include the environment systems $E_1$, $E_2$ to extend the channel to a joint unitary. Note that $Q^\prime_1Q^\prime_2RE_1E_2$ is pure, 
\bal
I(R;Q^\prime_1Q^\prime_2)&=S(Q^\prime_1Q^\prime_2)+S(Q^\prime_1Q^\prime_2|E_1E_2)\,,\\
I(RQ_2;Q^\prime_1)&= S(Q^\prime_1)+S(Q^\prime_1|E_1)\,,\\
I(RQ_1;Q^\prime_2)&= S(Q^\prime_2)+S(Q^\prime_2|E_2)\,.
\eal
Now we prove the subadditivity of conditional entropy. 
\bal 
&S(Q^\prime_1Q^\prime_2|E_1E_2)\\
&=S(Q^\prime_1|E_1E_2)+S(Q^\prime_2|E_1E_2)-I(Q^\prime_1;Q^\prime_2|E_1E_2)\\
&\le S(Q^\prime_1|E_1E_2)+S(Q^\prime_2|E_1E_2)\\
&= S(Q^\prime_1|E_1)- S(Q^\prime_1;E_2|E_1)+S(Q^\prime_2|E_2)- S(Q^\prime_2;E_1|E_2)\\
&\le  S(Q^\prime_1|E_1)+S(Q^\prime_2|E_2)
\eal
Combining together with the subadditivity of Von Neumann entropy, Eq.\eqref{eq:subadditivity} is proven.

For the total rate $J=U\equiv \{1,2,\ldots,s\}$. Plugging in $R= A^\prime, Q_1= B_1, Q_2= B_2$, we have
\bal
I_U&(\hat{\phi}_{A_1 A_2 })\\
&=I(A_1'A_2';B_1B_2)\\
&\le I(A_1'A_2';B_1)+I(A_1'A_2';B_2)\\
&=I_U(\hat{\phi}_{A_1 })+I_U(\hat{\phi}_{A_2 })
\label{eq:I_U_additivity}
\eal
Now we have the sub-additivity of functional $I_U(\hat\phi_{A })$.


We note that this proof does not specify the dimensions of the systems $\{A^{(1)}_k\}_{k=1}^s$, $\{A^{(2)}_k\}_{k=1}^s$, thus it extends to the case where each $A^{(1)}_k$, $A^{(2)}_k$ contains multiple modes.

\subsection{Proof of Lemma~\ref{lemma:subadditivity_region}}
\label{proof_lemma3}
We adopt the same setup in the proof of Lemma~\ref{lemma:subadditivity}. Consider the overall output state.

For the partial rates $J\neq U$, we define $A'$ that purifies the input $A^{(1)}A^{(2)}$. Let $R=A', Q_1'=B^{(1)}, Q_2'=B^{(2)}$.
By this definition the overall systems $Q_1Q'_2E_2R$, $Q'_1E_1Q_2R$ and $Q'_1Q'_2R$ are pure. In general, $R$ is likely inseparable, then the procedure proving the subadditivity of total rate above gives no conclusion. Nevertheless, the partial subadditivity holds that the correlation between $A^{(1)}[J],A^{(2)}[J]$ does not improve $I_J$ when $A^{(1)}[J^c],A^{(2)}[J^c]$ are independent. This is because $A'[J^c]$ is separable in this case as $A'[J^c]=A^{\prime(1)}[J^c]A^{\prime(2)}[J^c]$, then we can let $Q_1=A^{(1)}A^{\prime(1)}[J^c]$, $Q_2=A^{(2)}A^{\prime(2)}[J^c]$, $R=A^{\prime}[J]$ as shown in Fig.~\ref{fig:subadditivity}. Explicitly, conditioned on that $A[J^c]$'s are independent over different channel uses, we obtain the subadditivity in $A[J]$
\bal
&I_J(\calN_1\otimes\calN_2,\{\hat{\phi}_{A_i^{(1)} A_i^{(2)} }\}_{i=1}^s)\\
&=I(A'[J];B^{(1)}B^{(2)}A^{\prime(1)}[J^c]A^{\prime(2)[J^c]})_{}\\
&=I(R;Q^\prime_1Q^\prime_2)\\
&\le I(RQ_2;Q^\prime_1)+I(RQ_1;Q^\prime_2)_{}\\
&= I(RB^{(2)}A^{\prime(2)}[J^c];B^{(1)}A^{\prime(1)}[J^c])_{}\\
&\quad+I(RB^{(1)}A^{\prime(1)}[J^c];B^{(2)}A^{\prime(2)}[J^c])\\
&=I_{J}(\calN_1, \{\hat{\phi}_{A_i^{(1)}}\}_{i=1}^s)+I_{J}(\calN_2, \{\hat{\phi}_{A_i^{(2)}}\}_{i=1}^s)\,.
\label{eq:I_J_additivity}
\eal

We note that this proof does not specify the dimensions of the systems $\{A^{(1)}_k\}_{k=1}^s$, $\{A^{(2)}_k\}_{k=1}^s$, thus it extends to the case where each $A^{(1)}_k$, $A^{(2)}_k$ contains multiple modes.

\subsection{Proof of Lemma~\ref{lemma:invariance_Hadamard}}

First we examine the single-mode BGMAC. The $n$-copy BGMAC $\calN^{\otimes n}$ forms an identity transform over the inputs of the $n$ parallel channel uses per sender $j$, which commutes with $\calU^{\otimes s}$ up to a contraction $\calU^{\otimes s}\to \calU$. Explicitly, given the linear form Eq.\eqref{eq:BGMAC_phaseinsensitive} of $\calN$, we have
\bal
&\calN^{\otimes n}\circ \calU^{\otimes s}:\\
&\hat{a}_{B}^{(k)}\!\!=\!\left[\sum_{j=1}^s w_j \sum_{l=1}^n \frac{H^{\otimes m}_{kl}}{\sqrt n}\left(\left(1-\delta_j\right)\hat{a}_{A_j}^{(l)}+\delta_j\hat{a}_{A_j}^{(l)\dagger} \right)\right]\!+\!\hat\xi^{(k)}\,,\\
&\calU\circ \calN^{\otimes n}:\\
&\hat{a}_{B}^{(k)}=\sum_{l=1}^n \frac{H^{\otimes m}_{kl}}{\sqrt n}\left[\sum_{j=1}^s w_j \left(\left(1-\delta_j\right)\hat{a}_{A_j}^{(l)}+\delta_j\hat{a}_{A_j}^{(l)\dagger} \right)+\hat\xi^{(l)}\right]\,.
\eal
Thus $\calN^{\otimes n}\circ \calU^{\otimes s}= \calU\circ \calN^{\otimes n}$ since the noises $\{\hat\xi^{(\ell)}\}_{\ell=1}^n$ are independent and identically distributed. Therefore, any entropic functional is symmetric under $\calU^{\otimes s}$ since unitary transform on the output state does not change the Von Neumann entropies. \\

For an $N$-fold memory BGMAC, the structure within $\calN$ does not overturn the commutation relation
\begin{widetext}
 \bal
&\calN^{\otimes n}\circ (\otimes_{k=1}^s\calU^{\otimes n_k}):&& \hat{a}_{B,d}^{(k)}\!=\!\left[\sum_{j=1}^s \sum_{n_j=1}^{N_j}  w_{d,n_j} \!\sum_{l=1}^n \frac{H^{\otimes m}_{kl}}{\sqrt n}\!\!\left(\left(1-\delta_{d,n_j}\right)\hat{a}_{A_j,n_j}^{(l)}\!+\!\delta_{d,n_j}\hat{a}_{A_j,n_j}^{(l)\dagger} \right)\!\right]+\hat\xi_d^{(k)}\,,\\
&\calU^{\otimes N}\circ \calN^{\otimes n}: && \hat{a}_{B,d}^{(k)}\!=\!\sum_{l=1}^n\! \frac{H^{\otimes m}_{kl}}{\sqrt n}\left[\sum_{j=1}^s \sum_{n_j=1}^{N_j} \!\! w_{d,n_j}\! \left(\left(1-\delta_{d,n_j}\right)\hat{a}_{A_j,n_j}^{(l)}+\delta_{d,n_j}\hat{a}_{A_j,n_j}^{(l)\dagger} \right)+\hat\xi_d^{(l)}\right]\,,
\eal
\end{widetext}
for $1\le d\le N$. Note that given $d$ the noise terms $\hat \xi_d^{(l)}$ are independent and identically distributed over channel uses $1\le l\le n$, we have $\calN^{\otimes n}\circ (\otimes_{k=1}^s\calU^{\otimes n_k})= \calU^{\otimes N}\circ \calN^{\otimes n}$. 

Combining with Lemma \ref{lemma:subadditivity}, the proof above for the single-mode BGMAC straightforwardly extends to the multi-mode memory BGMAC.

\section{Numerical method of evaluating the union region}
\label{app:evalmethod}
Here we introduce our method to evaluate the union of Gaussian-state rate region of two-sender EA-BGMACs.

To illustrate the union, we optimize the achievable rate tuples along $n\gg 1$ rays starting from $(R_1,R_2)=(0,0)$ by fixing $R_2/R_1=c\in\{c_1, c_2, \ldots c_n\}$. Concretely, for each $c$ we numerically maximize the norm $|\bm R|=\sqrt{1+c^2}R_1$ of rate tuples satisfying Eq.\eqref{eq:CeMAC_qinfo} over Gaussian states $\hat\phi(\bm r,\bm\theta)$. When $n\to \infty$, the convex hull of the optimal rate tuples and $(0,0)$ exactly generates the capacity region; when $n$ is small, the resolution of the convex hull falls insufficient at the corners of the capacity region. In Fig.~\ref{fig:unionregion} we choose $c$ heterogeneously such that the polar angle of data points $\varphi=\arctan \left({\frac{R_2/C_2^{\rm coh}}{R_1/C_1^{\rm coh}}}\right)$ is uniformly distributed in $[0,\pi/2)$, with $n=20$. 

In pursuit of the optimal input state, we evaluate the set of functions 
\be 
F_J(\bm r, \bm \theta)=I(A^\prime[J];B|A^\prime[J^c])_{\hat{\phi}\left(\bm r, \bm \theta\right)}, \forall J,
\label{eq:F_J}
\ee 
and solve their maxima over the parameters $\bm r$ and $\bm \theta$. Here the input $\hat{\phi}\left(\bm r, \bm \theta\right)$ is given by Eq.~\eqref{Gaussian_input}.
Note that the ranges of $\bm r$ are finite. From Eq.~\eqref{eq:Nsprime}, the condition for a bona fide state with $N_{S,k}^\prime\ge 0$ is 
\be 
-r_k^\star\le r_k\le r_k^\star\,,1\le k\le s\,,
\ee
where 
\bal 
r_k^\star=\frac{1}{2}\log\left(1 + 2 N_{S,k} + 2 \sqrt{N_{S,k} (1+ N_{S,k})}\right)\,.
\eal 

Although it seems straightforward that each user consumes all available energy $N_{{\rm S},k}$ to optimize the capacity of total rate, we provide a proof for it as below. Define the total-rate capacity 
\be 
\tilde I_U(\bm N_{\rm S})=\max_{\hat\phi\in \calH_{\bm N_{\rm S}}} I_U(\calN, \hat\phi)
\ee
where $\bm N_{\rm S}=[N_{\rm S,1},\ldots, N_{\rm S,s}]^T$, $\calH_{\bm N_{\rm S}}$ consists of all possible states with energy constraint $\expval{{\hat a}_{A_k}^\dagger a_{A_k}}\le N_{{\rm S},k}$.
Indeed, the monotonicity of $\tilde I_U(\bm N_{\rm S})$ with respect to the consumed energies $\{N_{{\rm S},k}\}_{k=1}^s$ can be obtained from the concavity of $\tilde I_U$, as shown in the proposition below.
\begin{proposition}
The EA capacity of total rate of BGMAC monotonically increases with the energy consumption of each user.
\label{proposition:monotonicity_Ns}
\end{proposition}
\begin{proof}
First, we prove the concavity of $\tilde I_U$ with respect to the energy consumption of each user., i.e., for any $\bm N_{\rm S}^2-\bm N_{\rm S}^1=[0,\ldots,\delta N_{S,k},\ldots,0]^T$, $1\le k\le s$,
\be 
\tilde I_U(\bm N_{\rm S}^1)+\tilde I_U(\bm N_{\rm S}^2)\le 2I_U[(\bm N_{\rm S}^1+\bm N_{\rm S}^2)/2]\,.
\ee
Define an alternative form of $I_J$ that allows entanglement among senders within $A[J]$ or $A[J^c]$
\bal
f_J(\calN,  {\hat\phi}_{A[J]},{\hat\phi}_{A[J^c]})
&\equiv I(A^\prime[J];BA^\prime[J^c])_{\hat{\rho}_f}
\,,
\label{eq:fJ_define}
\eal
where $\hat{\rho}_f=\calN_{A\to B}\otimes\calI_{A^\prime}(\Phi_{A\to AA^\prime}({\hat\phi}_{A[J]}\otimes{\hat\phi}_{A[J^c]}))$. From Lemma \ref{lemma:fJconc}, we have the concavity of $f_U$ with respect to each $\hat\phi_{A_k}$, $1\le k\le s$. For any $\bm N_{\rm S}^2-\bm N_{\rm S}^1=[0,\ldots,\delta N_{S,k},\ldots,0]^T$, $1\le k\le s$, consider the average state $\hat{\bar \phi}_A=(\hat\phi_{A}^{1}+\hat\phi_{A}^{2})/2$ over the optimum states under different energy constraints
\bal 
\hat\phi_{A}^{1}&=\text{argmax}_{\hat\phi_A\in \calH_{\bm N_{\rm S}^1}} I_U(\calN, \hat\phi_A)\,,\\
\hat\phi_{A}^{2}&=\text{argmax}_{\hat\phi_A\in \calH_{\bm N_{\rm S}^2}} I_U(\calN, \hat\phi_A)\,,\\
\eal
the concavity gives
\bal 
&I_U(\calN,\hat\phi_{A}^{1})+I_U(\calN,\hat\phi_{A}^{2})\\
&=f_U(\calN,\hat\phi_{A}^{1})+f_U(\calN,\hat\phi_{A}^{2})\\
&\le 2f_U(\calN,\hat{\bar \phi}_A)
\eal 
The average state $\hat{\bar \phi}_A$ is under the energy constraint $\overline {\bm N}_{\rm S}=(\bm N_{\rm S}^1+\bm N_{\rm S}^2)/2$. 
Note that ${\hat\phi}_{A}^{1}$, ${\hat\phi}_{A}^{2}$ are $s$-partite product states and they only differ in one subsystem $A_k$, therefore the average state $\hat{\bar \phi}_A$ is in an $s$-partite product state. 
Then $f_U(\calN,\hat{\bar \phi}_A)$ reduces to $ I_U(\calN,\hat{\bar \phi}_A)$ and therefore we arrive the concavity with respect to the energy
\bal 
&I_U(\calN,\bm N_{\rm S}^1)+I_U(\calN,\bm N_{\rm S}^2)\\
&=I_U(\calN,{\hat\phi}_{A}^{1})+I_U(\calN,{\hat\phi}_{A}^{2})\\
&\le 2I_U(\calN,\hat{\bar \phi}_A)\\
&\le 2I_U[(\calN,(\bm N_{\rm S}^1+\bm N_{\rm S}^2)/2]
\eal 

The concavity indicates that $\tilde I_U$ has at most one local maximum along each $N_{{\rm S},k}, 1\le k\le s$. Combined with the fact that $\lim_{|\bm N_{\rm S}|\to \infty}\tilde I_U(\bm N_{\rm S})=\infty$, it requires the maxima to locate at the unphysical infinity $|\bm N_{\rm S}|\to \infty$ thus we obtain the proposition.
\end{proof}

\section{Proof of Proposition~\ref{proposition:gaugeinv_local} } 
\label{proof:prop_gaugeinv_local}
 
In this section, we will again utilize the function $f_J$ defined in Eq.~\eqref{eq:fJ_define}. This is not a physical information rate for MAC, but it reduces to $I_J$ when ${\hat\phi}_{A[J]}\otimes{\hat\phi}_{A[J^c]}$ is in a product state among the $s$ senders.

Similar to Lemma~\ref{lemma:subadditivity_region}, $f_J$ is subadditive under the condition that $A[J^c]$'s are independent among the channel uses.
\begin{lemma}
For any two parallel channel uses acting on $A_1$, $A_2$, the conditional additivity holds for $f_J$
\bal 
&f_J(\calN, {\hat\phi}_{A^{(1)}[J]A^{(2)}[J]}, {\hat\phi}_{A^{(1)}[J^c]A^{(2)}[J^c]})\\
&\quad\le f_J(\calN, {\hat\phi}_{A^{(1)}[J]}, {\hat\phi}_{A^{(1)}[J^c]})+f_J(\calN, {\hat\phi}_{A^{(2)}[J]}, {\hat\phi}_{A^{(2)}[J^c]})\,,
\eal
under the constraint that the inputs of senders in $J^c$ for $n$ channels $A^{(1)}[J^c],A^{(2)}[J^c],\ldots,A^{(n)}[J^c]$ are mutually independent.
\label{lemma:newsubadditivity}
\end{lemma}
\begin{proof}
The same technique in Eq.~\eqref{eq:I_J_additivity} can be utilized to give
\bal
f_J&(\calN, {\hat\phi}_{A^{(1)}[J] A^{(2)}[J]}, {\hat\phi}_{A^{(1)}[J^c] A^{(2)}[J^c]})\\
&=I(A^{\prime}[J];B^{(1)}B^{(2)}A^{\prime(1)}[J^c]A^{\prime(2)}[J^c])_{\rho_f}\\
&\le I(A^{\prime}[J]B^{(2)}A^{\prime(2)}[J^c];B^{(1)}A^{\prime(1)}[J^c])_{\rho_f}\\
&\quad+I(A^{\prime}[J]B^{(1)}A^{\prime(1)}[J^c];B^{(2)}A^{\prime(2)}[J^c])_{\rho_f}\\
&=f_J(\calN, {\hat\phi}_{A^{(1)}[J] }, {\hat\phi}_{A^{(1)}[J^c] })+f_J(\calN, {\hat\phi}_{ A^{(2)}[J]}, {\hat\phi}_{ A^{(2)}[J^c]})
\eal
It leads to the subadditivity of functional $f_J$, conditioned on that $A[J^c]$'s are independent over different channel uses. 
\end{proof}

From the subadditivity in Lemma~\ref{lemma:newsubadditivity}, we obtain the following.
\begin{lemma} 
Given the energy constraint $\expval{\hat a_{A_k}^\dagger \hat a_{A_k}}\le N_{{\rm S},k}$ and a fixed Gaussian state in system $A[J^c]$, the optimum of $f_J$ is achieved by a Gaussian state in system $A[J]$
\be 
f_J(\calN, {\hat\phi}_{A[J] }, {\hat\phi}_{A[J^c] }))\le f_J(\calN,{\hat\phi}^G_{A[J] }, {\hat\phi}^G_{A[J^c] }))\,.
\ee 
\label{lemma:newGaussopt}
\end{lemma}
\begin{proof}
The proof is similar to Lemma \ref{lemma:Gaussopt}. Here we prove the Gaussian optimality for functional $f({\hat\phi}_{A[J]})=f_J(\calN,{\hat\phi}_{A[J]},{\hat\phi}^0)$ given a fixed Gaussian state ${\hat\phi}^0$ in $A[J^c]$. For $n$ parallel channel uses, we need to prove that $f$ is invariant under the Hadamard transform $H^{\otimes m}$ defined in Eq.~\eqref{eq:Hadamard_def} with $m=\log_2 n$, acting on the two input systems $A^{(1)}[J], A^{(2)}[J]$. Similarly, we have
\bal
&\calN^{\otimes n}\circ \calU^{\otimes s}:\\
&\hat{a}_{B}^{(k)}\!\!=\!\!\sum_{j\in J} w_j \!\!\left(\!\!\left(1-\delta_j\right)\!\sum_{l=1}^n \frac{H^{\otimes m}_{kl}}{\sqrt n}\hat{a}_{A_j}^{(l)}\!\!+\!\delta_j\!\sum_{l=1}^n \frac{H^{\otimes m}_{kl}}{\sqrt n}\hat{a}_{A_j}^{(l)\dagger} \!\!\right)+\\
&\qquad\!\!\sum_{j\in J^c} w_j \left(\left(1-\delta_j\right) \hat{a}_{A_j}^{(k)}+\!\delta_j \hat{a}_{A_j}^{(k)\dagger} \right)+\ldots\!\,,\\
&\calU\circ \calN^{\otimes n}:\\
&\hat{a}_{B}^{(k)}=\sum_{l=1}^n \frac{H^{\otimes m}_{kl}}{\sqrt n}\sum_{j\in J} w_j \left(\left(1-\delta_j\right)\hat{a}_{A_j}^{(l)}+\delta_j\hat{a}_{A_j}^{(l)\dagger} \right)\\
&\qquad+\sum_{l=1}^n \frac{H^{\otimes m}_{kl}}{\sqrt n}\sum_{j\in J^c} w_j \left(\left(1-\delta_j\right)\hat{a}_{A_j}^{(l)}+\delta_j\hat{a}_{A_j}^{(l)\dagger} \right) +\cdots\,.
\eal

Now consider $j\in J^c$. Note that $\{\hat a_{A_j}^{(\ell)}\}_{\ell=1}^n$ are identical Gaussian states ${\hat\phi}^0$. Define $4\times 4$ Hadamard transform $\bm H=H\oplus H=H\otimes I$ acting on the covariance matrix. Assuming the covariance matrix of the $j$th sender is $\oplus_{\ell=1}^nV_{A_j}^{(\ell)}$, we have the covariance matrix of $\hat a^{\prime(k)}_{A_j}=\sum_{l=1}^n \frac{H^{\otimes m}_{kl}}{\sqrt n}\hat{a}_{A_j}^{(l)}$:
\be 
V'=\frac{\bm H^{\otimes m/2}}{\sqrt n}\left(\oplus_{\ell=1}^nV_{A_j}^{(\ell)}\right)\left(\frac{\bm H^{\otimes m/2}}{\sqrt n}\right)^T=\oplus_{\ell=1}^n V_{A_j}^{(\ell)}
\ee
since
\be 
\bm HV_{A_j^{\ell}}\bm H^T=2V_{A_j}^{(\ell)}\,.
\ee
Meanwhile, the displacement does not affect the Von Neumann entropy of a Gaussian state. Thus $\sum_{l=1}^n \frac{H^{\otimes m}_{kl}}{\sqrt n}\hat{a}_{A_j}^{(l)}=\hat a_{A_j}^{(k)}$. Hence we have the commutation
\be 
\calN^{\otimes n}\circ \calU^{\otimes s}=\calN^{\otimes n}\circ \calU^{\otimes s}\,.
\ee
\end{proof}

Explicitly, the energy constraint can be written with respect to the CM as
\be
\Tr V_{A_k}=\mbox{Constant}\,, 1\le k\le s\,.
\ee 

We prove the concavity of $f_J$ as
\begin{lemma}
$f_J(\calN, {\hat\phi}_{A[J]}, {\hat\phi}_{A[J^c]})$ is concave individually in $A[J]$, $A[J^c]$, namely
\bal
f_J(\calN, {\hat\phi}_{A[J]}^1, {\hat\phi}_{A[J^c]})+&f_J(\calN, {\hat\phi}_{A[J]}^2, {\hat\phi}_{A[J^c]})\\
&\quad \quad \le 2f_J(\calN, \overline{\hat\phi}_{A[J]}, {\hat\phi}_{A[J^c]})\,,\\
f_J(\calN, {\hat\phi}_{A[J]}, {\hat\phi}_{A[J^c]}^1)+&f_J(\calN, {\hat\phi}_{A[J]}, {\hat\phi}_{A[J^c]}^2)\\
&\quad \quad \le 2f_J(\calN, {\hat\phi}_{A[J]}, \overline{\hat\phi}_{A[J^c]})\,.
\label{eq:fconcavity}
\eal
\label{lemma:fJconc}
\end{lemma}
This also implies that $f_J(\calN,{\hat\phi}_{A[J]}, {\hat\phi}_{A[J^c]})$ is concave individually in each input $A_k$. 

\begin{proof}

To begin with, we prove the concavity in $A[J]$. Consider two inputs 
\be 
\hat{\phi}^\ell_{A}=\hat{\phi}^\ell_{A[J]}\otimes\hat{\phi}_{A[J^c]}, \ell=1,2\,,
\ee  
which are purified with the assistance of $A^\prime$. We hope to show that the average state $\overline {\hat\phi}_A=(\hat{\phi}^1_{A[J]}+\hat{\phi}^2_{A[J]})/2\otimes \hat{\phi}_{A[J^c]} $ gives a better rate. 
To purify the average state, we add an extra qubit $R$ to $A^\prime$. 
The overall input state
\be 
\ket{\Phi_{A\to AA^\prime R}(\overline{\hat\phi}_A)}=\frac{1}{\sqrt{2}}\sum_{\ell=1}^2 \ket{{\hat\phi}^\ell_{}}\ket{\ell}_R
\ee 
Denote the corresponding output as 
\be 
\hat{\bar{\rho}}_{B A^{\prime }R}=\calN_{A\to B}\otimes \calI_{A^\prime R} (\hat{\Phi}_{A\to AA^\prime R}(\overline{\hat\phi}_A)).
\ee 
We can consider the data processing of measuring $R$, 
leading to the quantum-classical state
\ba
\hat{\rho}_{B A^{\prime}R}
&=&\frac{1}{2}\left[\sum_{\ell=1}^2 \state{\ell}_R\otimes\hat{\rho}^\ell_{BA^\prime} \right],
\label{QCstate}
\ea 
where the output states for each input
\be 
\hat{\rho}^\ell_{BA^\prime}=\calN_{A\to B}\otimes \calI_{A^\prime} (\hat{\phi}^\ell_{}).
\ee 
This measurement will not increase the mutual information from data processing inequality. Since $RA'[J]$ purifies $A[J]$, we have
\begin{align}  
f_J(\calN, \overline{\hat\phi}_{A[J]}, \hat\phi_{A[J^c]})
&=I(A^\prime[J]R;BA^\prime[J^c])_{\hat{\bar{\rho}}}\\
&\ge  I(A^\prime[J]R;BA^\prime[J^c])_{\hat{\rho}}
\label{Ifinalstep1}
\\
&\ge I(A^\prime[J];BA^\prime[J^c]|R)_{\hat{\rho}}
\label{Ifinalstep2}
\\ 
&=\frac{1}{2}\sum_{\ell=1}^2I(A^\prime[J];BA^\prime[J^c])_{\hat{\rho}^\ell},
\label{Ifinalstep3}
\end{align}  
where \eqref{Ifinalstep1} follows from data processing inequality;
Eq.~\eqref{Ifinalstep2} follows from 
\begin{align}
&I(A^\prime[J]R;BA^\prime[J^c])_{\hat{\rho}}
\nonumber
\\
&=S(A^\prime[J]R)_{\hat{\rho}}+S(BA^\prime[J^c])_{\hat{\rho}}-S(BA^\prime R)_{\hat{\rho}}
\\
&=S(A^\prime[J]|R)_{\hat{\rho}}+S(BA^\prime[J^c])_{\hat{\rho}}-S(BA^\prime |R)_{\hat{\rho}}
\\
&\ge S(A^\prime[J]|R)_{\hat{\rho}}+S(BA^\prime[J^c]|R)_{\hat{\rho}}-S(BA^\prime |R)_{\hat{\rho}}
\label{step:concave}
\\
&=I(A^\prime[J];BA^\prime[J^c]|R)_{\hat{\rho}},
\end{align} 
where in the \eqref{step:concave} we utilized concavity of entropy; \eqref{Ifinalstep3} follows from the quantum classical state Eq.~\eqref{QCstate}. The equality of \eqref{step:concave} holds if and only if
\be 
\hat{\rho}^1_{BA^\prime[J^c]}=\hat{\rho}^2_{BA^\prime[J^c]},
\ee 
which requires the EA $BA^\prime[J^c]$ to be in the same state.
Note that
\bal  
&[f_J(\calN,  {\hat\phi}_{A[J]}^1,{\hat\phi}_{A[J^c]})+f_J(\calN,  {\hat\phi}_{A[J]}^2,{\hat\phi}_{A[J^c]})]/2\\
&=\frac{1}{2}\sum_{\ell=1}^2I(A^\prime[J];BA^\prime[J^c])_{\hat{\rho}^\ell}\,,
\eal
we have proven the concavity
\bal 
&f_J(\calN, \overline{\hat\phi}_{A[J]}, \hat\phi_{A[J^c]})\\
&\ge[f_J(\calN,  \phi_{A[J]}^1,\phi_{A[J^c]})+f_J(\calN,  \phi_{A[J]}^2,\phi_{A[J^c]})]/2\,.
\eal
The equality holds requires a necessary condition that the $BA^\prime[J^c]$ is in the same state for the two inputs.

The same proof works for the concavity in $A[J^c]$. Now $RA'[J^c]$ purifies $A[J^c]$, 
\begin{align}  
I(A^\prime[J];BA^\prime[J^c]R)_{\hat{\bar{\rho}}}
&\ge  I(A^\prime[J];BA^\prime[J^c]R)_{\hat{\rho}}
\label{step:ineq:case2}
\\
&= I(A^\prime[J];BA^\prime[J^c]|R)_{\hat{\rho}}
\label{step:equal_indep}
\\ 
&=\frac{1}{2}\sum_{\ell=1}^2I(A^\prime[J];BA^\prime[J^c])_{\hat{\rho}^\ell}. 
\end{align}  
\eqref{step:equal_indep} holds because $R$ is independent with $A'[J]$.

Therefore we have proven the concavity in $A[J]$ and $A[J^c]$ individually.
\end{proof}

Recall the phase-insensitive condition in Eq.~\eqref{eq:phase_invariance} of the main paper. We denote the channel corresponding to the phase rotation as $\calR_{ \theta}^s\left(\hat{\rho}_{A}\right)=\hat{R}\left(\bm \theta\right) \hat{\rho}_{A} \hat{R}^\dagger\left(\bm \theta\right)$, with $\theta_k=(-1)^{\delta_k} \theta$. Similarly, on the output side we denote the channel as $\calR_\theta$. The symmetry of the channel in Eq.~\eqref{eq:phase_invariance} can be equivalently written as
\be 
\calR_{-\theta}\circ \calN \circ \calR_{\theta}^s=\calN.
\label{eq:phase_invariance_supp}
\ee

First we prove for the squeezing on $A[J]$.
Consider a bunch of single-mode squeezing operations $\otimes_{k\in J}S_{A_k}$, with squeezing parameter $r_k$ at the phase angle $\theta_k$, on modes in system $A[J]$ of the product thermal state ${\hat\phi}_{A}$, which produces 
\be 
{\hat\phi}^\prime_{A}=[\otimes_{k\in J}S_{A_k}(r_k,\theta_k)]\otimes \calI_{A[J^c]}({\hat\phi}_{}).
\ee  
For the consistency we define $r_k=0$ for $k\in J^c$. The mean photon number of the $k$th mode of the thermal state is constrainded by $\expval{\hat a^\dagger \hat a}=\cosh(2r_k) N_{S,k}+\sinh^2(r_k)$.
The functional of interest $I_J(\calN,{\hat\phi}_A)$ only depends on the reduced state ${\hat\phi}^\prime_{A}$, which is a product of $|J|$ single-mode squeezed thermal states and $|J^c|$ single-mode thermal states.

We generate the average state using $\calR_{\theta}^s$ with $\theta=0,\pi/2$, the phase rotation does not change $A[J^c]$
\bal
\overline{{\hat\phi}}^\prime_A&=\frac{1}{2}\left[{{\hat\phi}^\prime}_A+ \calR_{\pi/2}^s({{\hat\phi}^\prime}_A)\right]\\
&=\frac{1}{2}\left[{{\hat\phi}^\prime}_{A[J]}+ \calR_{\pi/2}^s({{\hat\phi}^\prime}_{A[J]})\right]\otimes {\hat\phi}^\prime_{A[J^c]}\,.
\label{eq:averagestate_app}
\eal
Here we have utilized the fact that ${\hat\phi}^\prime_{A[J^c]}$ is in a thermal state, invariant under any phase rotations.
For the two states with $\theta=0,\pi/2$, the covariance matrix is
\be 
V_{A}^\prime(\theta_k,\theta)=\oplus_{k=1}^s V_{A_k}^{\prime r_k}(\theta_k,\theta)
\ee
with
\begin{widetext}
\be
V_{A_k}^{\prime r_k}(\theta_k,\theta)=
\left(
\begin{array}{cc}
 \frac{1}{2} e^{-2 r_k} \left(2 e^{4 r_k} \cos ^2\left(\theta_k\right)+\cos \left(2 \left(\theta-\theta_k\right)\right)+1\right) \left(2 N_{S,k}+1\right) & \sin \left(2 \left(\theta-\theta_k\right)\right) \sinh \left(2 r_k\right) \left(2 N_{S,k}+1\right) \\
 \sin \left(2 \left(\theta-\theta_k\right)\right) \sinh \left(2 r_k\right) \left(2 N_{S,k}+1\right) & e^{-2 r_k} \left(\cos ^2\left(\theta_k\right)+e^{4 r_k} \sin ^2\left(\theta_k\right)\right) \left(2 N_{S,k}+1\right) \\
\end{array}
\right)
\ee
\end{widetext}
Indeed, the covariance matrix $\overline V_A$ of $\overline{\hat\phi}^\prime_{A}$ is block-diagonal, as an average of two block-diagonal CMs $V_A^{0}$, $V_A^{\pi/2}$. Explicitly, the CM of $\overline{\hat\phi}^\prime_{A[J]}$ is 
\bal
\overline V_{A}&=\oplus_{k=1}^s\overline V_{A_k}^{r_k}\,,\\
\overline V_{A_k}^{r_k}&=
\left(
\begin{array}{cc}
 \cosh \left(2 r_k\right) \left(2 N_{S,k}+1\right) & 0 \\
 0 & \cosh \left(2 r_k\right) \left(2 N_{S,k}+1\right) \\
\end{array}
\right)
\,.\\
\eal
We see that the squeezing phases $\theta_k$'s are irrelevant.
The phase rotation $\calR^s_{\pi/2}$ does not change energy, thus the average state still suffices the energy constraint
\be 
\Tr (\hat a^\dagger \hat a\, {\hat{\bar \phi}}^\prime_{A_k})=\cosh(2r_k) N_{S,k}+\sinh^2(r_k)\,.
\ee

By individual concavity in Lemma~\ref{lemma:fJconc}, 
\begin{align}
f_J(\calN,\overline{\hat\phi}^\prime_{A[J]},{\hat\phi}^\prime_{A[J^c]})&\ge f_J(\calN,{\hat\phi}^\prime_{A[J]},{\hat\phi}^\prime_{A[J^c]})
\nonumber
\\
&=I_J(\calN,\{{\hat\phi}^\prime_{A_i}\}_{i=1}^s)\,.
\label{fJ_step1}
\end{align}  
Note that the CM of $\overline {\hat\phi}_{A}^\prime$ is block-diagonal among the $s$ senders. Thus the Gaussian state associated with the same CM is an $s$-partite product state, which turns out to be $\hat\phi_A^{\rm TMSV}$, the reduced state of product TMSV $\otimes_{k=1}^s\ket{\zeta(\cosh (2r_k) N_{S,k})}_{A_kA^\prime_k}$. By Gaussian extremality in Lemma~\ref{lemma:newGaussopt}, we have
\bal 
f_J(\calN,\overline{\hat\phi}^\prime_{A[J]},{\hat\phi}^\prime_{A[J^c]})&\le f_J(\calN,\hat{\phi}_{A[J]}^{\rm TMSV},{\hat\phi}^\prime_{A[J^c]})\\
&=f_J(\calN,\hat{\phi}_{A[J]}^{\rm TMSV},\hat{\phi}_{A[J^c]}^{\rm TMSV})\,.
\label{fJ_step2}
\eal
Note that $\hat\phi_A^{\rm TMSV}$ is an $s$-partite product state
\be 
I_J(\calN,\{\hat{\phi}_{A_i}^{\rm TMSV}\}_{i=1}^s)=f_J(\calN,\hat{\phi}_{A[J]}^{\rm TMSV},\hat{\phi}_{A[J^c]})\,.
\label{fJ_step3}
\ee 
Combining Ineqs.~\eqref{fJ_step1}, \eqref{fJ_step2}, and~\eqref{fJ_step3}, we arrive at
\be 
I_J(\calN,\{\hat{\phi}_{A_i}^{\rm TMSV}\}_{i=1}^s) \ge I_J(\calN,\{{\hat\phi}^\prime_{A_i}\}_{i=1}^s)\,,
\ee 
which proves the result.\\

The same proof works for the squeezing on $A[J^c]$
\be 
{\hat\phi}^\prime_{AA^\prime}=[\otimes_{k\in J^c}S_{A_k}(r_k,\theta_k)]\otimes \calI_{A[J]}({\hat\phi}_{AA^\prime})\,.
\ee 
By similar procedure we have
\bal 
I_J(\calN,\{\hat{\phi}_{A_i}^{\rm TMSV}\}_{i=1}^s)&=f_J(\calN,\hat{\phi}_{A[J]}^{\rm TMSV},\hat{\phi}_{A[J^c]}^{\rm TMSV})\\
&\ge f_J(\calN,{\hat\phi}^\prime_{A[J]},\overline{\hat\phi}^\prime_{A[J^c]})\\
&\ge f_J(\calN,{\hat\phi}^\prime_{A[J]},{\hat\phi}^\prime_{A[J^c]})\\
&=I_J(\calN,\{{\hat\phi}^\prime_{A_i}\}_{i=1}^s)\,.
\label{eq:prop_gaugeinv_Jc_app}
\eal

\section{Unravelling the memory in interference MAC$\rm s$}
\label{sec:Methods_unravel}

\begin{figure}[t]
    \centering
    \includegraphics[width=0.5\textwidth]{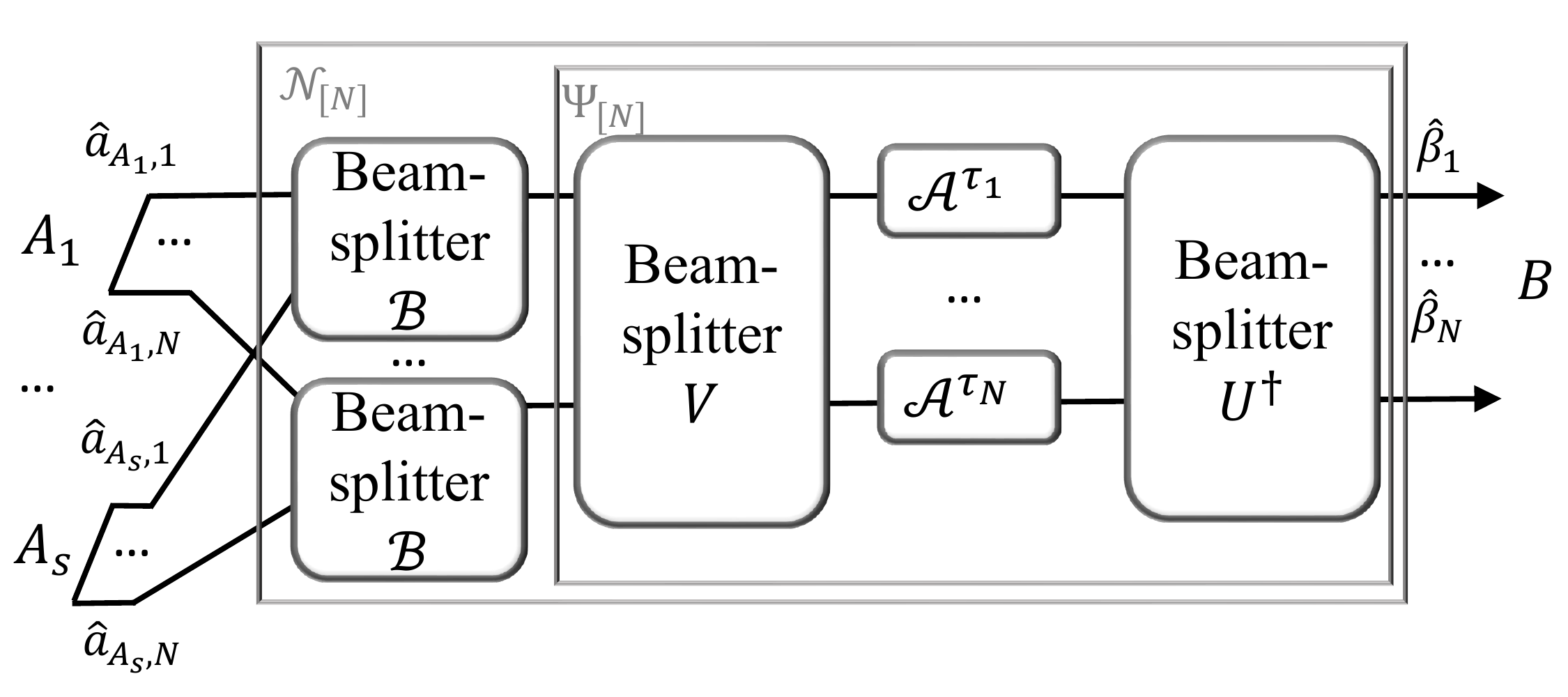}
    \caption{Memory unravelling of an $N$-fold memory interference MAC $\bm\calN_{[N]}=\bm\Psi_{[N]}\circ\calB^{\otimes N}$, which consists of $N$ uses of $s$-input-one-output beamsplitters $\calB$ concatenated with an $N$-fold memory BGC $\bm\Psi_{[N]}$. $\bm\Psi_{[N]}$ can be interpreted by a product of $N$ single-mode amplifiers $\otimes_{d=1}^N\calA^{\tau_d}$ (or lossy channel $\calL^{\tau_d}$ for $\tau_d\le 1$), sandwiched by two beamsplitter arrays fulfilling Bogoliubov transform $V$ and $U^\dagger$.}
    \label{fig:unravel_memory}
\end{figure}
We investigate the memory interference channel, which consists of $N$ parallel uses of $s$-input-one-output beamsplitter $\calB$ and the memory is invoked by an $N$-input-$N$-output memory BGC $\bm\Psi_{[N]}$, as shown in Fig.~\ref{fig:unravel_memory}. The general input-output relation of Eq.~\eqref{eq:BGMAC_phaseinsensitive_multimode}, when applied to the memory interference channel, can be denoted in a matrix form,
\be 
\bm\calN_{[N]}: \bm{\hat a}_B= W\bm{\hat a}_A+\bm{\hat \xi}\,,
\ee
where $\bm{\hat a}_A=\oplus_{k=1}^s[\hat a_{A_k,1},\ldots, \hat a_{A_k,N}]^T$ consists of $sN$ input modes, $\bm{\hat a}_B=[\hat a_B^1,\ldots, \hat a_B^N]^T$ consists of $N$ output modes and similarly $\bm{\hat \xi}$ denotes all the noise modes; the entries of the $N\times N$ transition matrix $W$ are defined by $W_{dh}=w_{d,h}$. The overall Bogoliubov transform of $\bm\calN_{[N]}$ consists of an $N\times sN$ 
beamsplitter matrix $F$ accounting for $\calB^{\otimes N}$ and an $N\times N$ transition matrix $W'$ accounting for $\bm\Psi_{[N]}$ as
$ 
W=W'F\,,
$
where
\be 
F=
\bp 
\sqrt{\eta_{1}} & 0 & 0 & \ldots &\sqrt{\eta_{s}} & 0 & 0\\
0 & \ddots & 0    & \ldots & 0 & \ddots & 0\\
0 & 0 & \sqrt{\eta_{1}} & \ldots & 0 & 0 & \sqrt{\eta_{s}}
\ep\,,
\ee
with the weights $\eta_k$ normalized as $\sum_{k=1}^s\eta_k=1$. 

It is known that an $N$-fold memory BGC $\bm\Psi_{[N]}$ can be unravelled into $N$ single-mode BGCs by the singular value decomposition (SVD) \cite{de2014classical}
\be 
UW'V^\dagger=D\,,
\ee 
where $U$, $V$ are $N\times N$ unitary matrices, $D={\rm diag}(\sqrt{\tau_1},\ldots, \sqrt{\tau_N})$. Concretely, the SVD unravels $\bm\Psi_{[N]}$ into a concatenation of an $N$-port beamsplitter fulfilling the Bogoliubov transform $V$, a combination of $N$ different single-mode BGCs fulfilling the diagonal matrix $D$, and finally an $N$-port beamsplitter fulfilling $U^\dagger$, as shown in Fig.~\ref{fig:unravel_memory}. By absorbing $V$, $U$ into the input and output modes, one can reduce $\bm\Psi_{[N]}$ to $N$ individual single-mode BGCs $\{\Psi_d\}_{d=1}^N$ that implements the diagonal matrix $D$. To absorb $V$ into the initial input modes $\bm{\hat a_A}$, note that $F$ commutes with $V$ up to an extension of dimension: 
\bal 
VF=FV_{\rm ext}\,, \mbox{ with }
V_{\rm ext}\equiv\oplus_{k=1}^s V=
\bp 
V & 0 & 0\\
0 & \ddots & 0 \\
0 & 0 &V
\ep\,.
\eal
Thus the overall Bogoliubov transform $W$ turns to $U^\dagger D F V_{\rm ext}$.
Formally, defined over the new modes $\bm{\hat\alpha}=V_{\rm ext}\bm{\hat a}_A$, $\bm{\hat\beta}=U\bm{\hat a}_B$, the overall channel $\bm\calN_{[N]}$ reduces to $N$ single-mode interference MACs $\{\calN_d\}_{d=1}^N$ mapping $\bm{\hat\alpha}$ to $\bm{\hat\beta}$ individually: 
\be 
\calN_d:\hat\beta_d=\sum_{d'=1}^{N}\sum_{h=1}^{sN} D_{dd'} F_{d'h}\hat\alpha_h+\sum_{d'=1}^N U_{dd'} \hat \xi_{d'}\,,1\le d\le N\,.
\label{eq:decomposedN_full}
\ee
It is easy to check that the modes $\bm{\hat\alpha}$ and $\bm{\hat\beta}$ satisfy the canonical commutation relation, and that the energy constraint
is automatically inherited to the new modes,
$
\sum_{n_k=1}^N\expval{\hat \alpha_{h(k,n_k)}^\dagger \hat\alpha_{h(k,n_k)}}=\sum_{n_k=1}^N\expval{\hat a_{A_k,n_k}^\dagger \hat a_{A_k,n_k}}\,.
$

Finally, we note that $\{\tau_d\}_{d=1}^N$ can be alternatively obtained from the $N\times N$ commutation matrix 
\be 
\Omega_{dd'}\equiv[\sum_{h=1}^{sN} F_{dh}\hat\alpha_h,\sum_{h=1}^{sN} F_{d'h}\hat\alpha_h^\dagger]=\sum_j W'_{dj}W_{d'j}^{\prime *}\,,
\ee
of which the eigenvalues are $\tau_1,\tau_2,\ldots,\tau_N$. For the causal memory thermal-loss channel Eq.~\eqref{eq:causalmemory_loss} parameterized by $\epsilon$ and $\gamma$, when the memory mode is inaccessible for either the sender or the receiver, there is a neat formula for the commutation matrix \cite{lupo2010capacities}
\be 
\Omega_{kk'}=\delta_{kk'}-(1-\gamma_{\min\{k,k'\}})\sqrt{\epsilon\gamma}^{|k-k'|}\,,
\label{eq:Omega_causalmemory}
\ee
where $\gamma_k\equiv \gamma+[1-(\epsilon\gamma)^{k-1}]{\epsilon(1-\gamma)^2}/{(1-\epsilon\gamma)}$.

\section{Proof of Proposition \ref{proposition:multi-mode}}
\label{proof_multi_mode}

Any multi-mode zero-mean Gaussian state can be completely characterized by its covariance matrix. We prove the optimality of correlated thermal state by proving that the rates are never improved by phase-sensitive correlation or single-mode squeezing.

We are interested in phase-insensitive memory BGMACs, whose information rate $I_J$ is preserved by the gauge transform $\calR_\theta^s$ in Eq.~\eqref{eq:phase_invariance_supp}. As it can be cancelled by applying $\calR_{-\theta}$ on each of the $N$ output modes in an $N$-fold BGMAC
\be
\calR_{-\theta}^{\otimes N}\circ \calN\circ \calR_\theta^s=\calN
\,.
\label{eq:symmetry_memory}
\ee

Phase-sensitivity can be characterized by the elements in the covariance matrix. Without loss of generality, we consider the $N_k$-mode covariance matrix for the $k$th sender, with respect to the annihilation operators,
\be 
V=
\bp 
\bm D_{1} &\ldots & \bm C_{1,N_k}\\
\ldots& \ldots& \ldots\\
\bm C_{1,N_k}^\dagger&\ldots & \bm D_{N_k}\\
\ep 
\ee
where the block matrices $\bm D_{i}, \bm C_{j,j'}$ are the $2\times 2$ covariance matrices of the $i$ mode, and the cross correlation matrix between the $j$th and the $j'$th mode respectively. With respect to the gauge transform $\calR_\theta^s$, each $\bm C_{j,j'}$ can be divided into the phase-sensitive correlation $c_s$ and the phase-insensitive correlation $c_i$
\be 
\bm C_{j,j'}=
\bp 
 c_{\rm i} &  c_{\rm s}\\
 c_{\rm s}* & c_{\rm i}^*
\ep\,.
\ee

Now consider an $N_k$-mode Gaussian state ${\hat\phi}^\prime_{A_k}$ for the $k$th sender, generated from a TMSV by two types of unitaries: (1) $N_k$ single-mode squeezers $S_{A_k}(\bm r_k,\bm \theta_k)=\otimes_{n_k=1}^{N_k}S_{A_k}^{n_k}(r_k^{n_k},\theta_k^{n_k})$ if $k\in J$; (2) phase-insensitive multi-mode Gaussian unitaries. The overall state is ${\hat\phi}^\prime_A=\otimes_{k=1}^s {\hat\phi}^\prime_{A_k}$. Denote the state satisfying the proposition as $\hat\phi_A^{\rm th}$, which is in a correlated thermal state per sender. Below we prove that the rate of $\hat\phi_A^{\rm th}$ outperforms any state ${\hat\phi}^\prime_A$.
Denote the average state as $\hat{\bar \phi}^\prime_A=\int d\theta/2\pi  \calR_{\theta}^s({\hat\phi}^\prime_A)$. 
The phase transformed state $R_\theta^s({\hat\phi}^\prime_A)$ has the covariance matrix $V_{A}^{\prime\bm r}(\bm \theta_1,\ldots,\bm \theta_s,\theta)=\oplus_{k=1}^s V_{A_k}^{\prime \bm r_k}(\bm \theta_k,\theta)$, with $r_k=0, \theta_k=0$ for $k\in J^c$, each block defined by
\be
V_{A_k}^{\prime \bm r_k}(\bm \theta_k,\theta)=
\bp 
\bm D_{1}^{r_k^{1}}(\theta_k^{1},\theta) &\ldots & \bm C_{1,N_k}\\
\ldots& \ldots& \ldots\\
\bm C_{1,N_k}^\dagger&\ldots & \bm D_{N_k}^{r_k^{N_k}}(\theta_k^{N_k},\theta)\\
\ep 
\,,
\ee
where
\begin{widetext}
\bal
&\bm D_{n_k}(\theta_k^{n_k},\theta)\\
&=\left(
\begin{array}{cc}
 \cosh \left(2 r_k^{n_k}\right) N_{S,k}^{n_k}+\cosh ^2\left(r_k^{n_k}\right)  & -e^{-2 i \theta } \cosh \left(r_k^{n_k}\right) \sinh \left(r_k^{n_k}\right) \left(2 N_{S,k}^{n_k}+1\right) \\
 -e^{2 i \theta } \cosh \left(r_k^{n_k}\right) \sinh \left(r_k^{n_k}\right) \left(2 N_{S,k}^{n_k}+1\right) & \cosh \left(2 r_k^{n_k}\right) N_{S,k}^{n_k}+\sinh ^2\left(r_k^{n_k}\right) \\
\end{array}
\right)\,,
\eal
under the energy constraint $\sum_{n_k} \cosh \left(2 r_k^{n_k}\right) N_{S,k}^{n_k}+\sinh ^2\left(r_k^{n_k}\right) =N_{S,k}$,
and 
\be 
\bm C_{n_k,n_k'}=
\left(
\begin{array}{cc}
 c_i &c_s e^{-2 i \theta } \\
 c_s^* e^{2 i \theta } &c_i^*\\
\end{array}
\right)\,,
\ee
with
\bal 
c_i&=e^{-2 i \theta_k^{n_k}} \sinh \left(r_k^{n_k'}\right) \left(\sinh \left(r_k^{n_k}\right) c_i-\cosh \left(r_k^{n_k}\right) c_s\right)+\cosh \left(r_k^{n_k'}\right) \left(\cosh \left(r_k^{n_k}\right) c_i-\sinh \left(r_k^{n_k}\right) c_s\right)\,,\\
c_s&= \cosh \left(r_k^{n_k'}\right) \left(\cosh \left(r_k^{n_k}\right) c_s-\sinh \left(r_k^{n_k}\right) c_i\right)-e^{2 i \theta _0} \sinh \left(r_k^{n_k'}\right) \left(\cosh \left(r_k^{n_k}\right) c_i-\sinh \left(r_k^{n_k}\right) c_s\right)\,.
\eal
Thus the average state has the covariance matrix with
\be 
\bar{ \bm D}_{n_k}(\theta_k^{n_k},\theta)=
\left(
\begin{array}{cc}
  \cosh \left(2 r_k^{n_k}\right) N_{S,k}^{n_k}+\cosh ^2\left(r_k^{n_k}\right)  & 0 \\
0 & \cosh \left(2 r_k^{n_k}\right) N_{S,k}^{n_k}+\sinh ^2\left(r_k^{n_k}\right) \\
\end{array}
\right)
\ee
\end{widetext}
and
\be 
\bar{\bm C}_{n_k,n_k'}=
\left(
\begin{array}{cc}
c_i & 0\\
0 & c_i^* \\
\end{array}
\right),
\ee
which is equal to the covariance matrix of an $n_k$-partite correlated thermal state $\hat\phi_{A_k}^{\rm th}$ under the same energy constraint. 

Note that $A[J^c]$ does not change under the gauge transform $R_\theta^s$. By a similar procedure as Ineqs.~\eqref{fJ_step1}\eqref{fJ_step2}\eqref{fJ_step3} in Appendix \ref{proof:prop_gaugeinv_local}, using the gauge symmetry of $f_J$ under $R_\theta^s$, the concavity of $f_J$ gives
\bal 
I_J(\calN,\{\hat{\phi}_{A_i}^{\rm th}\}_{i=1}^s)&=f_J(\calN,\hat{\phi}_{A[J]}^{\rm th},\hat{\phi}_{A[J^c]}^{\rm th})\\
&\ge f_J(\calN,\overline{\hat\phi}^\prime_{A[J]},{\hat\phi}^\prime_{A[J^c]})\\
&\ge f_J(\calN,{\hat\phi}^\prime_{A[J]},{\hat\phi}^\prime_{A[J^c]})\\
&=I_J(\calN,\{{\hat\phi}^\prime_{A_i}\}_{i=1}^s)\,.
\eal 

\section{Supplemental numerical evidences for the optimality of TMSV}
\label{app:supp_numerics}

In this section we provide more evidences to verify the optimality of TMSV for the rates $F_J$'s defined in Eq.~\eqref{eq:F_J}. We focus on the two-sender case, where we have the gauge symmetry on the global phase $\theta_1+\theta_2$. Thus only the relative phase $\theta=\theta_2-\theta_1$ matters. For the relative phase $\theta$, it is noteworthy that there is a degeneracy
\be 
F_J(r_1,r_2,\theta)=F_J(r_1,-r_2,\theta+\pi/2)\,,
\ee
since $S(r_2,\theta_2)=S(-r_2,\theta_2+\pi/2)$. Due to the degeneracy, it is sufficient to examine $\theta\in[0,\pi/2)$.

Fig.~\ref{fig:unionregion} shows the Gaussian-state capacity region of four channels of which $\Psi$'s cover all the four classes of the phase-insensitive BGC. We see that the EA Gaussian-state capacity region surpasses even the outer bound of the unassisted capacity region considerably. Fig.~\ref{fig:unionregion_Ns=0.001} shows more significant EA advantage under weaker illumination.

Fig.~\ref{fig:numeric_I_loss} shows the trend of the boundaries $F_J, J=\{1\},\{2\}$ when varying squeezing parameters $r_1,r_2$ for sender 1,2, and their relative phase $\theta$ for thermal-loss MAC. We also plotted the gradient vector $\bm\nabla_{\bm r} F_J(\bm r,\theta)$ in blue arrows. The gradient vanishes at $(0,0)$. We see that the optima coincide at the zero point, i.e. TMSV, in all the presented parameter settings. In addition, we have proven Theorem \ref{theorem: EA_MAC_main_TMSV} which indicates that TMSV is the optimum for the boundary $F_{\{1,2\}}$ of the total rate. Note that the capacity region is the convex hull of pentagons formed by the edges $R_1=F_{1}$, $R_2=F_{2}$, $R_1+R_2=F_{\{1,2\}}$ and $R_1=0$, $R_2=0$, and we have shown that all the pentagons are subsets of the TMSV region. Thus TMSV achieves the capacity region in this case.

Fig.~\ref{fig:numeric_I_amp},~\ref{fig:numeric_I_conj},~\ref{fig:numeric_I_AWGN} show the trends of the boundaries $F_{\{1\}},F_{\{2\}}$ for amplifier MAC, conjugate amplifier MAC and AWGN MAC. The layouts are similar to Fig.~\ref{fig:numeric_I_loss}. Again, the optima coincide at the zero point, where the gradient vanishes. Combining with Theorem \ref{theorem: EA_MAC_main_TMSV}, we obtain the optimality of TMSV in these cases.

It turns out that for each case we examined, the global optima of $F_J$ for all $J$ coincide at $\bm r=0$, i.e. the TMSV state. In these cases it concludes that the TMSV state achieves the capacity region.

\begin{figure*}[t]
    \centering
    \subfigure{
    \includegraphics[width=0.95\textwidth]{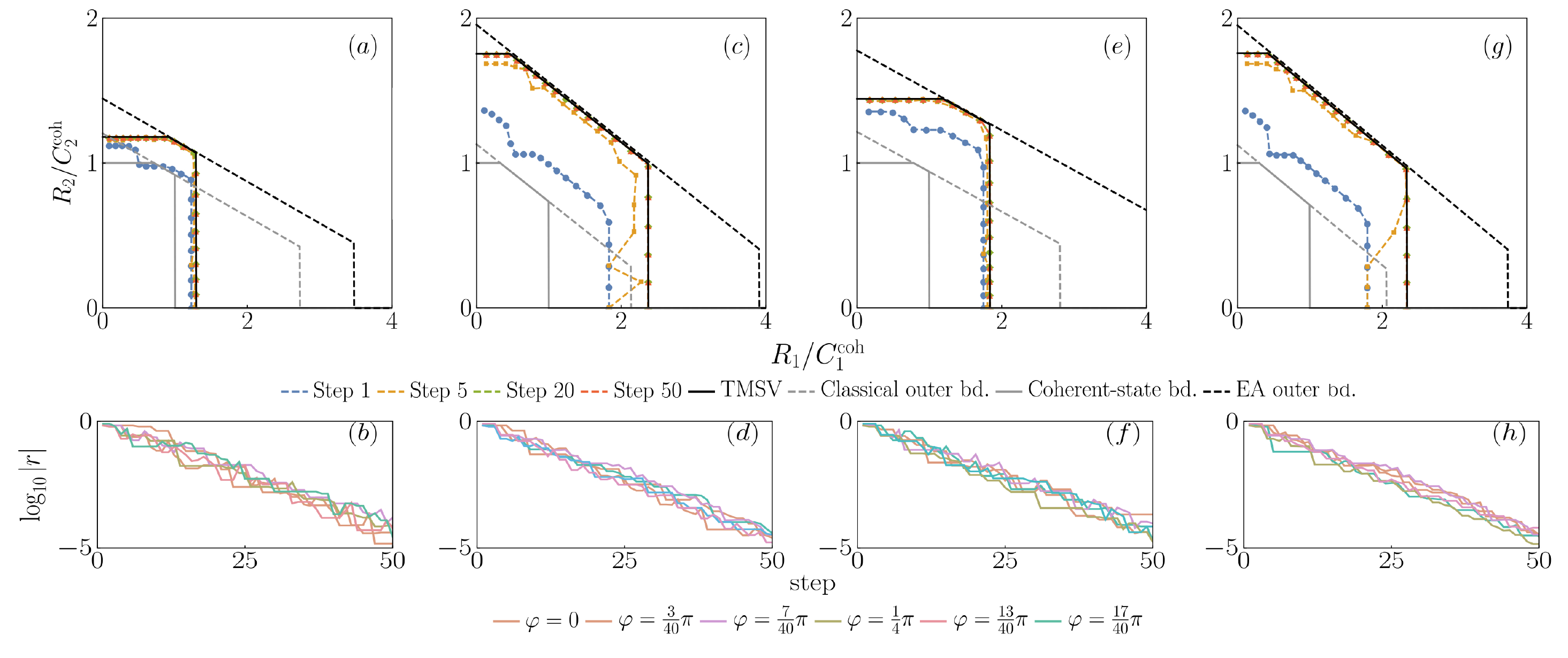}
    }
    \caption{Optimization of the Gaussian-state rate regions of two-sender interference BGMACs using bright illumination. Four classes of $\Psi$ are presented: (a)(b) Thermal-loss channel, $N_{{\rm S}, 1}=1,N_{{\rm S}, 2}=2,\eta_1=1/3,\eta_2=2/3, |w|^2=0.1, N_{\rm B}=0.1$; (c)(d) amplifier channel, $N_{{\rm S}, 1}=1,N_{{\rm S}, 2}=2,\eta_1=1/3,\eta_2=2/3, |\bm w|^2=1.1$,$N_{\rm B}=|\bm w|^2-1+0.1=0.2$; (e)(f) conjugate amplifier channel, $N_{{\rm S}, 1}=1,N_{{\rm S}, 2}=2,\eta_1=1/3,\eta_2=2/3, |\bm w|^2=0.1$,$N_{\rm B}=|\bm w|^2-1+0.1=0.2$; (g)(h) AWGN channel, $N_{{\rm S}, 1}=1,N_{{\rm S}, 2}=2,\eta_1=1/3,\eta_2=2/3, N_{\rm B}=0.1$. (a)(c)(e)(g) The evolution trend of the rate region. Colored according to the progress step of the numerical optimization. (b)(d)(f)(h) The evolution trend of $|r|=\sqrt{r_1^2+r_2^2}$ versus the progress step. Plotted with six typical $\varphi$'s, $\varphi=\tan^{-1} \left(\frac{R_2/C_2^{\rm coh}}{R_1/C_1^{\rm coh}}\right)$.  See Appendix~\ref{app:evalmethod} for the evaluation method.}
    \label{fig:unionregion}
\end{figure*}
\begin{figure*}[t]
    \centering
    \subfigure{
    \includegraphics[width=0.95\textwidth]{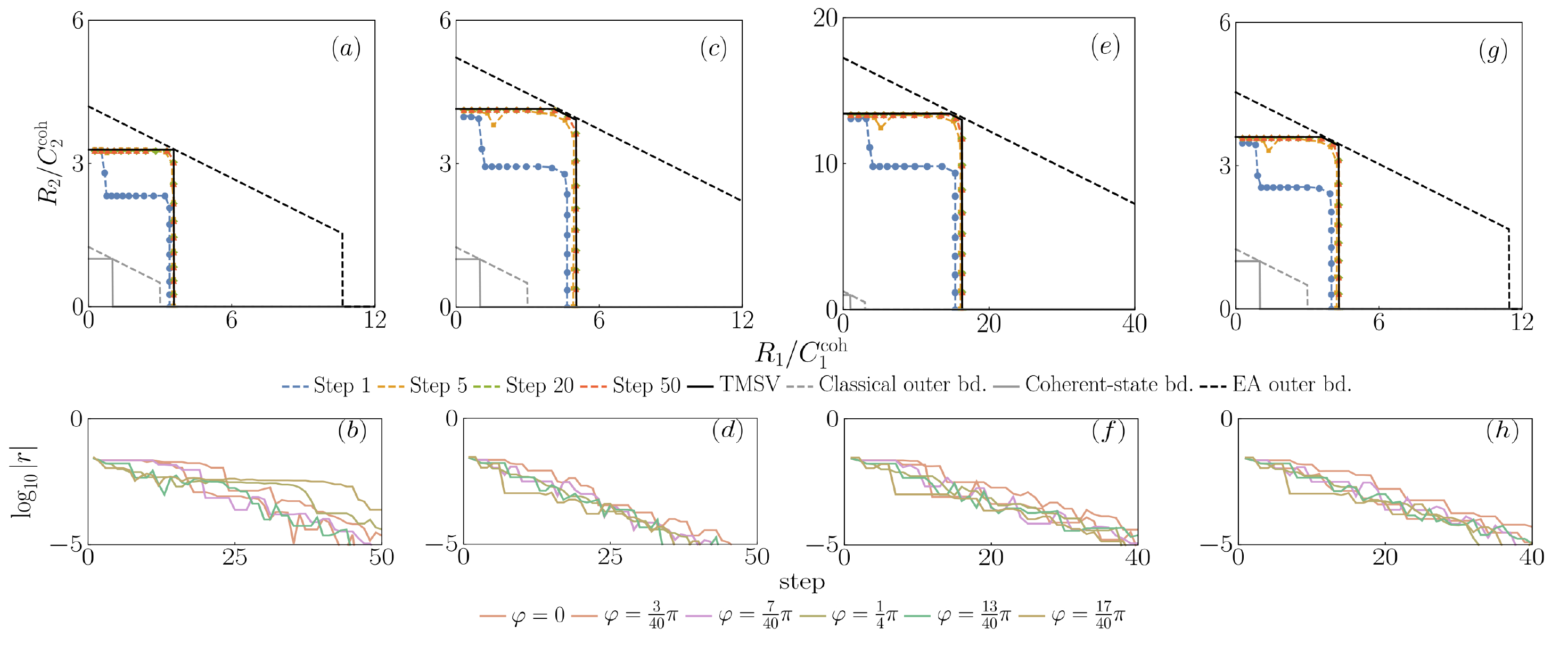}
    }
    \caption{Optimization of the Gaussian-state rate regions of two-sender interference BGMACs using weak illumination. Four classes of $\Psi$ are presented: (a)(b) Thermal-loss channel, $N_{{\rm S}, 1}=10^{-3},N_{{\rm S}, 2}=2\times 10^{-3},\eta_1=1/3,\eta_2=2/3, |w|^2=0.1, N_{\rm B}=0.1$; (c)(d) amplifier channel, $N_{{\rm S}, 1}=10^{-3},N_{{\rm S}, 2}=2\times 10^{-3},\eta_1=1/3,\eta_2=2/3, |\bm w|^2=1.1$,$N_{\rm B}=|\bm w|^2-1+0.1=0.2$; (e)(f) conjugate amplifier channel, $N_{{\rm S}, 1}=10^{-3},N_{{\rm S}, 2}=2\times 10^{-3},\eta_1=1/3,\eta_2=2/3, |\bm w|^2=0.1$, $N_{\rm B}=|\bm w|^2-1+0.1=0.2$; (g)(h) AWGN channel, $N_{{\rm S}, 1}=10^{-3},N_{{\rm S}, 2}=2\times 10^{-3},\eta_1=1/3,\eta_2=2/3, N_{\rm B}=0.1$. (a)(c)(e)(g) The evolution trend of the rate region. Colored according to the progress step of the numerical optimization. (b)(d)(f)(h) The evolution trend of $|r|=\sqrt{r_1^2+r_2^2}$ versus the progress step. Plotted with six typical $\varphi$'s, $\varphi=\tan^{-1} \left(\frac{R_2/C_2^{\rm coh}}{R_1/C_1^{\rm coh}}\right)$. See Appendix~\ref{app:evalmethod} for the evaluation method.  }
    \label{fig:unionregion_Ns=0.001}
\end{figure*}

\begin{figure*}[tbp]
    \centering
    \includegraphics[width=0.95\textwidth]{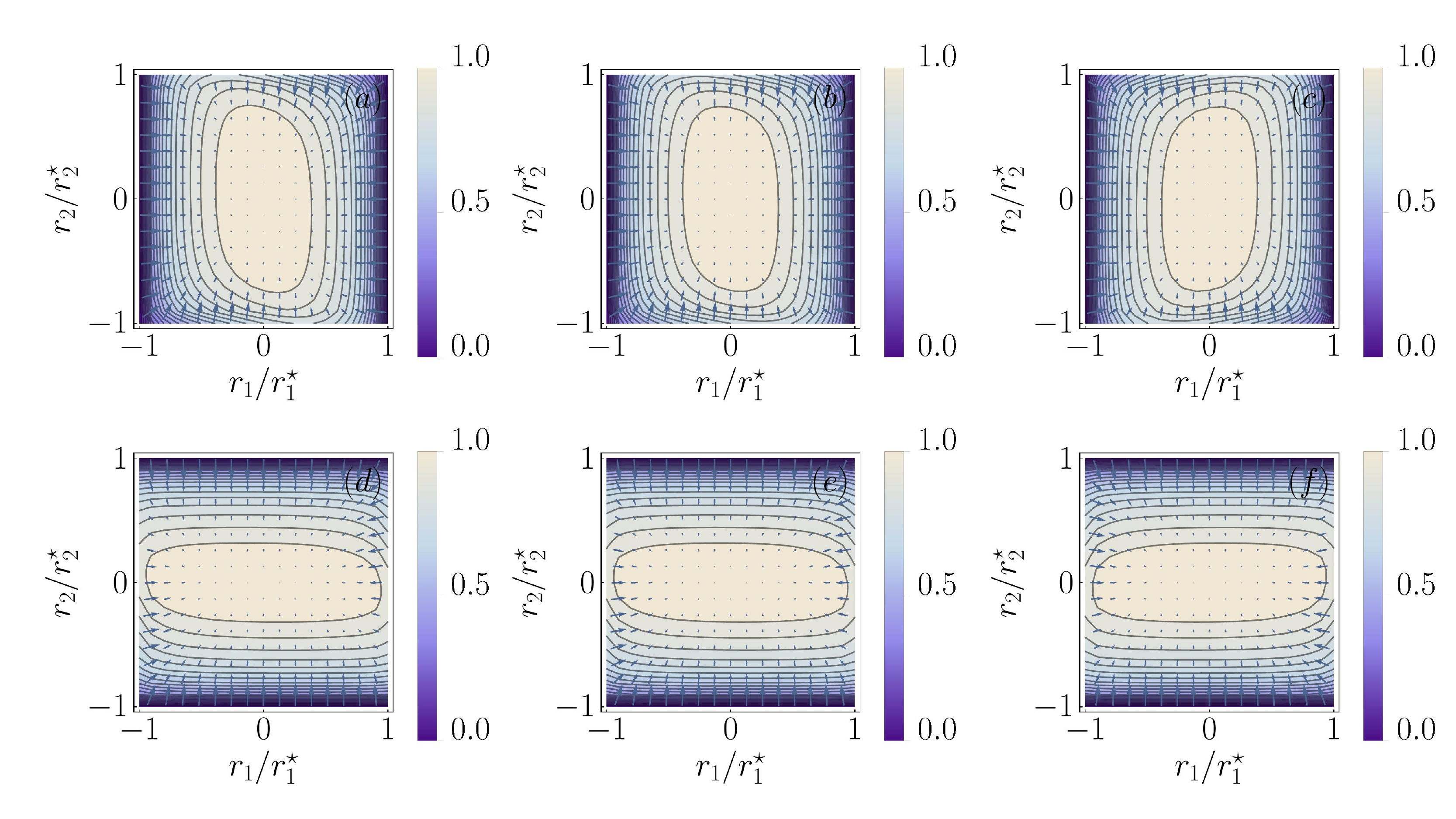}
    \caption{Two-sender EA capacities of thermal-loss BGMAC. Top: $J=\{1\}$, (a-c) $\theta=0, \pi/8$, and $3\pi/8$ respectively; bottom: $J=\{2\}$, (d-f) $\theta=0, \pi/8$, and $3\pi/8$ respectively. Normalized $I_J$ under different relative squeezer phase $\theta$ with respect to squeeze parameters $r_1$, $r_2$. $N_{{\rm S}, 1}=1,N_{{\rm S}, 2}=2,\eta_1=1/3,\eta_2=2/3, |w|^2=0.1, N_{\rm B}=0.1$. See Appendix \ref{app:evalmethod} for the definition of $r_1^\star, r_2^\star$.   \label{fig:numeric_I_loss}}

    \centering
     \includegraphics[width=0.95\textwidth]{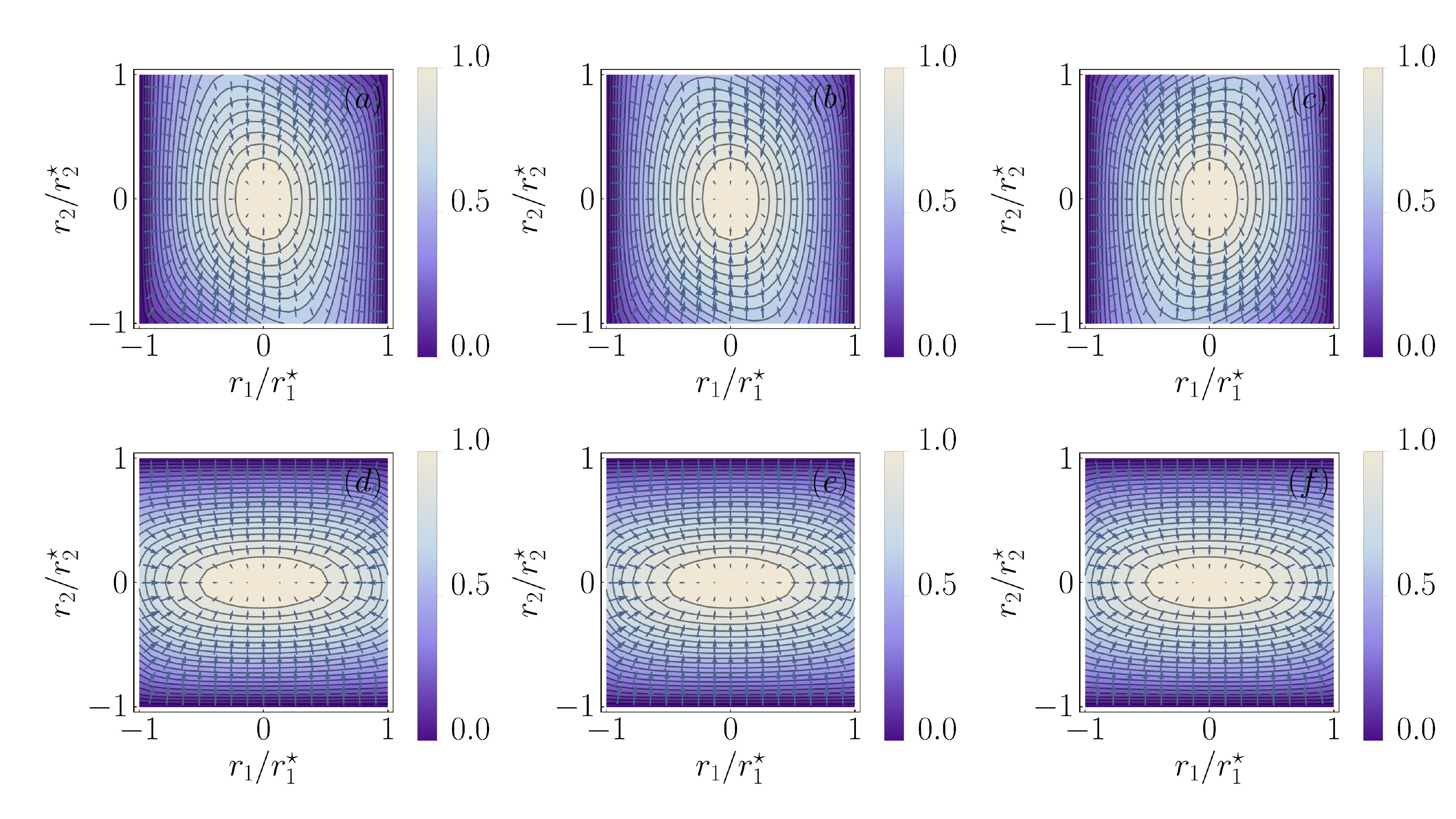}
    \caption{Two-sender EA capacities of amplifier BGMAC. Top: $J=\{1\}$, (a-c) $\theta=0, \pi/8$, and $3\pi/8$ respectively; bottom: $J=\{2\}$, (d-f) $\theta=0, \pi/8$, and $3\pi/8$ respectively. Normalized $I_J$ under different relative squeezer phase $\theta$ with respect to squeeze parameters $r_1$, $r_2$. $N_{{\rm S}, 1}=1,N_{{\rm S}, 2}=2,\eta_1=1/3,\eta_2=2/3, |\bm w|^2=1.1, N_{\rm B}=|\bm w|^2-1+0.1=0.2$. See Appendix \ref{app:evalmethod} for the definition of $r_1^\star, r_2^\star$.    \label{fig:numeric_I_amp}}
\end{figure*}
\begin{figure*}
    \centering
    \includegraphics[width=0.95\textwidth]{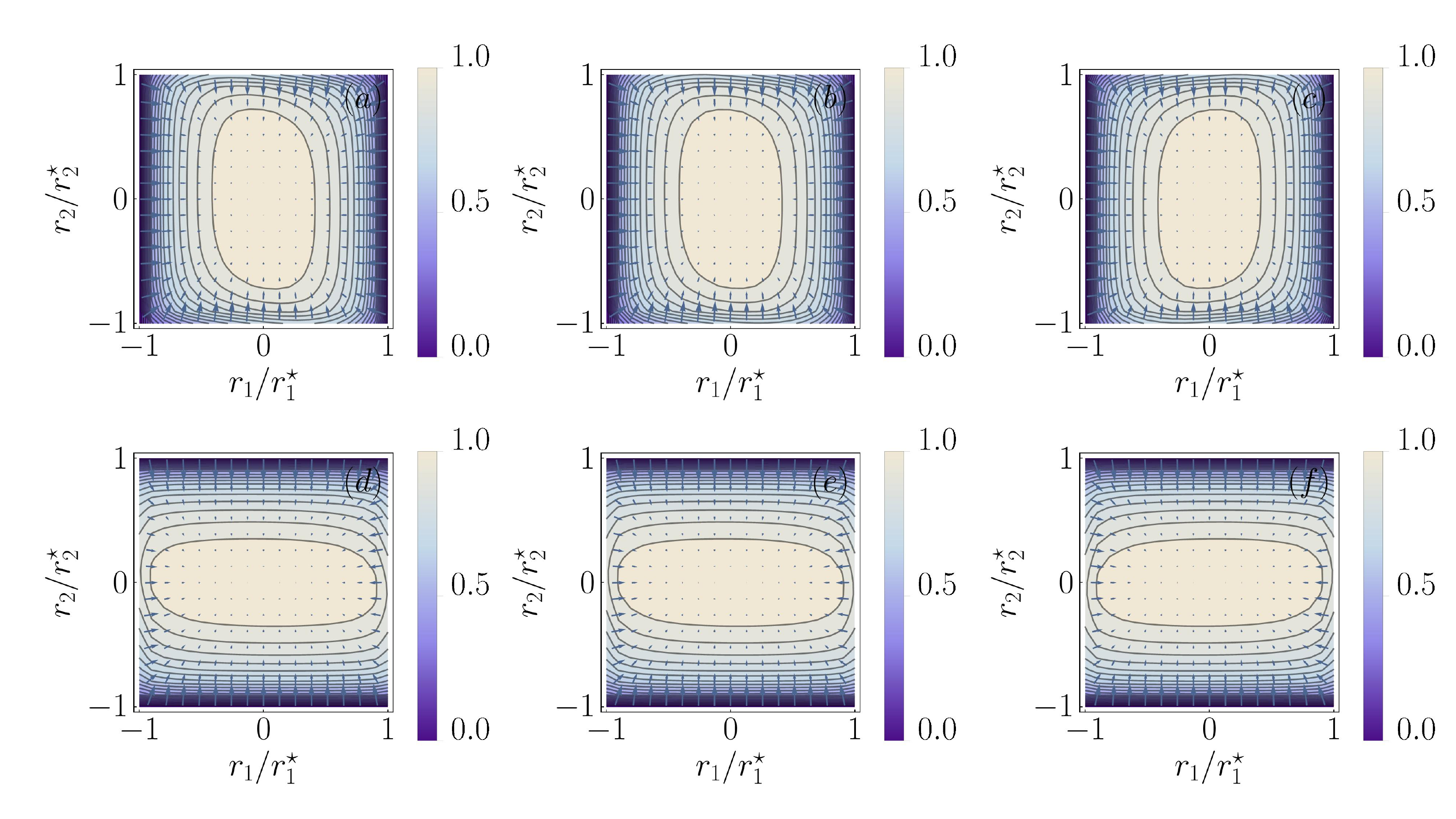}
    \caption{Two-sender EA capacities of contravariant amplifier BGMAC. Top: $J=\{1\}$, (a-c) $\theta=0, \pi/8$, and $3\pi/8$ respectively; bottom: $J=\{2\}$, (d-f) $\theta=0, \pi/8$, and $3\pi/8$ respectively. Normalized $I_J$ under different relative squeezer phase $\theta$ with respect to squeeze parameters $r_1$, $r_2$. $N_{{\rm S}, 1}=1,N_{{\rm S}, 2}=2,\eta_1=1/3,\eta_2=2/3, |\bm w|^2=0.1, N_{\rm B}=|\bm w|^2-1+0.1=0.2$. See Appendix \ref{app:evalmethod} for the definition of $r_1^\star, r_2^\star$.     \label{fig:numeric_I_conj}}
     \centering
     \includegraphics[width=0.95\textwidth]{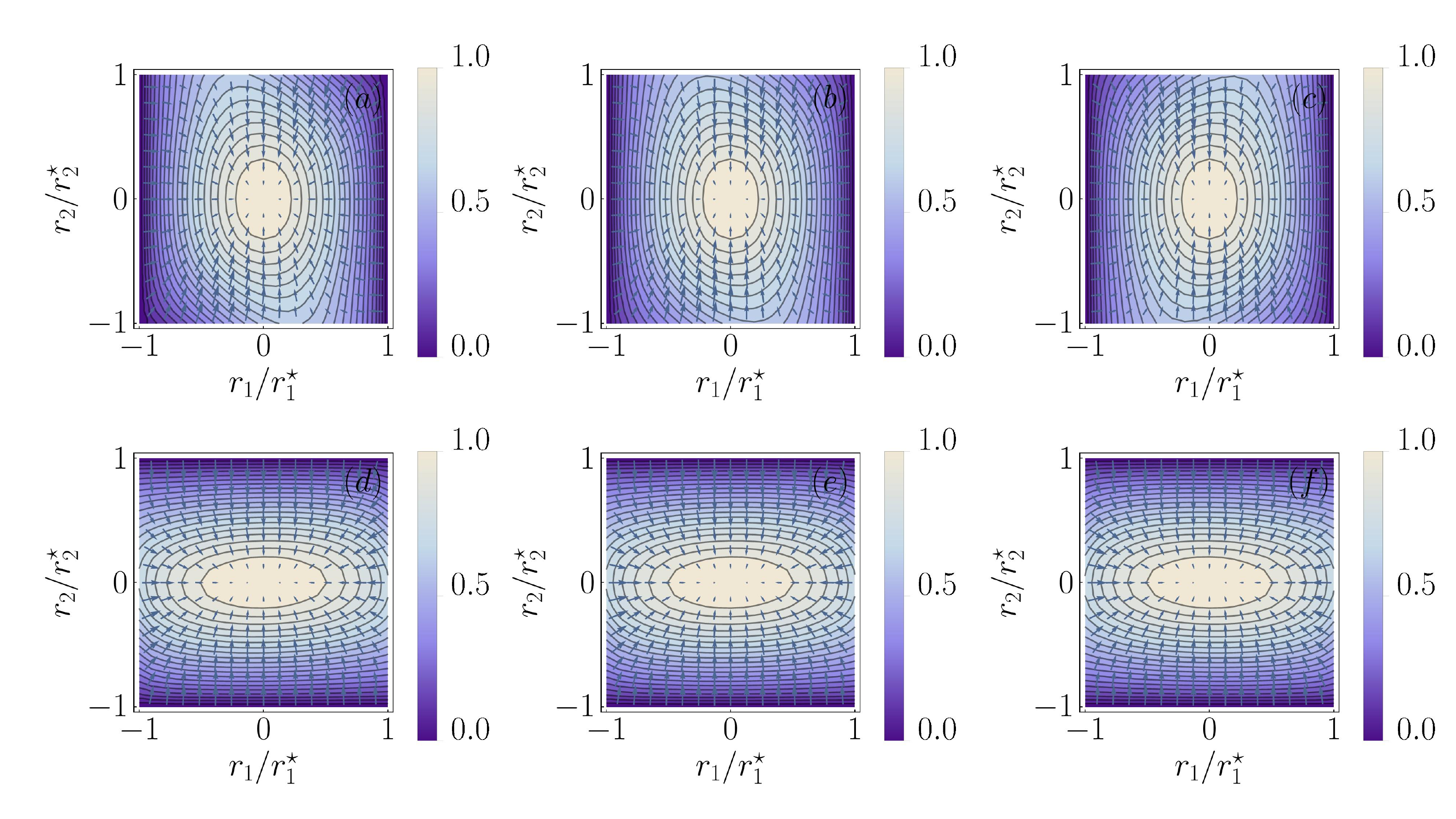}
    \caption{Two-sender EA capacities of AWGN BGMAC. Top: $J=\{1\}$, (a-c) $\theta=0, \pi/8$, and $3\pi/8$ respectively; bottom: $J=\{2\}$, (d-f) $\theta=0, \pi/8$, and $3\pi/8$ respectively. Normalized $I_J$ under different relative squeezer phase $\theta$ with respect to squeeze parameters $r_1$, $r_2$. $N_{{\rm S}, 1}=1,N_{{\rm S}, 2}=2,\eta_1=1/3,\eta_2=2/3, N_{\rm B}=0.1$. See Appendix \ref{app:evalmethod} for the definition of $r_1^\star, r_2^\star$.     \label{fig:numeric_I_AWGN}}
\end{figure*}

\end{document}